\documentclass{article}

\newif\ifdemon
\demonfalse

\usepackage{graphicx}
\usepackage{color}

\usepackage{amssymb}

\usepackage[latin1]{inputenc} 

\usepackage{amsmath}
\usepackage{stmaryrd}
\usepackage{latexsym}
\usepackage{prooftree}

\newtheorem{theorem}{Theorem}[section]
\newtheorem{proposition}[theorem]{Proposition}
\newtheorem{corollary}[theorem]{Corollary}
\newtheorem{lemma}[theorem]{Lemma}
\newenvironment{proof}{\emph{Proof.}}{\relax}
\newtheorem{definition}[theorem]{Definition}

\newcommand\qed{\hfill$\Box$\vskip0.2em}

\newcommand{\sqsupsetsim}{\mathrel{{\raisebox{0.5ex}{$\mathop{\kern0pt \sqsupset}\limits_{\textstyle\sim}$}}}}
\newcommand{\sqsubsetsim}{\mathrel{{\raisebox{0.5ex}{$\mathop{\kern0pt \sqsubset}\limits_{\textstyle\sim}$}}}}
\newcommand{\pow}{\mathbb{P}}
\newcommand\eq{\mathbin{==}}

\newcommand\fade\relax
\newcommand\proba\relax
\newcommand\angelic\relax
\newcommand\demonic\relax

\newcommand\nat{\mathbb{N}}
\newcommand\rat{\mathbb{Q}}

\newcommand\Demon{{\mathtt{D}}}
\newcommand\Nature{{\mathtt{P}}}

\newcommand\dG{{\mathsf{d}}}
\newcommand\extd[1]{{#1}^*} 
\newcommand\lfp{\mathop{\mathrm{lfp}}}

\newcommand\kwfont[1]{\pmb{\mathtt{#1}}}

\newcommand\rec{\mathop{\kwfont{rec}}}
\newcommand\produce{\mathop{\kwfont{produce}}}
\newcommand\pto[3]{{#1} \mathrel{\kwfont{to}} {#2}
  \mathrel{\kwfont{in}}{#3}}
\newcommand\thunk{\mathop{\kwfont{thunk}}}
\newcommand\force{\mathop{\kwfont{force}}}
\newcommand\retkw{\mathop{\kwfont{ret}}\nolimits}

\newcommand\dokw{\mathop{\kwfont{do}}\nolimits}
\newcommand\abort{\mathop{\kwfont{abort}}\nolimits}
\newcommand\anytype{{\overline{\sigma}}}
\newcommand\anyothertype{{\overline{\tau}}}

\newcommand\tagc[3]{[{#1}\colon {#2}]} 

\newcommand\Prob{\mathop{\text{Pr}}}

\newcommand\upc{\mathop{\uparrow}\nolimits}

\newcommand\uuarrow{\rlap{$\uparrow$}\raise.5ex\hbox{$\uparrow$}}
\newcommand\smalluuarrow{\rlap{$\scriptstyle\uparrow$}\raise.3ex\hbox{$\scriptstyle\uparrow$}}
\newcommand\ddarrow{\rlap{$\downarrow$}\raise.5ex\hbox{$\downarrow$}}

\newcommand\arrepixy{\ar@{->>}}
\newcommand\arrmonoxy{\ar@{>->}}

\newcommand\diff{\smallsetminus}

\newcommand\Eval[1]{\left\llbracket{#1}\right\rrbracket}
\newcommand{\real}{\mathbb{R}}
\newcommand\Rp{\real_+}
\newcommand{\creal}{\overline\real_+}

\newcommand{\interior}[1]{\text{int} ({#1})} 

\newcommand\Smyth{\mathcal Q}

\newcommand\Open{\mathcal O}

\newcommand\Sierp{\mathbb{S}}

\newcommand\Z{\mathbb{Z}}

\newcommand\Val{\mathbf V}

\newcommand\intT{\kwfont{int}}

\newcommand\unitT{\kwfont{unit}}

\newcommand\suc{\mathop{\kwfont{succ}}}
\newcommand\pred{\mathop{\kwfont{pred}}}

\newcommand\dom{\mathop{\mathrm{dom}}}

\newcommand\ifzkw{\kwfont{ifz}}
\newcommand\ifz[3]{\ifzkw\;#1\;#2\;#3}
\newcommand\pifkw{\kwfont{pif}}
\newcommand\pif[3]{\pifkw\;#1\;#2\;#3}
\newcommand\pifzkw{\kwfont{pifz}}
\newcommand\pifz[3]{\pifzkw\;#1\;#2\;#3}

\newcommand\pcasekw{\kwfont{pcase}}
\newcommand\pswitchkw{\kwfont{pswitch}}

\newcommand\vp{\mathbin{\downarrow}}

\newcommand\F{\kwfont{F}}
\newcommand\Vt{\kwfont{V}}
\newcommand\U{\kwfont{U}}

\newcommand\R{\mathrel{R}}
\newcommand\tr{{\triangleright}}
\newcommand\RR{\mathrel{S}}

\newcommand\CLatt{\mathbf{CLatt}}
\newcommand\Cont{\mathbf{Cont}}

\newcommand\pCCont{\mathbf{PCCont}}

\title{A Probabilistic and Non-Deterministic Call-by-Push-Value
  Language\thanks{This research was partially supported by Labex
    DigiCosme (project ANR-11-LABEX-0045-DIGICOSME) operated by ANR as
    part of the program ``Investissement d'Avenir'' Idex Paris-Saclay
    (ANR-11-IDEX-0003-02).}}

\author{Jean
  Goubault-Larrecq\thanks{LSV, ENS Paris-Saclay, CNRS, Universit\'e
    Paris-Saclay.  Address:  ENS Paris-Saclay, 61 avenue du président Wilson,
    94230 Cachan, France.  Email: \texttt{goubault@lsv.fr}}}

\begin{document}






\maketitle

\begin{abstract}
  There is no known way of giving a domain-theoretic semantics to
  higher-order probabilistic languages, in such a way that the
  involved domains are continuous or quasi-continuous---the latter is
  required to do any serious mathematics.  We argue that the problem
  naturally disappears for languages with two kinds of types, where
  one kind is interpreted in a Cartesian-closed category of continuous
  dcpos, and the other is interpreted in a category that is closed
  under the probabilistic powerdomain functor.  Such a setting is
  provided by Paul B. Levy's call-by-push-value paradigm.  Following
  this insight, we define a call-by-push-value language, with
  probabilistic choice sitting inside the value types, and where
  conversion from a value type to a computation type involves demonic
  non-determinism.  We give both a domain-theoretic semantics and an
  operational semantics for the resulting language, and we show that
  they are sound and adequate.  With the addition of statistical
  termination testers and parallel if, we show that the language is
  even fully abstract---and those two primitives are required for
  that.

\end{abstract}





\section{Introduction}
\label{sec:intro}

A central problem of domain theory is the following: is there any full
Cartesian-closed subcategory of the category $\Cont$ of continuous
dcpos that is closed under the probabilistic powerdomain functor
$\Val_{\leq 1}$ \cite{JT:troublesome}?  Solving the question in the
positive would allow for a simple semantics of probabilistic
higher-order languages, where types are interpreted as certain
continuous dcpos.

However, we have a conundrum here.  The category $\Cont$ itself is
closed under $\Val_{\leq 1}$ \cite{Jones:proba}, but is not
Cartesian-closed \cite[Exercise~3.3.12(11)]{AJ:domains}.  Among the
Cartesian-closed categories of continuous domains, none is known to be
closed under $\Val_{\leq 1}$, and most, such as the category of
bc-domains or the category $\CLatt$ of continuous complete lattices,
definitely are not \cite{JT:troublesome}.

Instead of solving this problem, one may wonder whether there are
other kinds of domain-theoretic semantics that would be free of the
issue.  Typically, can we imagine having \emph{two} classes of types?
One would be interpreted in a category of continuous dcpos that is
closed under $\Val_{\leq 1}$---$\Cont$ for example, although we will
prefer the category $\pCCont$ of pointed coherent continuous dcpos
(see below).  The other would be interpreted in a Cartesian-closed
category of continuous dcpos, and we will use $\CLatt$.  Such a
division in two classes of types is already present in Paul B. Levy's
\emph{call-by-push-value} \cite{Levy:CBPV} (a.k.a.\ \emph{CBPV}), and
although the division is justified there as to be between \emph{value}
types and \emph{computation} types, the formal structure will be
entirely similar.

\paragraph{Outline}

We briefly review some related work in Section~\ref{sec:related-work},
and give a few basic working definitions in
Section~\ref{sec:preliminaries}.  We define our probabilistic
call-by-push-value languages in Section~\ref{sec:language}, explaining
the design decisions we had to make in the process---notably the extra
need for demonic non-determinism.  We give domain-theoretic and
operational semantics there, too.  We establish soundness in
Section~\ref{sec:soundness} and adequacy in
Section~\ref{sec:adequacy}, to the effect that for every ground term
$M$ of the specific type $\F\Vt\unitT$, the probability $\Prob (M \vp)$
that $M$ must terminate, as defined from the operational semantics,
coincides with a similar notion of probability defined from the
denotational semantics.  In Section~\ref{sec:cons-adeq}, we review a
few useful consequences of adequacy, among which the coincidence
between the applicative preorder $\precsim^{app}_\tau$ and the
contextual preorder $\precsim_\tau$ (both will be defined there), a
fact sometimes called Milner's Context Lemma in the context of PCF
(see \cite[Theorem~8.1]{Streicher:pcf}).  We show that, among the
languages we have defined, CBPV$(\Demon, \Nature)$ is not
(inequationally) fully abstract in Section~\ref{sec:fail-full-abstr},
and that adding a parallel if operator $\pifzkw$ does not make it
fully abstract, but that adding both $\pifzkw$ and a statistical
termination tester operator $\bigcirc_{> b}$ (as in \cite{jgl-jlap14})
results in an (inequationally) fully abstract language.  The latter is
proved in Section~\ref{sec:full-abstraction}.  We conclude and list a
few remaining open questions in Section~\ref{sec:concl-open-probl}.

\paragraph{Acknowledgments}

I wish to thank Zhenchao Lyu and Xiaodong Jia, who participated in
many discussions on the theme of this paper; Ohad Kammar, who
kindly pointed me to \cite{VKS:stat:dom}; and the anonymous referees
of the LICS'19 conference.

\section{Related Work}
\label{sec:related-work}

Call-by-push-value (CBPV) is the creation of Paul B. Levy \cite{Levy:CBPV}
(see also the book \cite{Levy:CBPV:book}), and is a typed higher-order
pure functional language.  It was originally meant as a subsuming
paradigm, embodying both call-by-value and call-by-name disciplines.

The first probabilistic extension of CBPV was proposed recently by
Ehrhard and Tasson \cite{ET:CBPV:prob}, and its denotational semantics
rests on probabilistic coherence spaces.  Their typing discipline is
inspired by linear logic, and they also include a treatment of general
recursive types, which we will not.  In contrast, our extension of
CBPV will have first-class types of subprobability distributions $\Vt{}
\sigma$, and will also include a type former for demonic
non-determinism (a.k.a., must-non-determinism).

Statistical probabilistic programming has attracted quite some
attention recently, and quasi-Borel spaces and predomains have
recently been used to give adequate semantics to typed and untyped
probabilistic programming languages, see \cite{VKS:stat:dom}.  The
latter describes another way of circumventing the problem we stated in
the introduction.  One important point that V\'ak\'ar, Kammar and
Staton achieve is the commutativity of the probabilistic choice monad,
at all, even higher-order, types.  In standard domain theory, the
$\Val_{\leq 1}$ monad is known to be commutative in full subcategories
of $\Cont$ only.  That would be enough motivation to attempt to solve
the problem stated in the introduction, of finding a Cartesian-closed
category, closed under $\Val_{\leq 1}$ \cite{JT:troublesome}.  We
also implement a commutative $\Val_{\leq 1}$ monad in a higher-order
setting; our way of circumventing the problem is merely different.

There is a large body of literature concerned with the question of
full abstraction for PCF-like languages.  The first paper on the
subject is due to G. Plotkin \cite{Plotkin:PCF}, who defined the
language PCF, asked all the important questions (soundness, adequacy,
full abstraction, definability), and answered all of them, except for
the question of finding a fully abstract denotational model of PCF
\emph{without} parallel if, a question that was solved later, through
game semantics notably \cite{HO:games,AJM:games}.  Th. Streicher's book
\cite{Streicher:pcf} is an excellent reference on the subject.

Probabilistic coherence spaces provide a fully abstract semantics for
a version of PCF with probabilistic choice, as shown by Ehrhard,
Tasson, and Pagani \cite{ETP:probcoh:fa}.  The already cited paper of
Ehrhard and Tasson \cite{ET:CBPV:prob} gives an analogous result for
their probabilistic version of CBPV.  Our work is concerned with
languages with domain-theoretic semantics instead, and our former work
\cite{jgl-jlap14} gives soundness, adequacy and full abstraction
results for PCF plus angelic non-determinism, and for PCF plus
probabilistic choice and angelic non-determinism plus so-called
statistical termination testers.  We will see that CBPV naturally
calls for a form of demonic, rather than angelic, non-determinism.

\section{Preliminaries}
\label{sec:preliminaries}

We refer to \cite{GHKLMS:contlatt,AJ:domains,JGL-topology} for
material on domain theory and topology.  A dcpo is \emph{pointed} if
and only if it has a least element $\bot$.  Dcpos are always equipped
with their Scott topology.  $\creal = \real \cup \{\infty\}$ and
$[0, 1]$ are dcpos, with the usual ordering.  The \emph{way-below}
relation is written $\ll$: $x \ll y$ if and only if for every directed
family ${(x_i)}_{i \in I}$ such that $y \leq \sup_{i \in I} x_i$,
there is an $i \in I$ such that $x \leq x_i$.  A dcpo $X$ is
\emph{continuous} if and only if every element is the supremum of a
directed family of elements way-below it.  In that case, the sets
$\uuarrow x = \{y \in X \mid x \ll y\}$ form a base of open sets of
the Scott topology.  We recall that a \emph{base} of a topology is a
family $\mathcal B$ of open sets such that every open set is a union
of sets from $\mathcal B$.  A \emph{subbase} is a family $\mathcal S$
such that the finite intersections of elements of $\mathcal S$ form a
base.

A \emph{basis} $B$ of a dcpo $X$ (not to be confused with a base) is a
set of elements of $X$ such that, for every $x \in X$,
$\{b \in B \mid b \ll x\}$ is directed and has $x$ as supremum.  A
dcpo is continuous if and only if it has a basis.  Then the sets
$\uuarrow b$, $b \in B$, also form a base of the Scott topology.

We write $\leq$ for the specialization ordering of a $T_0$ topological
space.  For a dcpo $X$, that is the original ordering on $X$.  A
subset of a topological space is \emph{saturated} if and only if it is
upwards-closed in $\leq$, if and only if it is the intersection of its
open neighborhoods.  A topological space $X$ is locally compact if and
only if for every $x \in X$, for every open neighborhood $U$ of $x$,
there is a compact saturated set $Q$ such that
$x \in \interior Q \subseteq Q \subseteq U$.  ($\interior Q$ denotes
the interior of $Q$.)  In that case, for every compact saturated
subset $Q$ and every open neighborhood $U$ of $Q$, there is a compact
saturated set $Q'$ such that
$Q \subseteq \interior {Q'} \subseteq Q' \subseteq U$.  A topological
space is \emph{coherent} if and only if the intersection of any two
compact saturated subsets is compact.  It is \emph{well-filtered} if
and only if for every filtered family of compact saturated sets
${(Q_i)}_{i \in I}$ (\emph{filtered} meaning directed for reverse
inclusion), every open neighborhood $U$ of $\bigcap_{i \in I} Q_i$
already contains some $Q_i$.  In a well-filtered space, the
intersection $\bigcap_{i \in I} Q_i$ of such a filtered family is
compact saturated.  A \emph{stably compact} space is a $T_0$,
well-filtered, locally compact, coherent and compact space $X$.  Then
the complements of compact saturated sets form another topology on
$X$, the \emph{cocompact topology}, and $X$ with the cocompact
topology is the \emph{de Groot dual} $X^\dG$ of $X$.  For every stably
compact space, $X^{\dG\dG} = X$.  Every pointed, coherent, continuous
dcpo is stably compact.

Given two dcpos $X$ and $Y$, $[X \to Y]$ denotes the dcpo of all
Scott-continuous maps from $X$ to $Y$, ordered pointwise.  Directed
suprema are also pointwise, namely
$(\sup_{i \in I} f_i) (x) = \sup_{i \in I} (f_i (x))$ for every
directed family ${(f_i)}_{i \in I}$ in $[X \to Y]$.

\section{The Languages CBPV$(\Demon,\Nature)$ and CBPV$(\Demon,\Nature)+\pifzkw+\bigcirc$}
\label{sec:language}

The first language we introduce is called CBPV$(\Demon,\Nature)$: it
is a call-by-push-value language with $\Demon$emonic non-determinism
and $\Nature$robabilistic choice.  We will explain below why we do not
consider just probabilistic choice, but also demonic non-determinism.

\subsection{Types and their Semantics}
\label{sec:types-their-semant}

We consider the following grammar of types:
\begin{align*}
  \sigma, \tau, \ldots & ::= \U\underline\tau \mid \unitT \mid \intT \mid \sigma
                         \times \tau \mid \Vt{} \tau \\
  \underline\sigma, \underline\tau, \ldots & ::= \F{} \tau \mid \sigma
                                             \to \underline\tau.
\end{align*}
The types $\sigma$, $\tau$, \ldots, are the \emph{value} types, and
the types $\underline\sigma$, $\underline\tau$, \ldots, are the
\emph{computation} types, following Levy \cite{Levy:CBPV}.  Our types
differ from Levy's: we do not have countable sums in value types or
countable products in computation types, we write $\unitT$ instead of
$1$, and we have a primitive type $\intT$ of integers; the main
difference is the $\Vt{} \tau$ construction, denoting the type of
subprobability valuations on the space of elements of type
$\tau$.

We write $\anytype$, $\anyothertype$ for types when it is not
important whether they are value types or computation types.

We have already said in the introduction that computation types will
be interpreted in the category $\CLatt$ of continuous complete
lattices.  Value types $\tau$ will give rise to pointed, coherent,
continuous dcpos $\Eval \tau$:
\begin{itemize}
\item for every computation type $\underline\tau$, we will define
  $\Eval {\U\underline\tau}$ as $\Eval {\underline\tau}$: being a
  continuous complete lattice, it is in particular pointed, coherent,
  and a continuous dcpo;
\item $\Eval {\unitT}$ will be \emph{Sierpi\'nski space} $\Sierp = \{\bot, \top\}$
  with $\bot < \top$;
\item $\Eval {\intT}$ will be $\Z_\bot = \Z \cup \{\bot\}$, with the
  ordering that makes $\bot$ least and all integers be pairwise
  incomparable;
\item $\Eval {\Vt\tau}$ will be $\Val_{\leq 1} (\Eval \tau)$, where
  $\Val_{\leq 1} X$ denotes the dcpo of all subprobability valuations
  on the space $X$.
\end{itemize}
A \emph{subprobability valuation} on $X$ is a map $\nu$ from the
lattice $\Open X$ of open subsets of $X$ to $[0, 1]$ which is strict
($\nu (\emptyset)=0$), Scott-continuous, and modular
($\nu (U \cup V) + \nu (U \cap V) = \nu (U) + \nu (V)$).
When $X$ is a continuous dcpo, so is $\Val_{\leq 1} X$
\cite[Corollary~5.4]{Jones:proba}.  It is pointed, since the zero
valuation is least in $\Val_{\leq 1} X$.  If $X$ is also coherent,
then $\Val_{\leq 1} X$ is stably compact, see below.  Hence
$\Eval {\Vt\tau}$ is indeed a pointed, coherent continuous dcpo.
  
The fact that $\Val_{\leq 1} X$ is stably compact for every coherent
continuous dcpo $X$ is folklore.  We argue as follows.  The lift
$X_\bot$ of $X$, obtained by adding a fresh bottom element $\bot$ to
$X$, is stably compact.  Then the space $\Val_1 X_\bot$ of all
probability valuations $\nu$, i.e., such that $\nu (X_\bot)=1$, is
stably compact in the weak upwards topology
\cite[Theorem~39]{AMJK:scs:prob}.  The latter has a subbase of open
sets of the form $[U > r] = \{\nu \mid \nu (U) > r\}$, for every open
subset $U$ of $X$ and $r \in \Rp \diff \{0\}$.  The restriction map
$\nu \mapsto \nu_{|\Open X}$ is a homeomorphism from $\Val_1 X_\bot$
onto $\Val_{\leq 1} X$, both with their weak upwards topology, with
inverse $\nu \mapsto \nu + (1-\nu (X))\delta_\bot$.  Hence
$\Val_{\leq 1} X$ is stably compact in its weak upwards topology.
Since $X$ is continuous, the latter coincides with the Scott topology,
as shown by \cite[Satz~8.6]{Kirch:bewertung}, see also
\cite[Satz~4.10]{Tix:bewertung}.

It might seem curious that probabilistic non-determinism arises, as
$\Vt{} \sigma$, among the \emph{value} types.  I have no philosophical
backing for that, but this is somehow forced upon us by the
mathematics.

Similarly, computation types $\underline\tau$ will give rise to
continuous complete lattices $\Eval {\underline\tau}$---notably
$\Eval {\sigma \to \underline\tau}$ will be the continuous complete
lattice $[\Eval \sigma \to \Eval {\underline\tau}]$ of all
Scott-continuous maps $[\Eval \sigma \to \Eval {\underline\tau}]$ from
$\Eval \sigma$ to $\Eval {\underline\tau}$---, but we have to decide
on an interpretation of types of the form $\F{}\tau$.

If we had decided to interpret computation types as bc-domains instead
of continuous complete lattices, then a natural choice would be to
define $\Eval {\F{}\tau}$ as Ershov's \emph{bc-hull} of $\Eval\tau$
\cite{Ershov:bchull}.  (Bc-domains are, roughly speaking, continuous
complete lattices that may lack a top element.)  As Ershov notices,
``the construction of a bc-hull in the general case is highly
nonconstructive (using a Zorn's lemma)'' (ibid., page~13).  Fortunately,
the bc-hull of a space $X$ is a natural subspace of the \emph{Smyth
  powerdomain} $\Smyth (X)$ of $X$, at least when $X$ is a coherent
algebraic dcpo (ibid., Corollary~B), and $\Smyth (X)$ is easier to
work with.  Explicitly, $\Smyth (X)$ is the poset of all non-empty
compact saturated subsets of $X$, ordered by reverse inclusion, and is
used to interpret demonic non-determinism in denotational semantics.
When $X$ is well-filtered and locally compact, $\Smyth (X)$ is also a
continuous dcpo, and it is a bc-domain provided $X$ is also compact
and coherent.  We shall see below that $\Smyth^\top (X)$, the poset of
all (possibly empty) compact saturated subsets of $X$---alternatively,
$\Smyth (X)$ plus an additional top element $\top = \emptyset$---, is
a continuous complete lattice whenever $X$ is a stably compact space,
and that would make $\Smyth^\top (\Eval \tau)$ a good candidate for
$\Eval {\F{} \tau}$.

For technical reasons related to adequacy, we will need a certain map
$\extd f$ below to be strict, i.e., to map $\bot$ to $\bot$.
(Technically, this is needed so that the denotational semantics of the
construction $\pto M {x_\sigma} N$, to be introduced below, be strict
in that of $M$, in order to validate the fact that
$\pto M {x_\sigma} N$ loops forever if $M$ does.)
This will be obtained by defining $\Eval {\F{}\tau}$
as $\Smyth^\top_\bot (\Eval \tau)$ instead, where
$\Smyth^\top_\bot (X)$ is the \emph{lift} of $\Smyth^\top (X)$,
obtained by adding a fresh element $\bot$ below all others.

We recapitulate:
\begin{itemize}
\item $\Eval {\sigma \to \underline\tau} = [\Eval \sigma \to \Eval
  {\underline\tau}]$;
\item $\Eval {\F{} \sigma} = \Smyth^\top_\bot (\Eval \sigma)$.
\end{itemize}

Let us check that $\Smyth^\top_\bot (\Eval \sigma)$ has the required
property of being a continuous complete lattice, and let us prove some
additional properties that we will need later.  We start with the
similar properties of $\Smyth^\top (\Eval \sigma)$.  We let
$\eta^\Smyth \colon X \to \Smyth^\top (X)$ map every $x$ to $\upc x$.
\begin{proposition}
  \label{prop:Qtop:1}
  Let $X$ be a stably compact space.  Then:
  \begin{enumerate}
  \item $\Smyth^\top (X)$ is a continuous complete lattice, and $Q$ is
    way-below $Q'$ if and only if $Q' \subseteq \interior Q$;
  \item For every continuous complete lattice $L$, for every
    continuous map $f \colon X \to L$, there is a Scott-continuous map
    $\extd f \colon \Smyth^\top (X) \to L$ such that
    $\extd f \circ \eta^\Smyth = f$, and it is defined by
    $\extd f (Q) = \bigwedge_{x \in Q} f (x)$.
  \item $\extd f (\emptyset) = \top$, $\extd f (Q_1 \cup Q_2) = \extd
    f (Q_1) \wedge \extd f (Q_2)$.
  \end{enumerate}
\end{proposition}
\begin{proof}
    1. This is well-known, but here is a brief argument.  The elements
  of $\Smyth^\top (X)$ are exactly the closed subsets in the de Groot
  dual of $X$, and the closed sets of any topological space always
  form a complete lattice.  Note that the supremum of an arbitrary
  family ${(Q_i)}_{i \in I}$ in $\Smyth^\top (X)$ is
  $\bigcap_{i \in I} Q_i$.

  Given any compact saturated subset $Q'$ of $X$, the family $N (Q')$
  of compact saturated neighborhoods $Q''$ of $Q'$ is filtered, and
  has $Q'$ as intersection.  Indeed, since $Q'$ is saturated, it is
  the intersection of its open neighborhoods; for every open
  neighborhood $U$ of $Q'$, local compactness implies that there is a
  compact saturated set $Q''$ such that
  $Q' \subseteq \interior {Q''} \subseteq Q'' \subseteq U$; applying
  this to $U=X$ shows that $N (Q')$ is non-empty, and given
  $Q_1, Q_2 \in N (Q')$, applying it to
  $U = \interior {Q_1} \cap \interior {Q_2}$, shows that $N (Q')$
  contains an element included in both $Q_1$ and $Q_2$.

  It follows that, if $Q \ll Q'$, then $Q$ contains an element of
  $N (Q')$, hence in particular an open neighborhood of $Q'$.
  Conversely, if $Q \supseteq U \supseteq Q'$ where $U$ is open, then
  for every directed family ${(Q_i)}_{i \in I}$ in $\Smyth^\top (X)$
  such that $Q' \supseteq \bigcap_{i \in I} Q_i$, $U$ contains some
  $Q_i$ by well-filteredness, hence $Q \supseteq Q_i$.  Therefore
  $Q \ll Q'$.

  Finally, since every $Q'$ in $\Smyth^\top (X)$ is the filtered
  intersection of the elements of $N (Q')$, it is the supremum of the
  directed family $N (Q')$, and we have just argued that every element
  of $N (Q')$ is way-below $Q'$, showing that $\Smyth^\top (X)$ is
  continuous.

  2. We define $\extd f (Q)$ as $\bigwedge_{x \in Q} f (x)$.  This
  satisfies $\extd f \circ \eta^\Smyth = f$, and is monotonic.
  Note that this is defined even when $Q$ is empty, in which case
  $\extd f (Q)$ is the top element of $L$.  In order to show that
  $\extd f$ is Scott-continuous, let ${(Q_i)}_{i \in I}$ be a
  directed family in $\Smyth^\top (X)$, and
  $Q = \bigcap_{i \in I} Q_i$.  We wish to show that
  $\extd f (Q) \leq \sup_{i \in I} \extd f (Q_i)$; the converse
  inclusion is by monotonicity.  To this end, we let $y$ be an element
  of $L$ way-below $\extd f (Q)$.  Since $y \ll f (x)$ for every
  $x \in Q$, every element of $Q$ is in the open set
  $f^{-1} (\uuarrow y)$.  Then $Q = \bigcap_{i \in I} Q_i$ is included
  in $f^{-1} (\uuarrow y)$, so by well-filteredness some $Q_i$ is also
  included in $f^{-1} (\uuarrow y)$.  Then $y \ll f (x)$ for every
  $x \in Q_i$, so
  $y \leq \bigwedge_{x \in Q_i} f (x) = \extd f (Q_i)$.  Since that
  holds for every $y \ll \extd f (Q)$, the desired inequality
  follows.

  3. Easy check.
%
  \qed
\end{proof}
Note that item~2 does not state that $\extd f$ is \emph{unique}; we
have just chosen the largest one.  A similar construction is
well-known for $\Smyth (X)$.
Proposition~\ref{prop:Qtop:1} establishes the essential properties
needed to show that $\Smyth^\top$ defines a monad on the category of
stably compact spaces, and that is not only well-known, but we will
not require as much.

We turn to $\Smyth^\top_\bot (\Eval \tau)$.  We again write
$\eta^\Smyth$ for the function that maps $x$ to $\upc x$, this time
from $X$ to $\Smyth^\top_\bot (\Eval X)$.  Below, we again write
$\extd f$ for the extension of $f$ to $\Smyth^\top_\bot (X)$.  This
should not cause any confusion with the map $\extd f$ of
Proposition~\ref{prop:Qtop:1}, since the two maps coincide on
$\Smyth^\top (X)$.  Note that $\extd f$ is now strict.
\begin{proposition}
  \label{prop:Qtop}
  Let $X$ be a stably compact space.  Then:
  \begin{enumerate}
  \item $\Smyth^\top_\bot (X)$ is a continuous complete lattice, and
    $Q$ is way-below $Q'$ if and only if $Q=\bot$, or $Q, Q' \neq
    \bot$ and $Q' \subseteq \interior Q$;
  \item For every continuous complete lattice $L$, for every
    continuous map $f \colon X \to L$, there is a strict
    Scott-continuous map $\extd f \colon \Smyth^\top_\bot (X) \to L$
    such that $\extd f \circ \eta^\Smyth = f$.  This is defined by
    $\extd f (\bot) = \bot$, and for every $Q \neq \bot$,
    $\extd f (Q) = \bigwedge_{x \in Q} f (x)$.
  \item $\extd f (\emptyset) = \top$, $\extd f (Q_1 \wedge Q_2) = \extd
    f (Q_1) \wedge \extd f (Q_2)$.
  \item For every stably compact space $Y$, for every Scott-continuous
    map $f \colon X \to \Smyth^\top_\bot (Y)$, and for every
    Scott-continuous map $g$ from $Y$ to a continuous complete lattice
    $L$, $\extd g \circ \extd f = \extd {(\extd g \circ f)}$.
  \end{enumerate}
\end{proposition}
\begin{proof}
  1. The lift of a continuous complete lattice is a continuous
  complete lattice, and $\bot$ is always way-below every element.

  2. Easy.

  3. We check the second inequality.  That follows from
  Proposition~\ref{prop:Qtop:1}, item~3 if $Q_1, Q_2 \neq \bot$.  If
  $Q_1 = \bot$, then $\extd f (Q_1 \wedge Q_2) = \extd f (\bot) =
  \bot$ and $\extd f (Q_1) \wedge \extd f (Q_2) = \bot \wedge \extd f
  (Q_2) = \bot$.  Similarly if $Q_2=\bot$.

  4. Fix $Q \in \Smyth^\top_\bot (X)$.  If $Q = \bot$, then $\extd g
  (\extd f (Q)) = \bot = \extd {(\extd g \circ f)} (Q)$ by strictness.
  Henceforth, we assume that $Q \neq \bot$.

  If $f (x) = \bot$ for some $x \in Q$, then $\extd f (Q) = \bot$, so
  $\extd g (\extd f (Q)) = \bot$, and
  $\extd {(\extd g \circ f)} (Q) = \bigwedge_{x \in Q} \extd g (f (x))
  \leq \extd g (\bot) = \bot$, since $\extd g$ is strict.  Henceforth,
  we assume that $f (x) \neq \bot$ for every $x \in Q$.
  
  We claim that $\bigcup_{x \in Q} f (x)$ is compact.  Let
  ${(V_i)}_{i \in I}$ be a directed family of open subsets of $Y$
  whose union contains $\bigcup_{x \in Q} f (x)$.  For every
  $x \in Q$, $f (x)$ is compact and included in
  $\bigcup_{i \in I} V_i$, so $f (x) \subseteq V_i$ for some
  $i \in I$.  Hence $Q \subseteq \bigcup_{i \in I} {f}^{-1} (V_i)$.
  Since $Q$ is compact, $Q \subseteq {f}^{-1} (V_i)$ for some
  $i \in I$, whence $\bigcup_{x \in Q} f (x) \subseteq V_i$.

  $\bigcup_{x \in Q} f (x)$ is also saturated in $Y$, hence an element
  of $\Smyth^\top (Y)$, and therefore also of $\Smyth^\top_\bot (Y)$.
  It follows that this is the infimum of the elements $f (x)$,
  $x \in Q$, hence is equal to $\extd f (Q)$.  Therefore
  $\extd g (\extd f (Q)) = \bigwedge_{y \in \extd f (Q)} g (y) =
  \bigwedge_{x \in Q, y \in f (x)} g (y) = \bigwedge_{x \in Q} \extd g
  (f (x)) = \extd {(\extd g \circ f)} (Q)$.
  \qed
\end{proof}
Item~4 above is part of the properties needed to check that
$\Smyth^\top_\bot$ defines a monad on $\CLatt$.  We will not expand on
that.

\subsection{Syntax}
\label{sec:syntax}

We define the syntax of our language CBPV$(\Demon,\Nature)$ together
with its typing discipline, inductively, as in
Figure~\ref{fig:syntax}, using the notation $M \colon \overline\tau$
to say ``$M$ is a term of type $\overline\tau$''.  There are countably
infinitely many variables $x_\tau$, $y_\tau$, \ldots of each value
type $\tau$.
\begin{figure}
  \centering\footnotesize
  \[
    \begin{array}{cc}
      \multicolumn{2}{c}{
      \begin{prooftree}
        \justifies
        x_\tau \colon \tau
      \end{prooftree}
      \quad
      \begin{prooftree}
        \justifies
        \underline* \colon \unitT
      \end{prooftree}
      \quad
      \begin{prooftree}
        \justifies
        \underline n \colon \intT
        \using (n \in \Z)
      \end{prooftree}
      \quad
      \begin{prooftree}
        \justifies
        \abort_{\F{} \tau} \colon \F{} \tau
      \end{prooftree}
      }
      \\ \\
      \multicolumn{2}{c}{
      \begin{prooftree}
        M \colon \underline\tau \justifies \lambda x_\sigma . M \colon
        \sigma \to \underline \tau
      \end{prooftree}
      \quad
      \begin{prooftree}
        M \colon \sigma \to \underline\tau \quad N \colon \sigma
        \justifies MN \colon \underline\tau
      \end{prooftree}
      \quad
      \begin{prooftree}
        M \colon \sigma \justifies {} \rec x_\sigma . M \colon
        \sigma
      \end{prooftree}
      }
      \\ \\
      \begin{prooftree}
        M \colon \intT \justifies \suc M \colon \intT
      \end{prooftree}
      \quad
      \begin{prooftree}
        M \colon \intT \justifies \pred M \colon \intT
      \end{prooftree}
      &
        \begin{prooftree}
        M \colon \underline\tau \justifies \thunk M \colon \U{}
        \underline\tau
      \end{prooftree}
        \quad
        \begin{prooftree}
          M \colon \U\underline\tau \justifies \force M \colon
          \underline\tau
        \end{prooftree}
      \\ \\
      \begin{prooftree}
        M \colon \unitT \quad N \colon \anytype \justifies M; N
        \colon \anytype
      \end{prooftree}
      &
        \begin{prooftree}
          M \colon \intT \quad N \colon \anytype \quad P \colon \anytype
          \justifies \ifz M N P \colon \anytype
        \end{prooftree}
      \\ \\
      \begin{prooftree}
        M \colon \sigma \times \tau \justifies \pi_1 M \colon \sigma
      \end{prooftree}
      \quad
      \begin{prooftree}
        M \colon \sigma \times \tau \justifies \pi_2 M \colon \tau
      \end{prooftree}
      &
        \begin{prooftree}
          M \colon \sigma \quad N \colon \tau \justifies \langle M, N
          \rangle \colon \sigma \times \tau
        \end{prooftree}
      \\ \\
      \begin{prooftree}
        M \colon \Vt{} \tau \quad N \colon \Vt{} \tau \justifies M \oplus N
        \colon \Vt{} \tau
      \end{prooftree}
      \quad
      \begin{prooftree}
        M \colon \tau \justifies \retkw M \colon \Vt{} \tau
      \end{prooftree}
      &
        \begin{prooftree}
          M \colon \Vt{} \sigma \quad N \colon \Vt{} \tau \justifies \dokw
          {x_\sigma \leftarrow M}; N \colon \Vt{} \tau
        \end{prooftree}
      \\ \\
      \begin{prooftree}
        M \colon \F{} \tau \quad N \colon \F{} \tau \justifies M \owedge N
        \colon \F{} \tau
      \end{prooftree}
      \quad
      \begin{prooftree}
        M \colon \sigma \justifies \produce M \colon \F{} \sigma
      \end{prooftree}
      &
        \begin{prooftree}
          M \colon \F{} \sigma \quad N \colon \F{} \tau \justifies \pto M
          {x_\sigma} N \colon \F{} \tau
        \end{prooftree}
    \end{array}
  \]
  \caption{The syntax of CBPV$(\Demon,\Nature)$}
  \label{fig:syntax}
\end{figure}

We extend the notation $\pto M {x_\sigma} N$ to the case where $N$ has
an arbitrary computation type by: for every
$N \colon \lambda \to \underline\tau$,
$\pto M {x_\sigma} N = \lambda y_\lambda . \pto M {x_\sigma}
(Ny_\lambda)$, where $y_\lambda$ is fresh.  Similarly, we extend
$\abort_{\F{} \tau}$ to all computation types by letting
$\abort_{\lambda \to \underline\tau} = \lambda x_\lambda
. \abort_{\underline\tau}$.

The variable $x_\sigma$ is binding in $\lambda x_\sigma . M$, in
$\pto N {x_\sigma} M$, and in $\rec x_\sigma . M$, and its scope is
$M$ in all three cases.  We omit the standard definition of
$\alpha$-renaming and of capture-avoiding substitution.


Using recursion at value types may seem strange, but this allows us to
define some interesting values.  For example, we can define the
uniform distribution on $\{0, 1, 2\}$ by the term
$\rec x_{\Vt\intT} . (\retkw \underline 0 \oplus \retkw \underline 1)
\oplus (\retkw \underline 2 \oplus x_{\Vt\intT})$, which operates via a
form of rejection sampling.

We will also consider an extension of CBPV$(\Demon,\Nature)$ called
CBPV$(\Demon,\Nature)+\pifzkw+\bigcirc$, obtained by admitting the
following additional clauses:
\[
  \begin{prooftree}
    M \colon \F\Vt\unitT
    \justifies
    \bigcirc_{> b} M \colon \unitT
    \using (b \in \rat \cap (0, 1))
  \end{prooftree}
  \quad
  \begin{prooftree}
    M \colon \intT \quad
    N \colon \F{} \tau \quad
    P \colon \F{} \tau
    \justifies
    \pifz M N P \colon \F{} \tau
  \end{prooftree}
\]
$\bigcirc_{> b}$ is the \emph{statistical termination tester}, and
$\pifzkw$ is \emph{parallel if}.  We extend the notation $\pifz M N P$
to the case where $N$ and $P$ have an arbitrary computation type
$\underline\tau$ by letting $\pifz M N P$ denote
$\lambda x_\sigma . \pifz M {(N x_\sigma)} {(P x_\sigma)}$ when $N$,
$P$ have type $\sigma \to \underline\tau$, where $x_\sigma$ is a fresh
variable.

The language CBPV$(\Demon,\Nature)+\pifzkw$ is obtained by admitting
only the second one as extra clause, while
CBPV$(\Demon,\Nature)+\bigcirc$ only admits the first one as extra
clause.

\begin{figure}
  \centering
  \begin{align*}
    \Eval {x_\sigma} \rho & = \rho (x_\sigma) \\
    \Eval {\lambda x_\sigma . M} \rho
                          & = V \in \Eval \sigma \mapsto \Eval M (\rho [x_\sigma \mapsto V])
                            \qquad \Eval {MN} \rho
                            = \Eval M \rho (\Eval N \rho) \\
    \Eval {\produce M} \rho
                          & = \eta^\Smyth (\Eval M \rho) \\
    \Eval {\pto M {x_\sigma} N} \rho
                          & = \extd {(V \in \Eval \sigma \mapsto \Eval N
                            \rho [x_\sigma \mapsto V])} (\Eval M \rho) \\
    \Eval {\thunk M} \rho
                          & = \Eval M \rho
                          \qquad\qquad \Eval {\force M} \rho
                            = \Eval M \rho \\
    \Eval {\underline *} \rho
                          & = \top
                          \qquad\qquad \Eval {\underline n} \rho
                          = n \\
    \Eval {\suc M} \rho
                          & = \left\{
                            \begin{array}{ll}
                              n+1 & \text{if }n = \Eval M \rho \neq \bot \\
                              \bot & \text{otherwise}
                            \end{array}
                                     \right. \\
    \Eval {\pred M} \rho
                          & = \left\{
                            \begin{array}{ll}
                              n-1 & \text{if }n = \Eval M \rho \neq \bot \\
                              \bot & \text{otherwise}
                            \end{array}
                                     \right. \\
    \Eval {\ifz M N P} \rho
                          & = \left\{
                            \begin{array}{ll}
                              \Eval N \rho & \text{if} \Eval M \rho=0 \\
                              \Eval P \rho & \text{if} \Eval M \rho \neq
                                             0, \bot \\
                              \bot & \text{if} \Eval M \rho = \bot
                            \end{array}
                                     \right. \\
    \Eval {M; N} \rho
                          & = \left\{
                            \begin{array}{ll}
                              \Eval N \rho & \text{if} \Eval M
                                             \rho=\top \\
                              \bot & \text{otherwise}
                            \end{array}
                                     \right. \\
    \Eval {\pi_1 M} \rho
                          & = m,
    \Eval {\pi_2 M} \rho
                            = n \text{ where }\Eval M \rho = (m, n) \\
    \Eval {\langle M, N \rangle} \rho
                          & = (\Eval M \rho, \Eval N \rho) \\
    \Eval {\retkw M} \rho
                          & = \delta_{\Eval M \rho} \\
    \Eval {\dokw {x_\sigma \leftarrow M}; N} \rho
                          & = {(V \in \Eval \sigma \mapsto \Eval N
                            \rho [x_\sigma \mapsto V])}^\dagger (\Eval M \rho) \\
    \Eval {M \oplus N} \rho
                          & = \frac 1 2 (\Eval M \rho + \Eval N \rho) \\
    \Eval {M \owedge N} \rho
                          & = \Eval M \rho \wedge \Eval N \rho
                          \qquad\qquad \Eval {\abort_{\F{} \tau}} \rho
                            = \emptyset \\
    \Eval {\rec x_\sigma . M} \rho & = \lfp (V \in \Eval \sigma
                                     \mapsto \Eval M \rho [x_\sigma
                                     \mapsto V]) \\
    \\
    \Eval {\pifz M N P} \rho
                          & = \left\{
                            \begin{array}{ll}
                              \Eval N \rho & \text{if} \Eval M \rho=0 \\
                              \Eval P \rho & \text{if} \Eval M \rho \neq
                                             0, \bot \\
                              \Eval N \rho \wedge \Eval P \rho & \text{if} \Eval M \rho = \bot
                            \end{array}
                                                                 \right. \\
    \Eval {\bigcirc_{> b} M} \rho
                          & = \left\{
                            \begin{array}{ll}
                              \top & \text{if }\Eval M \rho \neq
                                     \bot\text{ and }\\
                              & b \ll \nu (\{\top\})\text{ for every }\nu \in \Eval M
                                     \rho \\
                              \bot & \text{otherwise}
                            \end{array}
                            \right. \\
  \end{align*}
  \caption{Denotational semantics}
  \label{fig:sem}
\end{figure}

\subsection{Denotational Semantics}
\label{sec:denot-semant}

Let $Env$, the dcpo of \emph{environments}, be the product of the
dcpos $\Eval \sigma$ over all variables $x_\sigma$.  Its elements are
maps $\rho$ from variables $x_\sigma$ to values $\rho (x_\sigma)$.
The denotational semantics is given by a family of Scott-continuous
maps $\Eval M$, one for each $M \colon \overline\tau$, from $Env$ to
$\Eval {\overline\tau}$: see Figure~\ref{fig:sem}, where the bottom
two clauses are specific to CBPV$(\Demon,\Nature)+\pifzkw$, resp.\ to
CBPV$(\Demon,\Nature)+\bigcirc$, and the two of them together are
specific to CBPV$(\Demon,\Nature)+\pifzkw+\bigcirc$.
We use the notation $V \in X \mapsto f (V)$ to
denote the function that maps each $V \in X$ to $f (V)$.  For every
$\rho \in Env$, and every $V \in \Eval \sigma$, we write
$\rho [x_\sigma \mapsto V]$ for the environment that maps $x_\sigma$
to $V$ and every variable $y \neq x_\sigma$ to $\rho (y)$.  The
operator $\lfp \colon [X \to X] \to X$ maps every Scott-continuous map
$f$ from a pointed dcpo to itself, to its least fixed point
$\lfp f = \sup_{n \in \nat} f^n (\bot)$.  The \emph{Dirac mass}
$\delta_x$ at $x$ is the probability valuation such that
$\delta_x (U) = 1$ if $x \in U$, $0$ otherwise.  For every continuous
map $f \colon X \to \Val_{\leq 1} Y$, $f^\dagger$ is the continuous
map from $\Val_{\leq 1} X$ to $\Val_{\leq 1} Y$ defined by
$f^\dagger (\nu) (V) = \int_{x \in X} f (x) (V) d\nu$ for every open
subset $V$ of $Y$.  For future reference, we note that
$f^\dagger (\delta_a) = f (a)$, and that, for every continuous map
$h \colon Y \to \creal$,
\begin{equation}
  \label{eq:dagger}
  \int_{y \in Y} h (y) df^\dagger (\nu) = \int_{x \in X} \left(
    \int_{y \in Y} h (y) d f (x)
    \right) d\nu.
\end{equation}
Implicit here is the fact that the map
$x \in X \mapsto \int_{y \in Y} h (y) d f (x)$ is itself continuous.
Also, integration is linear in both the integrated function $h$ and
the continuous valuation $\nu$, and Scott-continuous in each.  These
facts can be found in Jones' PhD thesis \cite{Jones:proba}.

The fact that the semantics $\Eval M \rho$ is well-defined and
continuous in $\rho$ is standard.  Note the use of binary infimum
($\wedge$) in the semantics of $\owedge$ and of $\pifzkw$, for which
we use the following lemma.
\begin{lemma}
  \label{lemma:sound:aux}
  Let $L$ be a continuous complete lattice.
  \begin{enumerate}
  \item The infimum map $\wedge \colon L \times L \to L$ is
    Scott-continuous.
  \item For any two continuous maps $f, g \colon X \to L$, where $X$
    is a any topological space, the infimum $f \wedge g$ is computed
    pointwise: $(f \wedge g) (x) = f (x) \wedge g (x)$.
  \end{enumerate}
\end{lemma}
\begin{proof}
  Item~1 is well-known.  Explicitly, one must show that for every
  directed family ${(x_i)}_{i \in I}$ with supremum $x$ in $M$, for
  every $y \in L$, $y \wedge \sup_{i \in I} x_i \leq \sup_{i \in I} (y
  \wedge x_i)$: for every $z \ll y \wedge \sup_{i \in I} x_i$, $z$ is
  below $y$ and below some $x_i$, hence below $y \wedge x_i$ for some
  $i \in I$.

  As for item~2, the composition of $\wedge$ with $x \mapsto (f (x), g
  (x))$ is continuous by item~1, is below $f$ and $g$, and is clearly
  above any lower bound of $f$ and $g$.  \qed
\end{proof}

\subsection{Operational Semantics}
\label{sec:oper-semant}

We choose an operational semantics in the style of \cite{jgl-jlap14}.
It operates on configurations, which are pairs $C \cdot M$ of an
evaluation context $C$ and a term $M$.  The deterministic part of the
calculus will be defined by rewrite rules $C \cdot M \to C' \cdot M'$
between configurations.  For the probabilistic and non-deterministic
part of the calculus, we will rely on judgments $C \cdot M \vp a$,
which state, roughly, that the probability that computation
terminates, starting from $C \cdot M$, is larger than $a$.

The \emph{elementary contexts}, together with their types
$\anytype \vdash \anyothertype$ (where $\anytype$, $\anyothertype$ are
value or computation types) are defined by:
\begin{itemize}
\item $[\_ N] \colon (\sigma \to \underline\tau) \vdash \underline\tau$,
  for every $N \colon \sigma$ and every computation type $\underline
  \tau$;
\item $[\pto {\_} {x_\sigma} N] \colon \F{} \sigma \vdash \F{}\tau$
  for every $N \colon \F{}\tau$;
\item $[\force \_] \colon \U\underline\tau \vdash \underline\tau$, for
  every computation type $\underline\tau$;
\item $[\suc \_], [\pred \_] \colon \intT \vdash \intT$;
\item $[\ifz \_ N P] \colon \intT \vdash \anytype$ for all
  $N, P \colon \anytype$; 
\item $[\_; N] \colon \unitT \vdash \anytype$ for every $N \colon \anytype$;
\item $[\pi_1 \_] \colon \sigma \times \tau \vdash \sigma$ and
  $[\pi_2 \_] \colon \sigma \times \tau \vdash \tau$, for all value
  types $\sigma$ and $\tau$;
\item
  $[\dokw {x_\sigma \leftarrow \_}; N] \colon \Vt{} \sigma \vdash \Vt{} \tau$,
  for every $N \colon \sigma \to \Vt{} \tau$.
\end{itemize}
The \emph{initial contexts} are $[\_] \colon \anytype \vdash \anytype$,
$[\produce \_] \colon \sigma \vdash \F{} \sigma$ and
$[\produce \retkw \_] \colon \sigma\vdash \F\Vt{} \sigma$.  For every
elementary or initial context $E \colon \anytype \vdash \anyothertype$ and every
$M \colon \anytype$, we write $E [M]$ for the result of replacing the unique
occurrence of the hole $\_$ in $E$ (after removing the outer square
brackets) by $M$.  E.g,
$[\suc \_] [\underline 3] = \suc \underline 3$.

A \emph{context} (of type $\anytype \vdash \anyothertype$) is a finite list
$E_0 E_1 E_2 \cdots E_n$ ($n \in \nat$) where $E_0$ is an initial
context, $E_1$, \ldots, $E_n$ are elementary contexts, and
$E_i \colon \anytype_{i+1} \vdash \anytype_i$, $\anytype_{n+1}=\anytype$, and $\anytype_0=\anyothertype$.  We then
write $C [M]$ for $E_0 [E_1 [E_2 [ \cdots E_n [M] \cdots]]]$.

Note that the contexts are defined in exactly the same way for
CBPV$(\Demon,\Nature)$ and for CBPV$(\Demon,\Nature)+\pifzkw$,
CBPV$(\Demon,\Nature)+\bigcirc$, and
CBPV$(\Demon,\Nature)+\pifzkw+\bigcirc$.

The \emph{configurations} of the operational semantics are pairs
$C \cdot M$ where $C \colon \anytype \vdash \F\Vt\unitT$ and
$M \colon \anytype$.  The rules of the operational semantics are given
in Figure~\ref{fig:opsem}.  The last row is specific to
CBPV$(\Demon,\Nature)+\pifzkw$, CBPV$(\Demon,\Nature)+\bigcirc$, or to
CBPV$(\Demon,\Nature)+\pifzkw+\bigcirc$.  The first rewrite rule---the
\emph{redex discovery rule} $C \cdot E [M] \to C E \cdot M$---applies
provided $E$ is an elementary context.  The notation
$N [x_\sigma := M]$ denotes capture-avoiding substitution of $M$ for
$x_\sigma$ in $N$.

\begin{figure}
  \centering\footnotesize
  \begin{align*}
    C \cdot E [M]
    & \to C E \cdot M
    & C [\_ N] \cdot \lambda x_\sigma . M
    & \to C \cdot M [x_\sigma := N] \\
    C [\pto {\_} {x_\sigma} N] \cdot \produce M
    & \to C \cdot N [x_\sigma := M]
    & C [\force \_] \cdot \thunk M
    & \to C \cdot M \\
    [\_] \cdot \produce M
    & \to{} [\produce \_] \cdot M \\
    C [\pred \_] \cdot \underline n
    & \to C \cdot \underline{n-1}
    & C [\suc \_] \cdot \underline n
    & \to C \cdot \underline{n+1} \\
    C [\ifz \_ N P] \cdot \underline 0
    & \to C \cdot N
    & C [\ifz \_ N P] \cdot \underline n
    & \to C \cdot P \quad (n\neq 0) \\
    C [\_; N] \cdot \underline *
    & \to C \cdot N \\
    C [\pi_1 \_] \cdot \langle M, N \rangle
    & \to C \cdot M
    & C [\pi_2 \_] \cdot \langle M, N \rangle
    & \to C \cdot N \\
    C [\dokw {x_\sigma \leftarrow \_}; N] \cdot \retkw M
    & \to C \cdot N [x_\sigma := M]
    & {} [\produce \_] \cdot \retkw M
    & \to{} [\produce \retkw \_] \cdot M \\
    C \cdot \rec x_\sigma . M
    & \to C \cdot M [x_\sigma := \rec x_\sigma . M]\mskip-200mu \\
  \end{align*}
  \[
    \begin{array}{c}
      \begin{prooftree}
        \strut
        \justifies
        [\produce \retkw \_] \cdot \underline* \vp a
        \using (a \in \rat \cap [0, 1))
      \end{prooftree}
      \qquad
      \begin{prooftree}
        \strut
        \justifies
        C \cdot M \vp 0
      \end{prooftree}
      \qquad
      \begin{prooftree}
        \strut
        \justifies
        C \cdot \abort_{\F{} \tau} \vp a
        \using (a \in \rat \cap [0, 1))
      \end{prooftree}
      \\
      \\
      \begin{prooftree}
        C' \cdot M' \vp a
        \justifies
        C \cdot M \vp a
        \using (\text{if }C \cdot M \to C' \cdot M')
      \end{prooftree}
      \qquad
      \begin{prooftree}
        C \cdot M \vp a \quad C \cdot N \vp b
        \justifies
        C \cdot M \oplus N \vp (a+b)/2
      \end{prooftree}
      \qquad
      \begin{prooftree}
        C \cdot M \vp a \quad C \cdot N \vp a
        \justifies
        C \cdot M \owedge N \vp a
      \end{prooftree}
      \\
      \\
      \begin{prooftree}
        [\_] \cdot M \vp b \quad
        C \cdot \underline* \vp a
        \justifies
        C \cdot \bigcirc_{> b} M \vp a
      \end{prooftree}
      \qquad
      \begin{prooftree}
        C \cdot \ifz M N P \vp a
        \justifies
        C \cdot \pifz M N P \vp a
      \end{prooftree}
      \qquad
      \begin{prooftree}
        C \cdot N \vp a \quad C \cdot P \vp a
        \justifies
        C \cdot \pifz M N P \vp a
      \end{prooftree}
    \end{array}
  \]
  \caption{Operational semantics}
  \label{fig:opsem}
\end{figure}

The judgments $C \cdot M \vp a$ are defined for all terms
$M \colon \anytype$, contexts $C \colon \anytype \vdash \F\Vt{} \unitT$,
and $a \in \rat \cap [0, 1)$, and mean that $a$ is way-below the
probability of termination of $C \cdot M$ (i.e., either $a=0$ or $a$
is strictly less than the probability that $C \cdot M$ terminates).
Since $\owedge$ induces non-deterministic choice, we really mean the
probability of \emph{must-}termination, namely that, in whichever way
the non-determinism involved in the use of the $\owedge$ operator is
resolved (evaluating left, or right), the final probability is larger
than $a$.

We write $\Prob (C \cdot M \vp)$ for
$\sup \{a \in \rat \in [0, 1) \mid C \cdot M \vp a \text{ is
  derivable}\}$, where sups are taken in $[0, 1]$.  This leads to the
following central notion, which we only state for ground terms.  A
term is \emph{ground} if and only if it has no free variable.  (We
define \emph{ground} contexts similarly.)  The case of non-ground
terms can be dealt with using appropriate quantifications over
substitutions, but will not be needed.
\begin{definition}
  \label{defn:context:preorder}
  The \emph{contextual preorder} $\precsim_\anytype$ between ground
  CBPV$(\Demon,\Nature)$ terms of type $\anytype$ is defined by
  $M \precsim_\anytype N$ if and only if for every ground evaluation context
  $C \colon \anytype \vdash \F\Vt\unitT$,
  $\Prob (C \cdot M \vp) \leq \Prob (C \cdot N \vp)$.
%
\end{definition}
We will freely reuse the notations $\precsim_\anytype$, for the
similarly defined notions on the related languages
CBPV$(\Demon,\Nature)+\pifzkw$, CBPV$(\Demon,\Nature)+\bigcirc$, and
CBPV$(\Demon,\Nature)+\pifzkw+\bigcirc$.  If there is any need to make
the language precise, we will mention it explicitly.

We end this section with a few elementary lemmata, which will come in
handy later on, and which should help the reader train with the way
the operational semantics works.
\begin{lemma}
  \label{lemma:prob:leq}
  If $C \cdot M \vp a$ is derivable and $b \in \rat$ is such that
  $0 \leq b \leq a$, then $C \cdot M \vp b$ is also derivable, whether
  in CBPV$(\Demon,\Nature)$, CBPV$(\Demon,\Nature)+\pifzkw$,
  CBPV$(\Demon,\Nature)+\bigcirc$, or
  CBPV$(\Demon,\Nature)+\pifzkw+\bigcirc$.
\end{lemma}
\begin{proof}
  Easy induction on the rules of Figure~\ref{fig:opsem}.  In the case
  of a derivation of the form $C \cdot M \oplus N \vp a$, where
  $a = (a_1+a_2)/2$, from $C \cdot M \vp a_1$ and $C \cdot N \vp a_2$,
  we write $b$ as $(b_1+b_2)/2$ where $b_1$ and $b_2$ are rational and
  between $0$ and $a_1$, resp.\ $a_2$.  (E.g., we let
  $b_1 = \min (a_1, 2b)$ and $b_2 = 2b-b_1 = \max (2b - a_1, 0)$.)
  By induction hypothesis we can derive $C \cdot M \vp b_1$ and
  $C \cdot N \vp b_2$, so we can derive
  $C \cdot M \oplus N \vp (b_1+b_2)/2 = b$.  \qed
\end{proof}

\begin{lemma}
  \label{lemma:red}
  If $C \cdot M \to C' \cdot M'$, then $\Prob (C \cdot M\vp) \geq
  \Prob (C' \cdot M'\vp)$.
\end{lemma}
\begin{proof}
  Whenever we can derive $C' \cdot M' \vp a$, we can derive $C \cdot M
  \vp a$ by the leftmost rule of the next-to-last row of
  Figure~\ref{fig:opsem}.  \qed
\end{proof}



\begin{lemma}
  \label{lemma:prob:disc}
  Let $C' = E_1 \cdots E_n$ be a sequence of elementary contexts, of
  type $\anytype \vdash \anyothertype$.  For every context
  $C \colon \anyothertype \vdash \F\Vt\unitT$, for every term
  $N \colon \anytype$,
  $\Prob (C \cdot C' [N] \vp) = \Prob (C C' \cdot N \vp)$.
\end{lemma}
\begin{proof}
  By the redex discovery rule, $C \cdot C' [N] \to^* C C' \cdot N$, so
  $\Prob (C \cdot C' [N] \vp) \geq \Prob (C C' \cdot N \vp)$ by
  Lemma~\ref{lemma:red}.  Conversely, if $C \cdot C' [N] \vp a$ is
  derivable, then we show that $C C' \cdot N \vp a$ is derivable by
  induction on $n$.  If $n=0$, this is clear.  Otherwise, there are
  only two rules that allow us to derive $C \cdot C' [N]\vp a$.  In
  the case of the first of these rules (the middle rule of the first
  of the three rows of rules), $a=0$, and we can derive
  $C C' \cdot N \vp a$ by the same rule.  In the case of the other
  rule, $C \cdot C' [N] \vp a$ was derived from a shorter derivation
  of $C E_1 \cdot C'' [N]\vp a$, where $C'' = E_2 \cdots E_n$, using
  the redex discovery rule
  $C \cdot C' [N] = C \cdot E_1 [C'' [N]] \to C E_1 \cdot C'' [N]$.
  By induction hypothesis, $C E_1 C'' \cdot N$ is derivable, namely
  $C C' \cdot N$ is derivable.  \qed
\end{proof}

\begin{lemma}
  \label{lemma:prob:reduce}
  Let $C' = E_1 \cdots E_n$ be a sequence of elementary contexts.  If
  $C \cdot M \to^* C C' \cdot N$ then
  $\Prob (C \cdot M\vp) \geq \Prob (C \cdot C' [N]\vp)$.
\end{lemma}
\begin{proof}
  $\Prob (C \cdot M \vp) \geq \Prob (C C' \cdot N \vp) = \Prob (C
  \cdot C' [N] \vp)$, by Lemma~\ref{lemma:red} and
  Lemma~\ref{lemma:prob:disc}.  \qed
\end{proof}

For short, let us write $\Prob (M \vp)$ for
$\Prob ([\_] \cdot M \vp)$.
\begin{lemma}
  \label{lemma:prob:disc:_}
  Let $C$ be any context of type $\anytype \vdash \F\Vt\unitT$.  For every term
  $M \colon \anytype$, $\Prob (C [M] \vp) = \Prob (C \cdot M \vp)$.
\end{lemma}
\begin{proof}
  Let us write $C$ as $E_0 C'$ where
  $E_0$ is an initial context and $C' = E_1 E_2 \cdots E_n$ is a
  sequence of elementary contexts.  We first show that
  $\Prob ([\_] \cdot C [M] \vp) = \Prob (E_0 \cdot C' [M] \vp)$.  Once
  this is done, Lemma~\ref{lemma:prob:disc} states that
  $\Prob (E_0 \cdot C' [M] \vp) = \Prob (E_0 C' \cdot M \vp) = \Prob
  (C \cdot M \vp)$, and that will finish the proof.

  We assume $E_0 \neq [\_]$, otherwise the claim is trivial.  Then
  $[\_] \cdot C [M] \to^* E_0 \cdot C' [M]$.  Indeed,
  $[\_] \cdot C [M] \to [\produce \_] \cdot C' [M]$ if
  $E_0 = [\produce \_]$, and
  $[\_] \cdot C [M] \to [\produce \_] \cdot \retkw C' [M] \to
  [\produce \retkw \_] \cdot C' [M]$ if $E_0 = [\produce \retkw \_]$.
  By Lemma~\ref{lemma:red},
  $\Prob ([\_] \cdot C [M] \vp) \geq \Prob (E_0 \cdot C' [M] \vp)$.

  In the converse direction, assume that $[\_] \cdot C [M] \vp a$ is
  derivable.  If $a=0$, then $E_0 \cdot C' [M] \vp a$ is also
  derivable.  Otherwise, if $E_0 = [\produce \_]$, then the only
  remaining possible derivation is obtained from a smaller derivation
  of $[\produce \_] \cdot C' [M] \vp a$, so $E_0 \cdot C' [M] \vp a$
  is again derivable.  If $a \neq 0$ and $E_0 = [\produce \retkw \_]$,
  then we can only have derived $[\_] \cdot C [M] \vp a$ from a
  smaller derivation of $[\produce \_] \cdot \retkw C' [M] \vp a$, and
  then from another derivation of
  $[\produce \retkw \_] \cdot C' [M] \vp a$, namely
  $E_0 \cdot C' [M] \vp a$.  Since that holds for every $a$,
  $\Prob ([\_] \cdot C [M] \vp) \leq \Prob (E_0 \cdot C' [M] \vp)$.
  \qed
\end{proof}

\section{Soundness}
\label{sec:soundness}

We let the \emph{rank} of a type be $0$ for a value type that is not
of the form $\Vt{} \sigma$, $1/2$ for types of the form $\Vt{} \sigma$, and
$1$ for computation types.  This will play a key role in our soundness
proof, for the following reason: for every elementary or initial
context $E \colon \anytype \vdash \anyothertype$, the rank of
$\anytype$ is less than or equal to the rank of $\anyothertype$.
Hence if $C = E_0 E_1 E_2 \cdots E_n$ is of type
$\anytype \vdash \anyothertype$, and $E_i$ is of type
$\anytype_{i+1} \vdash \anytype_i$, then every $\anytype_i$ has rank
between those of $\anytype$ and $\anyothertype$.

Beyond its role as a technical aide, the concept of rank is profitably
interpreted from the point of view of the type and effect discipline
\cite{TJ:effects}.  While the separation between value types and
computation types exhibits two kinds of effects, ranks refine this
further by distinguishing between rank $0$ value types, where the only
effect is recursion, from rank $1/2$ value types, where probabilistic
choice is also allowed.  Rank $1$ types further allow for
non-deterministic choice effects.  With that viewpoint, one might be
puzzled by the fact that the rank $0$ types $\U\underline \tau$ are
able to encapsulate arbitrary rank $1$ types.  However, the typical
inhabitants of types $\U\underline \tau$ are \emph{thunks} $\thunk M$,
which do \emph{not} execute, hence do not produce any side effect,
unless being forced to, using the $\force$ operation, yielding again a
value of the rank $1$ type $\underline \tau$.

We will also need to define the semantics of contexts
$C \colon \anytype \vdash \F\Vt\unitT$ so that
$\Eval {C [M]} \rho = \Eval C \rho (\Eval M \rho)$ for every
$M \colon \anytype$ and for every environment $\rho$.
$\Eval {E_0 E_1 E_2 \cdots E_n} \rho$ is the composition of
$\Eval {E_0} \rho$, $\Eval {E_1} \rho$, $\Eval {E_2} \rho$, \ldots,
$\Eval {E_n} \rho$, where:
\begin{itemize}
\item $\Eval {[\_ N]} \rho$ maps $f$ to $f (\Eval N \rho)$,
\item
  $\Eval {[\pto {\_} {x_\sigma} N]} \rho = \extd {(V \in \Eval \sigma
    \mapsto \Eval N \rho [x_\sigma \mapsto V])}$,
\item $\Eval {[\force \_]} \rho$ is the identity map,
\item $\Eval {[\suc \_]} \rho$ maps $\bot$ to $\bot$ and otherwise
  adds one,
\item $\Eval {[\pred \_]} \rho$ maps $\bot$ to $\bot$ and otherwise
  subtracts one,
\item $\Eval {[\ifz \_ N P]} \rho$ maps $0$ to $\Eval N \rho$, every
  non-zero number to $\Eval P \rho$ and $\bot$ to $\bot$,
\item $\Eval {[\_; N]} \rho$ maps $\top$ to $\Eval N \rho$, and $\bot$
  to $\bot$,
\item $\Eval {[\pi_1 \_]} \rho$ is first projection,
\item $\Eval {[\pi_2 \_]} \rho$ is second projection,
\item
  $\Eval {[\dokw {x_\sigma \leftarrow \_}; N]} \rho = {(V \in \Eval
    \sigma \mapsto \Eval N \rho [x_\sigma \mapsto V])}^\dagger$,
\item $\Eval {[\produce \_]} \rho = \eta^\Smyth$, and
\item $\Eval {[\produce \retkw \_]} \rho$ maps $V$ to
  $\eta^\Smyth (\delta_V)$.
\end{itemize}

\begin{proposition}[Soundness]
  \label{prop:sound}
  Let $C \colon \anytype \vdash \F\Vt\unitT$, $M \colon \anytype$, where $\anytype$ is a value
  or computation type, and let $\rho \in Env$.  In
  CBPV$(\Demon,\Nature)$, in CBPV$(\Demon,\Nature)+\pifzkw$, in
  CBPV$(\Demon,\Nature)+\bigcirc$, and in
  CBPV$(\Demon,\Nature)+\pifzkw+\bigcirc$:
  \begin{enumerate}
  \item For every $a \in \rat \cap [0, 1)$, if $C \cdot M \vp a$ is
    derivable, then either $\Eval {C [M]} \rho=\bot$ and $a=0$, or
    $\Eval {C [M]} \rho\neq \bot$ and for every
    $\nu \in \Eval {C [M]} \rho$, $a \ll \nu (\{\top\})$.
  \item If $\Eval {C [M]}\rho=\bot$ then $\Prob (C \cdot M \vp)=0$,
    otherwise for every $\nu \in \Eval {C [M]} \rho$,
    $\nu (\{\top\}) \geq \Prob (C \cdot M \vp)$.
  \end{enumerate}
\end{proposition}

\begin{proof}
  Item~2 is an easy consequence of item~1, which we prove by induction
  on the derivation.

  In the case of the first rule
  ($[\produce \retkw \_] \cdot \underline* \vp a$),
  $C [M] = \produce \retkw \underline*$, and
  $\Eval {C [M]} \neq \bot$.  For every
  $\nu \in \Eval {C [M]} \rho = \eta^\Smyth (\delta_\top)$, we have
  $\nu \geq \delta_\top$, so $\nu (\{\top\}) \geq 1$, and certainly
  $a \ll 1$ for every $a \in \rat \cap [0, 1)$.

  The case of the second rule $C \cdot M \vp 0$ is obvious.

  The case of the leftmost rule of the next row follows from the
  observation that if $C \cdot M \to C' \cdot M'$, then
  $\Eval {C [M]} \rho = \Eval {C' [M']} \rho$.  We use the standard
  substitution lemma
  $\Eval M (\rho [x_\sigma \mapsto \Eval N \rho]) = \Eval {M [x_\sigma
    := N]} \rho$ in the case of $\beta$-reduction
  ($C [\_ N] \cdot \lambda x_\sigma . M \to C \cdot M [x_\sigma :=
  N]$): the value of the left-hand side is
  $\Eval C \rho (\Eval M (\rho [x_\sigma \mapsto \Eval N \rho]))$, and
  the value of the right-hand side is
  $\Eval C \rho (\Eval {M [x_\sigma:=N]} \rho)$).  In the case of
  $C [\pto {\_} {x_\sigma} N] \cdot \produce M \to C \cdot N [x_\sigma
  := M]$, we also use the fact that
  $\extd {(V \in \Eval \sigma \mapsto \Eval N \rho [x_\sigma := V])}
  (\eta^\Smyth (\Eval M \rho)) = \Eval N \rho [x_\sigma \mapsto \Eval
  M \rho]$ (Proposition~\ref{prop:Qtop}, item~2).  In the case of
  $C [\dokw {x_\sigma \leftarrow \_}; N] \cdot \retkw M \to C \cdot N
  [x_\sigma := M]$, we use the equality $f^\dagger (\delta_x) = f (x)$
  and the substitution lemma.

  By our observation on ranks, if
  $C \colon \F{} \sigma \vdash \F\Vt{} \unitT$, where
  $C = E_0 E_1 E_2 \cdots \allowbreak E_n$ and
  $E_i \colon \anytype_{i+1} \vdash \anytype_i$ for each $i$, then all the types
  $\anytype_i$ are computation types (rank $1$).  In that case, $E_i$ can
  only be of one of the two forms $[\_ N]$,
  $[\pto {\_} {x_\sigma} N]$.  (Further inspection would reveal that
  the first case is impossible, but we will not need that yet.)  We
  now observe that in each case, $\Eval {E_i} \rho$ maps top to top:
  in the case of $[\pto {\_} {x_\sigma} N]$, this is by
  Proposition~\ref{prop:Qtop}, item~3.  It follows that $\Eval C \rho$
  also maps top to top, whence
  $\Eval {C [\abort_{\F{}\sigma}]} \rho = \Eval C \rho (\Eval
  {\abort_{\F{}\sigma}} \rho) = \Eval C \rho (\emptyset) = \emptyset$.
  As a consequence, $\Eval {C [\abort_{\F{}\sigma}]} \rho \neq \bot$, and
  the claim that for every
  $\nu \in \Eval {C [\abort_{\F{}\sigma}]} \rho$,
  $\nu (\{\top\}) \geq \Prob (C \cdot \abort \vp)$ is vacuously true:
  the rule that derives $C \cdot \abort_{\F{}\sigma} \vp a$ for every
  $a \in \rat \cap [0, 1)$ is sound.

  Similarly, and still assuming
  $C \colon \F{} \sigma \vdash \F\Vt\unitT$, for each $i$,
  $\Eval {E_i} \rho$ preserves binary infima.  When
  $E_i = [\pto {\_} {x_\sigma} N]$, this is because the function
  $\extd {(V \in \Eval \sigma \mapsto \Eval N \rho [x_\sigma \mapsto
    V])}$ maps binary infima to binary infima by
  Proposition~\ref{prop:Qtop}, item~3.  When $E_i = [\_ N]$,
  $\Eval {[\_ N]} \rho$ maps every $f$ to $f (\Eval N \rho)$, and
  therefore preserves binary infima by Lemma~\ref{lemma:sound:aux},
  item~2.  It follows that $\Eval C \rho$ preserves binary infima.  We
  apply this to the rightmost rule of the middle row (if
  $C \cdot M \vp a$ and $C \cdot N \vp a$ then
  $C \cdot M \owedge N \vp a$).  We have
  $\Eval {C [M \owedge N]} \rho = \Eval C \rho (\Eval M \rho \wedge
  \Eval N \rho) = \Eval C \rho (\Eval M \rho) \wedge \Eval C \rho
  (\Eval N \rho) = \Eval {C [M]} \rho \wedge \Eval {C [N]} \rho$.

  In particular, if $\Eval {C [M \owedge N]} \rho = \bot$, and since
  $a \wedge b = \bot$ implies $a=\bot$ or $b=\bot$ in any space of the
  form $\Smyth^\top_\bot (X)$, then $\Eval {C [M]} \rho$ or
  $\Eval {C [N]} \rho$ is equal to $\bot$.  By symmetry, let us assume
  that $\Eval {C [M]} \rho = \bot$.  By induction hypothesis, the only
  value of $a$ such that $C \cdot M \vp a$ is derivable is $a=0$.
  There are only two rules that can end a derivation of
  $C \cdot M \owedge N \vp a$, and they both require $a=0$.

  If $\Eval {C [M \owedge N]} \rho \neq \bot$, then
  $\Eval {C [M]} \rho \neq \bot$ and $\Eval {C [N]} \rho \neq \bot$,
  so by induction hypothesis, for every $\nu$ in $\Eval {C [M]} \rho$,
  and for every $\nu$ in $\Eval {C [N]} \rho$, $a \ll \nu (\{\top\})$.
  Hence this holds for every
  $\nu \in \Eval {C [M \owedge N]} \rho = \Eval {C [M]} \rho \wedge
  \Eval {C [N]} \rho = \Eval {C [M]} \rho \cup \Eval {C [N]} \rho$.

  Let us deal with the last of the CBPV$(\Demon,\Nature)$ rules
  (middle rule, middle row of Figure~\ref{fig:opsem}): we have deduced
  $C \cdot M \oplus N \vp (a+b)/2$ from $C \cdot M \vp a$ and
  $C \cdot N \vp b$, hence by induction hypothesis: $(a)$ either
  $\Eval {C [M]} \rho = \bot$ and $a=0$, or for every
  $\nu \in \Eval {C [M]} \rho$, $a \ll \nu (\{\top\})$; and $(b)$
  either $\Eval {C [N]} \rho = \bot$ and $b=0$, or for every
  $\nu \in \Eval {C [N]} \rho$, $b \ll \nu (\{\top\})$.  In that case
  $C = E_0 E_1 E_2 \cdots E_n$ has type $\Vt{} \sigma \vdash \F\Vt\unitT$ for
  some value type $\sigma$, and every intermediate type $\anytype_i$ must
  therefore have rank $1/2$ or $1$.  The only eligible elementary
  contexts $E_i \colon \anytype_{i+1} \vdash \anytype_i$ ($1\leq i \leq n$) are of
  the form $[\_ N]$, $[\pto {\_} {x_\sigma} N]$, or
  $[\dokw {x_\sigma \leftarrow \_}; N]$.  In each case, the rank of
  $\anytype_i$ is equal to that of $\anytype_{i+1}$.  Since $\anytype_{n+1} = \F\Vt\unitT$ has
  rank $1$ and $\anytype_0 = \Vt{} \sigma$ has rank $1/2$, $E_0$ cannot be
  $[\_] \colon \F\Vt{} \unitT \vdash \F\Vt{} \unitT$.  It cannot be
  $[\produce \retkw \_] \colon \unitT \vdash \F\Vt{} \unitT$ either since
  $\unitT$ has rank $0$.  Hence $E_0$ is equal to
  $[\produce \_] \colon \Vt\unitT \vdash \F\Vt{} \unitT$, and every $E_i$
  ($1\leq i\leq n$) is of the form
  $[\dokw {x_\sigma \leftarrow \_}; N]$.  We note that
  $\Eval {[\dokw {x_\sigma \leftarrow \_}; N]} \rho = {(V \in \Eval
    \sigma \mapsto \Eval N \rho [x_\sigma \mapsto V])}^\dagger$ is a
  linear map, i.e., preserves sums and scalar multiplication.  Indeed
  the formula $f^\dagger (\nu) (V) = \int_{x \in X} f (x) (V) d\nu$ is
  linear in $\nu$.  It follows that $\Eval {E_1 E_2 \cdots E_n} \rho$
  is also linear, so
  $\Eval {E_1 E_2 \cdots E_n [M \oplus N]} \rho = \Eval {E_1 E_2
    \cdots E_n} \rho (\frac 1 2 (\Eval M \rho + \Eval N \rho)) = \frac
  1 2 (\nu_1 + \nu_2)$, where
  $\nu_1 = \Eval {E_1 E_2 \cdots E_n [M]} \rho$ and
  $\nu_2 = \Eval {E_1 E_2 \cdots E_n [N]} \rho$.  Note that
  $\Eval {C [M]} \rho = \eta^\Smyth (\nu_1) = \upc \nu_1$, and
  similarly $\Eval {C [N]} \rho = \upc \nu_2$, and that those values
  are different from $\bot$.  Similarly,
  $\Eval {C [M \oplus N]} \rho = \upc (\frac 1 2 (\nu_1+\nu_2))$ is
  different from $\bot$.  Since $\nu_1 \in \Eval {C [M]} \rho$, we
  obtain that $a \ll \nu_1 (\{\top\})$ by $(a)$.  Similarly,
  $b \ll \nu_2 (\{\top\})$.  Using the fact that, for all
  $s, t \in [0, 1]$, $s \ll t$ if and only if $s=0$ or $s<t$,
  $(a+b)/2 \ll \frac 1 2 (\nu_1 (\{\top\}) + \nu_2 (\{\top\}))$.  For
  every
  $\nu \in \Eval {C [M \oplus N]} \rho = \upc (\frac 1 2 (\nu_1 +
  \nu_2))$, and we therefore obtain that $(a+b)/2 \ll \nu (\{\top\})$.

  We turn to the rules of the bottom row, which are specific to the
  extensions of CBPV$(\Demon,\Nature)$ with $\pifzkw$, or $\bigcirc$,
  or both.  For the first one, by induction hypothesis either
  $\Eval {M} \rho = \bot$ and then $b=0$, or else
  $b \ll \mu (\{\top\})$ for every $\mu \in \Eval M \rho$.  The first
  case is impossible since it is a requirement of the syntax of
  $\bigcirc_{> b} M$ that $b$ be non-zero.  This implies that
  $\Eval {\bigcirc_{> b} M} \rho = \top$.  It follows that
  $\Eval {C [\bigcirc_{> b} M]} \rho = \Eval C \rho (\Eval
  {\bigcirc_{> b} M} \rho) = \Eval C \rho (\top) = \Eval {C
    [\underline *]} \rho$.  By induction hypothesis again, either
  $\Eval {C [\underline *]} \rho = \bot$ and $a=0$, or
  $a \ll \nu (\{\top\})$ for every $\nu$ in
  $\Eval {C [\underline*]} \rho$.  Hence either
  $\Eval {C [\bigcirc_{> b} M]} \rho = \bot$ and $a=0$, or
  $a \ll \nu (\{\top\})$ for every $\nu$ in
  $\Eval {C [\bigcirc_{> b} M]} \rho$.

  For the last two, we note that, in all three cases on $\Eval M \rho$
  (equal to $\bot$, to $0$, or other), $\Eval {C [\pifz M N P]} \rho$
  is equal to one of the terms $\Eval {C [\ifz M N P]} \rho$ or
  $\Eval {C [N \owedge P]} \rho$, and is larger than or equal to the
  other one.  In other words,
  $\Eval {C [\pifz M N P]} \rho = \max (\Eval {C [\ifz M N P]} \rho,
  \Eval {C [N \owedge P]} \rho)$.  If that is equal to $\bot$, then
  both terms $\Eval {C [\ifz M N P]} \rho$ and
  $\Eval {C [N \owedge P]} \rho$ are equal to $\bot$, so by induction
  hypothesis the only derivations of $C \cdot \ifz M N P \vp a$ and
  $C \cdot N \owedge P \vp a$ are such that $a=0$.  Hence the only
  derivations of $C \cdot \pifz M N P \vp a$ are such that $a=0$,
  using any of the three possible rules.  If
  $\Eval {C [\pifz M N P]} \rho \neq \bot$, then let us assume that
  $C \cdot \pifz M N P \vp a$ by any of the last two rules.  If $a=0$,
  then certainly $a \ll \nu (\{\top\})$ for every
  $\nu \in \Eval {C [\pifz M N P]} \rho$.  Otherwise, by induction
  hypothesis we have $\Eval {C [\ifz M N P]} \rho \neq \bot$ and
  $a \ll \nu (\{\top\})$ for every
  $\nu \in \Eval {C [\ifz M N P]} \rho$, or
  $\Eval {C [N \owedge P]} \rho \neq \bot$ and $a \ll \nu (\{\top\})$
  for every $\nu \in \Eval {C [N \owedge P]} \rho$.  Since
  $\Eval {C [\pifz M N P]} \rho = \max (\Eval {C [\ifz M N P]} \rho,
  \Eval {C [N \owedge P]} \rho)$, and $\max$ means smallest with
  respect to inclusion (for non-bottom elements), in particular
  $a \ll \nu (\{\top\})$ for every
  $\nu \in \Eval {C [\pifz M N P]} \rho$.  \qed
\end{proof}

\section{Adequacy}
\label{sec:adequacy}

Adequacy is proved through the use of a suitable logical relation
${(\R_\anytype)}_{\anytype \text{ type}}$, where $\R_\anytype$ relates
ground terms of type $\anytype$ with elements of $\Eval \anytype$.  (A
term is \emph{ground} if and only if it has no free variable.  We
define \emph{ground} contexts similarly.)  Again we work in
CBPV$(\Demon,\Nature)$ or any of its extensions with $\pifzkw$ or
$\bigcirc$ or both, without further mention.

Defining $\R_\anytype$ necessitates that we also define auxiliary
relations $\R_\sigma^\perp$ between ground contexts
$C \colon \Vt\sigma \vdash \F\Vt\unitT$ (resp., $\R_\sigma^*$ between
ground contexts $C \colon \F{}\sigma \vdash \F\Vt\unitT$) and continuous
maps $h \colon \Eval \sigma \to [0, 1]$.  This pattern is similar to
the technique of $\top\top$-lifting, and particularly to Katsumata's
$\top\top$-logical predicates \cite{Katsumata:TTlifting}.  We write
``for all $C \R_\sigma^\perp h$'' instead of ``for every ground
context $C \colon \Vt\sigma \vdash \F\Vt\unitT$ and for every continuous
map $h \colon \Eval \sigma \to [0, 1]$ such that
$C \R_\sigma^\perp h$'', $C \R_\sigma^* h$ instead of ``for every
ground context $C \colon \F{}\sigma \vdash \F\Vt\unitT$ and for every
continuous map $h \colon \Eval \sigma \to [0, 1]$ such that
$C \R_\sigma^* h$'', and $M \R_\sigma a$ instead of ``for every ground
term $M \colon \sigma$ and for every $a \in \Eval \sigma$ such that
$M \R_\sigma a$''.  We define:
\begin{itemize}
\item $M \R_{\U\underline\tau} h$ iff $\force M \R_{\underline\tau} h$;
\item $M \R_{\unitT} \top$ if $\underline * \precsim_{\unitT} M$, and
  $M \R_{\unitT} \bot$ always;
\item $M \R_{\intT} n$ if $\underline n \precsim_{\intT} M$, and
  $M \R_{\intT} \bot$ always;
\item $M \R_{\sigma \times \tau} (V_1, V_2)$ if and only if $\pi_1 M
  \R_\sigma V_1$ and $\pi_2 M \R_\tau V_2$;
\item $M \R_{\Vt{} \sigma} \nu$ if and only if $\Prob (C \cdot M\vp) \geq
  \int_{x \in \Eval \sigma} h (x) d\nu$ for all $C \R_\sigma^\perp h$;
\item $C \R_\sigma^\perp h$ if and only if
  $\Prob (C \cdot \retkw M\vp) \geq h (V)$ for all $M \R_\sigma V$;
\item $M \R_{\F{} \sigma} Q$ if and only if for all $C \R_\sigma^* h$,
  $\Prob (C \cdot M\vp) \geq \extd h (Q)$; here $h$ is any continuous
  map from $\Eval \sigma$ to the continuous complete lattice $[0, 1]$,
  so $\extd h$ makes sense: $\extd h (\bot) = 0$, and if
  $Q \neq \bot$, then
  $\extd h (Q) = \bigwedge_{a \in Q} h (a)$---which is equal to $1$ if
  $Q=\emptyset$;
\item $C \R_\sigma^* h$ if and only if $\Prob (C \cdot \produce M\vp)
  \geq h (V)$ for all $M \R_\sigma V$;
\item $M \R_{\sigma \to \underline\tau} f$ if and only if $MN
  \R_{\underline\tau} f (V)$ for all $N \R_\sigma V$.
\end{itemize}

\begin{lemma}
  \label{lemma:R:reduc}
  For all ground terms $M, N \colon \anytype$, if
  $M \precsim_\anytype N$ 
  and $M \R_\anytype V$ then $N \R_\anytype V$.
\end{lemma}
\begin{proof}
  By induction on $\anytype$.  If $\anytype = \U\underline\tau$, then
  $M \R_\anytype V$ means that $\force M \R_{\underline\tau} V$.  For
  every ground context $C \colon \underline\tau \vdash \F\Vt\unitT$,
  $\Prob (C \cdot \force M \vp) = \Prob (C [\force \_] \cdot M \vp)$
  and
  $\Prob (C \cdot \force N \vp) = \Prob (C [\force \_] \cdot N \vp)$
  by Lemma~\ref{lemma:prob:disc}.  Since
  $M \precsim_{\U\underline\tau} N$,
  $\Prob (C [\force \_] \cdot M \vp) \leq \Prob (C [\force \_] \cdot N
  \vp)$, so
  $\Prob (C \cdot \force M \vp) \leq \Prob (C \cdot \force N \vp)$.
  It follows that $\force M \precsim_{\underline\tau} \force N$.  By
  induction hypothesis, $\force N \R_{\underline\tau} V$, whence
  $N \R_{\U\underline\tau} V$.
  
  If $\anytype = \unitT$, then $M \R_\anytype V$ means that $V=\bot$,
  or that $V=\top$ and $V \precsim_{\unitT} M$.  In the first case,
  $N \R_\anytype V$ holds vacuously.  In the second case,
  $V \precsim_{\unitT} M \precsim_{\unitT} N$, so $N \R_\anytype V$
  again.  The case $\anytype = \intT$ is dealt with similarly---in the
  second case, $V = n \in \nat$ and $\underline *$ has to be replaced
  by $\underline n$.

  If $\anytype = \sigma \times \tau$, then $M \R_\anytype V$ means
  that $\pi_1 M \R_\sigma V_1$ and $\pi_2 M \R_\tau V_2$, where
  $V = (V_1, V_2)$.  We note that for every ground context
  $C \colon \sigma \to \F\Vt\unitT$,
  $\Prob (C \cdot \pi_1 M \vp) = \Prob (C [\pi_1 \_] \cdot M \vp)$ by
  Lemma~\ref{lemma:prob:disc}.  In turn,
  $\Prob (C [\pi_1 \_] \cdot M \vp) \leq \Prob (C [\pi_1 \_] \cdot N
  \vp) = \Prob (C \cdot \pi_1 N \vp)$ since
  $M \precsim_{\sigma \times \tau}$ and using
  Lemma~\ref{lemma:prob:disc} again.  Therefore
  $\pi_1 M \precsim_\sigma \pi_1 N$.  By induction hypothesis,
  $\pi_1 N \R_\sigma V_1$.  Similarly, $\pi_2 N \R_\tau V_2$, so
  $N \R_{\sigma \times \tau} (V_1, V_2) = V$.
  
  If $\anytype = \Vt{} \sigma$, then $M \R_\anytype V$ means that
  $V = \nu$ for some $\nu \in \Val_{\leq 1} (\Eval \sigma)$, and that
  $\Prob (C \cdot M\vp) \geq \int_{x \in \Eval \sigma} h (x)d\nu$ for
  all $C \R_\sigma^\perp h$.  For all such $C$ and $h$,
  $\Prob (C \cdot N \vp)$ is even larger, so $N \R_\anytype V$.

  If $\anytype = \F{}\sigma$, then $M \R_\anytype V$ means that for all
  $C \R_\sigma^* h$, $\Prob (C \cdot M \vp) \geq \extd h (V)$.  Then
  $\Prob (C \cdot N \vp)$ is even larger, so $N \R_\anytype V$.

  If $\anytype = \sigma \to \underline\tau$, then $M \R_\anytype V$
  means that $V$ is some function
  $f \in [\Eval \sigma \to \Eval {\underline\tau}]$, and that for all
  $P \R_\sigma V'$, $M P \R_{\underline\tau} f (V')$.  For every
  ground context $C \colon \underline\tau \vdash \F\Vt\unitT$,
  $\Prob (C \cdot NP \vp) = \Prob (C [\_ P] \cdot N \vp) \geq \Prob (C
  [\_ P] \cdot M \vp) = \Prob (C \cdot MP \vp)$, by
  Lemma~\ref{lemma:prob:disc}, the assumption
  $M \precsim_{\sigma \to \underline\tau} N$, and
  Lemma~\ref{lemma:prob:disc} again.  Hence
  $MP \precsim_{\underline\tau} NP$.  By induction hypothesis,
  $N P \R_{\underline\tau} f (V')$.  It follows that
  $N \R_{\sigma \to \underline\tau} f=V$.  \qed
\end{proof}

For each type $\anytype$, and every ground term $M \colon \anytype$,
let us write $M \R_\anytype$ for the set of values
$a \in \Eval \anytype$ such that $M \R_\anytype a$.
\begin{lemma}
  \label{lemma:MRA}
  For every type $\anytype$, for every ground term
  $M \colon \anytype$, $M \R_\anytype$ is Scott-closed and contains
  $\bot$.
\end{lemma}
\begin{proof}
  This is an easy induction on types.  Only the cases
  $\anytype=\Vt\sigma$ and $\anytype=\F{}\sigma$ need some care.  In
  the case $\anytype=\Vt\sigma$, $M \R_{\Vt\sigma}$ is Scott-closed
  because integration is Scott-continuous in the valuation.
  Explicitly, it is upwards-closed, and for every directed family
  ${(\nu_i)}_{i \in I}$ in $M \R_{\Vt\sigma}$, with supremum $\nu$,
  for all $C \R_\sigma^\perp h$,
  $\Prob (C \cdot M \vp) \geq \int_{x \in \Eval \sigma} h (x) d\nu_i$
  for every $i \in I$, so
  $\Prob (C \cdot M \vp) \geq \sup_{i \in I} \int_{x \in \Eval \sigma}
  h (x) d\nu_i = \int_{x \in \Eval \sigma} h (x) d\nu$.  In order to
  show that $M \R_{\Vt\sigma}$ contains $\bot$ (the zero valuation),
  we must show that
  $\Prob (C \cdot M\vp) \geq \int_{x \in \Eval \sigma} h (x) d0 = 0$
  for all $C \R_\sigma^\perp h$, and that is trivial.
  
  In the case $\anytype=\F{}\sigma$,
  let us fix $C$ and $h$ so that $C \R_\sigma^* h$.  Since $\extd h$
  is continuous by Proposition~\ref{prop:Qtop}, item~2,
  $\{Q \in \Smyth^\top_\bot (\Eval \sigma) \mid \extd h (Q) > r\} =
  {(\extd h)}^{-1} ((r, \infty])$ is open for every $r \in \Rp$. By
  taking complements, the set
  $F_{C, h} = \{Q \in \Smyth^\top_\bot (\Eval \sigma) \mid \extd h (Q)
  \leq \Prob (C \cdot M\vp)\}$ is closed.  Hence
  $M \R_{\F{} \sigma} = \bigcap_{C \R_\sigma^* h} F_{C, h}$ is closed in
  $\Smyth^\top_\bot (\Eval \sigma)$.  Finally, $M \R_{\F{} \sigma}$
  contains $\bot$ because $\extd h$ is strict, hence
  $\extd h (\bot) = 0$.  \qed
\end{proof}

Let us say that a term $M$ \emph{has $x_\sigma$ as sole free variable}
if and only if the set of free variables of $M$ is included in
$\{x_\sigma\}$, namely if $M$ is ground or if the only free variable
of $M$ is $x_\sigma$.  In that case, for every ground term $N$,
$M [x_\sigma := N]$ is ground.
\begin{corollary}
  \label{corl:R:rec}
  Let $M \colon \sigma$ have $x_\sigma$ as sole free variable, let
  $f$ be a Scott-continuous map from $\Eval \sigma$ to $\Eval \sigma$,
  and assume that for all $N \R_\sigma V$,
  $M [x_\sigma := N] \R_\sigma f (V)$.  Then
  $\rec x_\sigma . M \R_\sigma \lfp f$.
\end{corollary}
\begin{proof}
  We show that $\rec x_\sigma . M \R_\sigma f^n (\bot)$ for every
  $n \in \nat$.  Since $\rec x_\sigma . M \R_\sigma$ contains $\bot$
  by Lemma~\ref{lemma:MRA}, this is true when $n=0$.  If
  $\rec x_\sigma . M \R_\sigma f^n (\bot)$, then
  $M [x_\sigma := \rec x_\sigma . M] \R_\sigma f^{n+1} (\bot)$ by
  assumption.  We now use the fact that for every ground context
  $C \colon \sigma \to \F\Vt\unitT$,
  $C \cdot \rec x_\sigma . M \to C \cdot M [x_\sigma := \rec x_\sigma
  . M]$, hence
  $\Prob (C \cdot \rec x_\sigma . M \vp) \geq \Prob (C \cdot M
  [x_\sigma := \rec x_\sigma . M] \vp)$, by Lemma~\ref{lemma:red}.
  Using Lemma~\ref{lemma:R:reduc}, we obtain that
  $\rec x_\sigma . M \R_\sigma f^{n+1} (\bot)$.

  Since $\rec x_\sigma . M \R_\sigma f^n (\bot)$ for every
  $n \in \nat$ and since $\rec x_\sigma .  M \R_\sigma$ is
  Scott-closed (Lemma~\ref{lemma:MRA}),
  $\rec x_\sigma . M \R_\sigma \lfp f$.  \qed
\end{proof}

\begin{lemma}
  \label{lemma:R:ret}
  Let $\sigma$ be a value type.  For all $M \R_\sigma V$,
  $\retkw M \R_{\Vt{} \sigma} \delta_V$.
\end{lemma}
\begin{proof}
  Let $C$ be a ground context of type $\Vt{} \sigma \vdash \F\Vt\unitT$, $h$
  be a continuous map from $\Eval \sigma$ to $[0, 1]$, and assume that
  $C \R_\sigma^\perp h$.  By definition of $\R_\sigma^\perp$, and
  since $M \R_\sigma V$, $\Prob (C \cdot \retkw M \vp) \geq h (V)$,
  and $h (V) = \int_{x \in \Eval \sigma} h (x) d\delta_V$.
  \qed
\end{proof}

\begin{lemma}
  \label{lemma:R:bind}
  Let $\sigma$ and $\tau$ be two value types.  Let $N \colon \Vt{} \tau$
  be a term with $x_\sigma$ as sole free variable,
  $f \in [\Eval \sigma \to \Eval {\Vt\tau}]$, and assume that for all
  $P \R_\sigma V$, $N [x_\sigma := P] \R_{\Vt\tau} f (V)$.  For all
  $M \R_{\Vt\sigma} \nu$,
  $\dokw {x_\sigma \leftarrow M}; N \R_{\Vt\tau} f^\dagger (\nu)$.
\end{lemma}
\begin{proof}
  Let $C \colon \Vt{} \tau \to \F\Vt\unitT$ be a ground context, and
  $h$ be a Scott-continuous map from $\Eval \tau$ to $[0, 1]$ such
  that $C \R_\tau^\perp h$.  We wish to show that
  $\Prob (C \cdot \dokw {x_\sigma \leftarrow M}; N \vp) \geq \int_{y
    \in \Eval \tau} h (y) df^\dagger (\nu)$, namely that
  $\Prob (C \cdot \dokw {x_\sigma \leftarrow M}; N \vp) \geq \int_{x
    \in \Eval \sigma} \allowbreak (\int_{y \in \Eval \tau} h (y) df
  (x)) d\nu$, using (\ref{eq:dagger}).

  We first show that
  $C [\dokw {x_\sigma \leftarrow \_}; N] \R_\sigma^\perp g$, where
  $g (x) = \int_{y \in \Eval \tau} h (y) d f (x)$ for every
  $x \in \Eval \sigma$.  That reduces to showing that
  $\Prob (C [\dokw {x_\sigma \leftarrow \_}; N] \cdot \retkw P\vp)
  \geq g (x)$ for all $P \R_\sigma x$.  Now
  $C [\dokw {x_\sigma \leftarrow \_}; N] \cdot \retkw P \to C \cdot N
  [x_\sigma := P]$, so
  $\Prob (C [\dokw {x_\sigma \leftarrow \_}; N] \cdot \retkw P \vp)
  \geq \Prob (C \cdot N [x_\sigma := P] \vp)$, by
  Lemma~\ref{lemma:red}, and
  $\Prob (C \cdot N [x_\sigma := P] \vp) \geq g (x) = \int_{y \in
    \Eval \tau} h (y) d f (x)$ because $C \R_\tau^\perp h$ and
  $N [x_\sigma := P] \R_{\Vt{} \sigma} f (x)$ for all $P \R_\sigma x$.

  Using this together with the fact that $M \R_{\Vt\sigma} \nu$,
  $\Prob (C [\dokw {x_\sigma \leftarrow \_}; N] \cdot M \vp) \geq
  \int_{x \in \Eval \sigma} g (x) d\nu$.  Since
  $C \cdot (\dokw {x_\sigma \leftarrow M}; N) \to C [\dokw {x_\sigma
    \leftarrow \_}; N] \cdot M$, by Lemma~\ref{lemma:red},
  $\Prob (C \cdot \dokw {x_\sigma \leftarrow M}; N \vp) \geq \int_{x
    \in \Eval \sigma} g (x) d\nu$.  \qed
\end{proof}

\begin{lemma}
  \label{lemma:R:produce}
  Let $\sigma$ be a value type.  For all $M \R_\sigma V$, $\produce M
  \R_{\F{} \sigma} \eta^\Smyth (V)$.
\end{lemma}
\begin{proof}
  Let $Q = \eta^\Smyth (V)$.  Let $C$ be a ground context of type
  $\F{} \sigma \vdash \F\Vt\unitT$, $h$ be a continuous map from
  $\Eval \sigma$ to $[0, 1]$, and assume that $C \R_\sigma^* h$.  By
  definition of $\R_\sigma^*$,
  $\Prob (C \cdot \produce M \vp) \geq h (V)$.  Since $V \in Q$,
  $h (V) \geq \bigwedge_{x \in Q} h (x) = \extd h (Q)$, so
  $\produce M \R_{\F{} \sigma} Q$.  \qed
\end{proof}

\begin{lemma}
  \label{lemma:R:pto:F}
  Let $\sigma$, $\tau$ be two value types.  Let $N \colon \F{} \tau$ be a
  term with $x_\sigma$ as sole free variable,
  $f \in [\Eval \sigma \to \Eval {\F{} \tau}]$, and assume that for all
  $P \R_\sigma V$, $N [x_\sigma := P] \R_{\F{}\tau} f (V)$.  For all
  $M \R_{\F{}\sigma} Q$, $\pto M {x_\sigma} N \R_{\F{}\tau} \extd f (Q)$.
\end{lemma}
\begin{proof}
  Let $C \colon \F{}\tau \vdash \F\Vt\unitT$ be a ground context, and $h$ be
  a Scott-continuous map from $\Eval \tau$ to $[0, 1]$ such that
  $C \R_\tau^* h$.  We wish to show that
  $\Prob (C \cdot \pto M {x_\sigma} N \vp) \geq \extd h (\extd f
  (Q))$.  Using Proposition~\ref{prop:Qtop}, we will show the
  equivalent claim that
  $\Prob (C \cdot \pto M {x_\sigma} N \vp) \geq \extd {(\extd h \circ
    f)} (Q)$.

  We first show that
  $C [\pto {\_} {x_\sigma} N] \R_\sigma^* \extd h \circ f$.  For all
  $P \R_\sigma V$, we aim to show that
  $\Prob (C [\pto {\_} {x_\sigma} N] \cdot \produce P \vp) \geq \extd
  h (f (V))$.  Since $N [x_\sigma := P] \R_{\F{}\tau} f (V)$ and
  $C \R_\tau^* h$,
  $\Prob (C \cdot N [x_\sigma := P] \vp) \geq \extd h (f (V))$.  Since
  $C [\pto {\_} {x_\sigma} N] \cdot \produce P \to C \cdot N [x_\sigma
  := P]$, by Lemma~\ref{lemma:red}
  $\Prob (C [\pto {\_} {x_\sigma} N] \cdot \produce P \vp) \geq \extd
  h (f (V))$, as desired.

  Knowing that
  $C [\pto {\_} {x_\sigma} N] \R_\sigma^* \extd h \circ f$, and using
  $M \R_{\F{}\sigma} Q$, we obtain that
  $\Prob (C [\pto {\_} {x_\sigma} N] \cdot M \vp) \geq \extd {(\extd h
    \circ f)} (Q)$.  Since
  $C \cdot \pto M {x_\sigma} N \to C [\pto {\_} {x_\sigma} N] \cdot
  M$, by Lemma~\ref{lemma:red}
  $\Prob (C \cdot \pto M {x_\sigma} N \vp) \geq \extd {(\extd h \circ
    f)} (Q)$.  \qed
\end{proof}

We write $\chi_U \colon X \to \Sierp$ for the characteristic map of an
open subset $U$ of a space $X$.
\begin{lemma}
  \label{lemma:R:FVunit}
  \begin{enumerate}
  \item $[\produce \_] \R_{\unitT}^\perp \chi_{\{\top\}}$;
  \item
    $[\_] \R_{\Vt{} \unitT}^* (\nu \in \Val_{\leq 1} \Sierp \mapsto \nu
    (\{\top\}))$.
  \end{enumerate}
\end{lemma}
\begin{proof}
  1. Let $M \R_{\unitT} V$.  It suffices to show that
  $\Prob ([\produce \_] \cdot \retkw M \vp) \geq \chi_{\{\top\}} (V)$.
  If $V=\bot$, then the right-hand side is $0$, and the inequality is
  clear.  Otherwise, we claim that the left-hand side is (greater than
  or) equal to $1$.  We have
  $[\produce \_] \cdot \retkw M \to [\produce \retkw \_] \cdot M$, so
  $\Prob ([\produce \_] \cdot \retkw M \vp) \geq \Prob ([\produce
  \retkw \_] \cdot M \vp)$ by Lemma~\ref{lemma:red}.  Since
  $M \R_{\unitT} \top$,
  $\Prob ([\produce \retkw \_] \cdot M \vp) \geq \Prob ([\produce
  \retkw \_] \cdot \underline * \vp)$, and
  $\Prob ([\produce \retkw \_] \cdot \underline * \vp) = 1$ since we
  can deduce $[\produce \retkw \_] \cdot \underline * \vp a$ for every
  $a \in \rat \cap [0, 1)$.

  2.  Let $M \R_{\Vt{} \unitT} \nu$.  It suffices to show that
  $\Prob ([\_] \cdot \produce M \vp) \geq \nu (\{\top\})$.  Since
  $[\produce \_] \R_{\unitT}^\perp \chi_{\{\top\}}$ by item~1,
  $\Prob ([\produce \_] \cdot M \vp) \geq \int_{x \in \Eval {\unitT}}
  \chi_{\{\top\}} (x) \allowbreak d\nu = \nu (\{\top\})$.  We use
  Lemma~\ref{lemma:red} together with
  $[\_] \cdot \produce M \to [\produce \_] \cdot M$ and we obtain the
  desired inequality.  \qed
\end{proof}

A \emph{substitution} $\theta = [x_1:=N_1, \cdots, x_n:=N_n]$ is a map
of finite domain $\dom \theta = \{x_1, \cdots, x_n\}$ from pairwise
distinct variables $x_i$ to ground terms $N_i$ of the same type as
$x_i$.  We omit the definition of (parallel) substitution application
$M \theta$.  The case $M [x_\sigma:=N]$ is the special case $n=1$.  We
note that
$M [x_1:=N_1, \cdots, x_n:=N_n] [x:=N] = M [x_1:=N_1, \cdots,
x_n:=N_n, x:=N]$ when $x$ is distinct from $x_1$, \ldots, $x_n$, and
not free in $N_1$, \ldots, $N_n$.  Also, if $\dom \theta$ contains all
the free variables of $M$, then $M \theta$ is ground.

We define the relation $\R^\bullet$ between substitutions and
environments by $\theta \R^\bullet \rho$ if and only if for every
$x_\sigma \in \dom \theta$, $x_\sigma \theta \R_\sigma \rho
(x_\sigma)$.
\begin{proposition}
  \label{prop:R}
  For every type $\anytype$, for every term $M \colon \anytype$ of
  CBPV$(\Demon,\Nature)$ or any of its extensions with $\pifzkw$ or
  $\bigcirc$ or both, for every substitution $\theta$ whose domain
  contains all the free variables of $M$, for every environment
  $\rho$, if $\theta \R^\bullet \rho$ then
  $M \theta \R_\anytype \Eval M \rho$.
\end{proposition}
\begin{proof}
  By induction on $M$.  This is the assumption $\theta \R^\bullet
  \rho$ when $M$ is a variable.

  In the case of $\lambda$-abstractions
  $\lambda x_\sigma . M \colon \sigma \to \underline\tau$, let us
  write $\theta$ as $[x_1:=N_1, \cdots, x_n:=N_n]$, and assume by
  $\alpha$-renaming that $x_\sigma$ is different from every $x_i$ and
  free in no $N_i$.  For all $N \R_\sigma V$, we define $\theta'$ as
  $[x_1:=N_1, \cdots, x_n:=N_n, x_\sigma:=N]$ and we observe that
  $\theta' \R^\bullet \rho [x_\sigma \mapsto V]$, so by induction
  hypothesis
  $M \theta' \R_{\underline\tau} \Eval M \rho [x_\sigma \mapsto V]$.
  Hence
  $M \theta [x_\sigma:=N] \R_{\underline\tau} \Eval M \rho [x_\sigma
  \mapsto V]$.  By Lemma~\ref{lemma:red} and
  Lemma~\ref{lemma:R:reduc}, using the fact that
  $C \cdot (\lambda x_\sigma . M \theta) N \to C [\_ N] \cdot \lambda
  x_\sigma . M \theta \to C \cdot M \theta [x_\sigma:=N]$ for all
  ground contexts $C \colon \underline\tau \vdash \F\Vt\unitT$, we obtain
  that
  $(\lambda x_\sigma . M \theta) N \R_{\underline\tau} \Eval M \rho
  [x_\sigma \mapsto V]$.  Since that holds for all $N \R_\sigma V$,
  $(\lambda x_\sigma . M) \theta = \lambda x_\sigma . M \theta
  \R_{\sigma \to \underline\tau} \Eval {\lambda x_\sigma . M} \rho$.
  
  The case of applications is by definition of $\R_{\sigma \to
    \underline\tau}$.

  For terms of the form $\produce M \colon \F{} \sigma$, by assumption $M
  \theta \R_\sigma \Eval M \rho$.  By Lemma~\ref{lemma:R:produce},
  $\produce M \theta \R_{\F{} \sigma} \eta^\Smyth (\Eval M \rho) = \Eval
  {\produce M} \rho$.

  For terms of the form $\pto M {x_\sigma} N$ where
  $M \colon \F{} \sigma$ and $N \colon \F{} \tau$, as for
  $\lambda$-abstractions we write $\theta$ as
  $[x_1:=N_1, \cdots, x_n:=N_n]$, and we assume by that $x_\sigma$ is
  different from every $x_i$ and free in no $N_i$.  By induction
  hypothesis $M \theta \R_{\F{} \sigma} \Eval M \rho$, and for all
  $P \R_\sigma V$, $N \theta' \R_{\F{} \tau} \Eval N \rho [x_\sigma:=V]$
  where $\theta' = [x_1:=N_1, \cdots, x_n:=N_n, x_\sigma:=P]$.  As for
  $\lambda$-abstractions, the latter means that
  $N \theta [x_\sigma:=P] \R_{\F{} \tau} \Eval N \rho [x_\sigma:=V]$.
  Letting $f$ be the map
  $V \in \Eval \sigma \mapsto \Eval N \rho [x_\sigma:=V]$, therefore,
  $N \theta [x_\sigma:=P] \R_{\F{} \tau} f (V)$ for all $P \R_\sigma V$.
  By Lemma~\ref{lemma:R:pto:F},
  $(\pto M {x_\sigma} N) \theta = \pto {M \theta} {x_\sigma} {N
    \theta} \R_{\F{} \tau} \extd f (\Eval M \rho) = \Eval {\pto M
    {x_\sigma} N} \rho$.

  For terms $\thunk M \colon \U\underline\tau$, by induction
  hypothesis $M \theta \R_{\underline\tau} \Eval M \rho$.  For every
  ground context $C \colon \underline\tau \vdash \F\Vt\unitT$,
  $C \cdot \force \thunk M \theta \to C [\force \_] \cdot \thunk M
  \theta \to C \cdot M \theta$, so by Lemma~\ref{lemma:red} and
  Lemma~\ref{lemma:R:reduc},
  $\force \thunk M \R_{\underline\tau} \Eval M \rho$.  By definition
  of $\R_{\U\underline\tau}$,
  $\thunk M \theta \R_{\U\underline\tau} \Eval M \rho = \Eval {\thunk
    M} \rho$.

  The case of terms of the form $\force M \colon \underline\tau$ is by
  definition of $\R_{\U\underline\tau}$.

  In the case of $\underline *$, we have $\underline * \R_{\unitT}
  \top$ by definition.  Similarly, $\underline n \R_{\intT} n$.

  For terms $\suc M$ with $M \colon \intT$, by induction hypothesis
  $M \theta \R_{\intT} \Eval M \rho$.  If $\Eval M \rho = \bot$, then
  $\suc M \theta \R_{\intT} \bot = \Eval {\suc M} \rho$.  Otherwise,
  let $n = \Eval M \rho \in \Z$.  By definition,
  $\Prob (C \cdot M \theta \vp) \geq \Prob (C \cdot \underline n \vp)$
  for every ground context $C \colon \intT \vdash \F\Vt\unitT$.
  Replacing $C$ by $C [\suc \_]$,
  $\Prob (C [\suc \_] \cdot M \theta \vp) \geq \Prob (C [\suc \_]
  \cdot \underline n \vp)$.  Since
  $C \cdot \suc M \theta \to C [\suc \_] \cdot M$ and since
  $C [\suc \_] \cdot \underline n \to C \cdot \underline {n+1}$, using
  Lemma~\ref{lemma:red} we obtain
  $\Prob (C \cdot \suc M \theta \vp) \geq \Prob (C [\suc \_] \cdot M
  \vp) \geq \Prob (C [\suc \_] \cdot \underline n \vp) \geq \Prob (C
  \cdot \underline {n+1} \vp)$.  This shows that
  $\suc M \theta \R_{\intT} \underline {n+1} = \Eval {\suc M} \rho$.
  The case of terms $\pred M$ is similar.

  For terms $\ifz M N P \colon \anytype$, by induction hypothesis
  $M \theta \R_{\intT} \Eval M \rho$, $N \theta \R_\anytype \Eval N \rho$,
  and $P \theta \R_\anytype \Eval P \rho$.  If $\Eval M \rho = \bot$, then
  $(\ifz M N P) \theta \R_\anytype \bot = \Eval {\ifz M N P} \rho$, by
  Lemma~\ref{lemma:MRA}.  Otherwise, let $n = \Eval M \rho \in \Z$.
  Since $M \theta \R_{\intT} \Eval M \rho$,
  $\Prob (C \cdot M \theta \vp) \geq \Prob (C \cdot \underline n \vp)$
  for every ground context $C \colon \intT \vdash \F\Vt\unitT$.  In
  particular, for every ground context $C \colon \anytype \vdash \F\Vt\unitT$,
  $\Prob (C [\ifz \_ {N \theta} {P \theta}] \cdot M \theta \vp) \geq
  \Prob (C [\ifz \_ {N \theta} {P \theta}] \cdot \underline n \vp)$.
  Using Lemma~\ref{lemma:prob:disc}, it follows that
  $\Prob (C \cdot \ifz {M \theta} {N \theta} {P \theta} \vp) \geq
  \Prob (C [\ifz \_ {N \theta} {P \theta}] \cdot \underline n \vp)$.
  Since $C [\ifz \_ {N \theta} {P \theta}] \cdot \underline n$ reduces
  to $C \cdot N \theta$ if $n=0$, and to $C \cdot P \theta$ if
  $n \neq 0$, by Lemma~\ref{lemma:red},
  $\Prob (C \cdot \ifz {M \theta} {N \theta} {P \theta} \vp)$ is
  larger than or equal to $\Prob (C \cdot N \theta \vp)$ if $n=0$, and
  to $\Prob (C \cdot P \theta \vp)$ otherwise.  By
  Lemma~\ref{lemma:R:reduc},
  $\ifz {M \theta} {N \theta} {P \theta} \R_\anytype \Eval N \rho$ if $n=0$,
  and $\ifz {M \theta} {N \theta} {P \theta} \R_\anytype \Eval P \rho$ if
  $n \neq 0$.  In any case,
  $(\ifz M N P) \theta = \ifz {M \theta} {N \theta} {P \theta} \R_\anytype
  \Eval {\ifz M N P} \rho$.

  The case of terms $M; N \colon \anytype$ is similar.  By induction
  hypothesis, $M \theta \R_{\unitT} \Eval M \rho$ and
  $N \theta \R_\anytype \Eval N \rho$.  If $\Eval M \rho=\bot$, then
  $\Eval {M; N} \rho = \bot$, so
  $(M; N) \theta \R_\anytype \Eval {M; N} \rho$ by Lemma~\ref{lemma:MRA}.
  Otherwise, $\Eval M \rho=\top$, so $M \theta \R_{\unitT} \top$,
  meaning that
  $\Prob (C \cdot M \theta \vp) \geq \Prob (C \cdot \underline * \vp)$
  for every ground context $C \colon \unitT \vdash \F\Vt\unitT$.  In
  particular, for every ground context $C \colon \anytype \vdash \F\Vt\unitT$,
  $\Prob (C [\_; N \theta] \cdot M \theta \vp) \geq \Prob (C [\_; N
  \theta] \cdot \underline * \vp)$.  By Lemma~\ref{lemma:prob:disc},
  $\Prob (C \cdot M \theta;N \theta \vp) = \Prob (C [\_; N \theta]
  \cdot M \theta \vp)$, and by Lemma~\ref{lemma:red},
  $\Prob (C [\_; N \theta] \cdot \underline * \vp) \geq \Prob (C \cdot
  N \theta \vp)$, using the rule
  $C [\_; N \theta] \cdot \underline * \to C \cdot N \theta$.  Hence
  $\Prob (C \cdot M \theta;N \theta \vp) \geq \Prob (C \cdot N \theta
  \vp)$.  By Lemma~\ref{lemma:R:reduc},
  $(M; N) \theta = M \theta; N \theta \R_\anytype \Eval N \rho = \Eval {M; N}
  \rho$.

  The case of terms $\pi_1 M$ and $\pi_2 M$ follows from the
  definition of $\R_{\sigma \times \tau}$.

  For terms $\langle M, N \rangle \colon \sigma \times \tau$, by
  induction hypothesis $M \theta \R_\sigma \Eval M \rho$ and
  $N \theta \R_\tau \Eval N \rho$.  For every ground context
  $C \colon \sigma \to \F\Vt\unitT$,
  $C \cdot \pi_1 \langle M \theta, N \theta \rangle \to C [\pi_1 \_]
  \cdot \langle M \theta, N \theta \rangle \to C \cdot M \theta$, so
  by Lemma~\ref{lemma:red} and Lemma~\ref{lemma:R:reduc},
  $\pi_1 \langle M, N \rangle \theta = \pi_1 \langle M \theta,
  \allowbreak N \theta \rangle \R_\sigma \Eval M \rho$.  Similarly,
  $\pi_2 \langle M, N \rangle \theta \R_\tau \Eval N \rho$.  By
  definition of $\R_{\sigma \times \tau}$, it follows that
  $\langle M, N \rangle \theta \R_{\sigma \times \tau} \Eval {\langle
    M, N \rangle} \rho$.

  For terms $\retkw M \colon \Vt{} \tau$, by induction hypothesis
  $M \theta \R_\tau \Eval M \rho$, so
  $\retkw M \theta \R_{\Vt{} \tau} \delta_{\Eval M \rho} = \Eval {\retkw
    M} \rho$ by Lemma~\ref{lemma:R:ret}.

  For terms $\dokw {x_\sigma \leftarrow M}; N$ where
  $M \colon \Vt{} \sigma$ and $N \colon \Vt{} \tau$, as for
  $\lambda$-abstractions we write $\theta$ as
  $[x_1:=N_1, \cdots, x_n:=N_n]$, and we assume that $x_\sigma$ is
  different from every $x_i$ and free in no $N_i$.  By induction
  hypothesis $M \theta \R_{\Vt{} \sigma} \Eval M \rho$, and for all
  $P \R_\sigma V$, $N \theta' \R_{\Vt{} \tau} \Eval N \rho [x_\sigma:=V]$
  where $\theta' = [x_1:=N_1, \cdots, x_n:=N_n, x_\sigma:=P]$.  As for
  $\lambda$-abstractions, the latter means that
  $N \theta [x_\sigma:=P] \R_{\Vt{} \tau} \Eval N \rho [x_\sigma:=V]$.
  Letting $f$ be the map
  $V \in \Eval \sigma \mapsto \Eval N \rho [x_\sigma:=V]$, therefore,
  $N \theta [x_\sigma:=P] \R_{\Vt{} \tau} f (V)$ for all $P \R_\sigma V$.
  By Lemma~\ref{lemma:R:bind},
  $(\dokw {x_\sigma \leftarrow M}; N) \theta = \dokw {x_\sigma
    \leftarrow M \theta}; (N \theta) \R_{\Vt{} \tau} f^\dagger (\Eval M
  \rho) = \Eval {\dokw {x_\sigma \leftarrow M}; N} \rho$.

  For terms $M \oplus N \colon \Vt{} \tau$, by induction hypothesis
  $M \theta \R_\tau \Eval M \rho$ and $N \theta \R_\tau \Eval N \rho$.
  For all $C \R_\tau^\perp h$,
  $\Prob (C \cdot M \theta \vp) \geq \int_{x \in \Eval \tau} h (x) d
  \Eval M \rho$, and
  $\Prob (C \cdot N \theta \vp) \geq \int_{x \in \Eval \tau} h (x) d
  \Eval N \rho$.  For all $a$ and $b$, if we can deduce
  $C \cdot M \theta \vp a$ and $C \cdot N \theta \vp b$, then we can
  deduce $C \cdot (M \oplus N) \theta \vp (a+b)/2$.  Therefore
  $\Prob (C \cdot (M \oplus N) \theta \vp) \geq \frac 1 2 (\Prob (C
  \cdot M \theta \vp) + \Prob (C \cdot N \theta \vp)) \geq \frac 1 2
  (\int_{x \in \Eval \tau} h (x) d \Eval M \rho + \int_{x \in \Eval
    \tau} h (x) d \Eval N \rho) = \int_{x \in \Eval \tau} h (x) d
  \Eval {M \oplus N} \rho$.  Hence
  $(M \oplus N) \theta \R_\tau \Eval {M \oplus N} \rho$.

  The case of terms $M \owedge N \colon \F{} \tau$ is similar, using the
  fact that
  $\Prob (C \cdot (M \owedge N) \theta \vp) \geq \min (\Prob (C \cdot M
  \theta \vp), \Prob (C \cdot N \theta \vp))$ instead.  The latter
  follows from the fact that if we can deduce both
  $C \cdot M \theta \vp a$ and $C \cdot N \theta \vp a$, then we can
  deduce $C \cdot (M \owedge N) \theta \vp a$.  By induction hypothesis,
  $M \theta \R_{\F{} \tau} \Eval M \rho$ and
  $N \theta \R_{\F{} \tau} \Eval N \rho$.  For all $C \R_\tau^* h$,
  $\Prob (C \cdot M \theta \vp) \geq \extd h (\Eval M \rho)$ and
  $\Prob (C \cdot N \theta \vp) \geq \extd h (\Eval N \rho)$, so
  $\Prob (C \cdot (M \owedge N) \theta \vp) \geq \min (\extd h (\Eval M
  \rho), \extd h (\Eval N \rho)) = \extd h (\Eval M \rho) \wedge \extd
  h (\Eval N \rho) = \extd h (\Eval M \rho \wedge \Eval N \rho)$
  (because $\extd h$ preserves binary infima, see
  Proposition~\ref{prop:Qtop}, item~3)
  $= \extd h (\Eval {M \owedge N} \rho)$.  Hence
  $(M \owedge N) \theta \R_{\F{} \tau} \Eval {M \owedge N} \rho$.

  For $\abort_{\F{} \tau} \colon \F{} \tau$, we show that
  $\abort_{\F{} \tau} \R_{\F{} \tau} \Eval {\abort_{\F{} \tau}} \rho =
  \emptyset$ ($\neq \bot$) by showing that for all $C \R_\tau^* h$,
  $\Prob (C \cdot \abort_{\F{} \tau} \vp) \geq \extd h (\emptyset)$.
  Indeed, by the rule $C \cdot \abort_{\F{} \tau} \vp a$
  ($a \in \rat \cap [0, 1)$),
  $\Prob (C \cdot \abort_{\F{} \tau} \vp) = 1$.
  
  For $\rec x_\sigma . M$ where $M \colon \sigma$, as for
  $\lambda$-abstractions, let us write $\theta$ as
  $[x_1:=N_1, \cdots, x_n:=N_n]$, and assume by $\alpha$-renaming that
  $x_\sigma$ is different from every $x_i$ and free in no $N_i$.  For
  all $N \R_\sigma V$, we define $\theta'$ as
  $[x_1:=N_1, \cdots, x_n:=N_n, x_\sigma:=N]$ and we observe that
  $\theta' \R^\bullet \rho [x_\sigma \mapsto V]$, so by induction
  hypothesis $M \theta' \R_\sigma \Eval M \rho [x_\sigma \mapsto V]$.
  Let $f (V) = \Eval M \rho [x_\sigma \mapsto V]$.  We have just shown
  that $M \theta [x_\sigma := N] \R_\sigma f (V)$ for all
  $N \R_\sigma V$.  By Corollary~\ref{corl:R:rec},
  $(\rec x_\sigma . M) \theta = \rec x_\sigma . (M \theta) \R_\sigma
  \lfp (f) = \Eval {\rec x_\sigma . M} \rho$.

  We finish with the constructions involving $\bigcirc$ or $\pifzkw$.
  For $\bigcirc_{> b} M$ where $M \colon \F\Vt\unitT$, by induction
  hypothesis $M \theta \R_{\F\Vt\unitT} \Eval M \rho$.  Using
  Lemma~\ref{lemma:R:FVunit}, item~2, we obtain that
  $\Prob ([\_] \cdot M \theta \vp) \geq \extd {(\nu \in \Val_{\leq 1}
    \Sierp \mapsto \nu (\{\top\}))} (\Eval M \rho)$.  If
  $\Eval M \rho = \bot$, then $\Eval {\bigcirc_{> b} M} \rho = \bot$,
  so
  $\bigcirc_{> b} M \theta \R_{\unitT} \Eval {\bigcirc_{> b} M} \rho$,
  trivially.  Otherwise, $\Eval M \rho$ is a compact saturated subset
  of $\Val_{\leq 1} \Sierp$.  If $b \not\ll \nu (\{\top\})$ for some
  $\nu \in \Eval M \rho$, then again
  $\Eval {\bigcirc_{> b} M} \rho = \bot$, so
  $\bigcirc_{> b} M \theta \R_{\unitT} \Eval {\bigcirc_{> b} M} \rho$
  is again trivial.  Finally, if $b \ll \nu (\{\top\})$ for every
  $\nu \in \Eval M \rho$, then we verify that
  $b \ll \extd {(\nu \in \Val_{\leq 1} \Sierp \mapsto \nu (\{\top\}))}
  (\Eval M \rho)$: if $\Eval M \rho \neq \emptyset$,
  $\extd {(\nu \in \Val_{\leq 1} \Sierp \mapsto \nu (\{\top\}))}
  (\Eval M \rho) = \min_{\nu \in \Eval M \rho} \nu (\{\top\})$, so $b$
  is way-below that value; while if $\Eval M \rho = \emptyset$, then
  $\extd {(\nu \in \Val_{\leq 1} \Sierp \mapsto \nu (\{\top\}))}
  (\Eval M \rho) = 1$, and $b \ll 1$ because the $\bigcirc_{> b}$
  operator requires $b < 1$.  It follows that
  $b \ll \Prob ([\_] \cdot M \theta \vp)$, so there is a number $a$
  such that $b \leq a$ and $[\_] \cdot M \vp a$ is derivable.  By
  Lemma~\ref{lemma:prob:leq}, $[\_] \cdot M \vp b$ is derivable.  For
  every ground context $C \colon \unitT \vdash \F\Vt\unitT$, for every
  $a$ such that $C \cdot \underline * \vp a$ is derivable, the
  leftmost rule of the bottom row of Figure~\ref{fig:opsem} allows us
  to derive $C \cdot \bigcirc_{> b} M \vp a$, so
  $\Prob (C \cdot \bigcirc_{> b} M \vp) \geq \Prob (C \cdot \underline
  * \vp)$.  It follows that
  $\bigcirc_{> b} M \R_{\unitT} \top = \Eval {\bigcirc_{> b} M} \rho$.

  Finally, for terms of the form $\pifz M N P$, where $M \colon \intT$
  and $N, P \colon \F{} \tau$, we wish to show that
  $(\pifz M N P) \theta \R_{\F{} \tau} \Eval {\pifz M N P} \rho$.  This
  means showing that, for all $C \R_\tau^* h$,
  $\Prob (C \cdot \pifz {M \theta} {N \theta} {P \theta} \vp) \geq
  \extd h (\Eval {\pifz M N P} \rho)$.

  If $\Eval M \rho = \bot$, then
  $\Eval {\pifz M N P} \rho = \Eval N \rho \wedge \Eval P \rho$.  In
  that case, we note that
  $\Prob (C \cdot \pifz {M \theta} {N \theta} {P \theta} \vp)$ is
  larger than or equal to
  $\min (\Prob (C \cdot N \theta \vp), \Prob (C \cdot P \theta \vp))$:
  for every $a \in \rat \cap [0, 1)$ way-below
  $\min (\Prob (C \cdot N \theta \vp), \Prob (C \cdot P \theta \vp))$,
  we can derive $C \cdot N \theta \vp b$ for some $b \geq a$, and
  $C \cdot P \theta \vp c$ for some $c \geq a$; then, by
  Lemma~\ref{lemma:prob:leq}, we can derive $C \cdot N \theta \vp a$
  and $C \cdot P \theta \vp a$, hence
  $C \cdot \pifz {M \theta} {N \theta} {P \theta} \vp a$.  By
  induction hypothesis, $N \theta \R_{\F{} \tau} \Eval N \rho$, so
  $\Prob (C \cdot N \theta \vp) \geq \extd h (\Eval N \rho)$, and
  similarly
  $\Prob (C \cdot P \theta \vp) \geq \extd h (\Eval P \rho)$.
  Therefore
  $\Prob (C \cdot \pifz {M \theta} {N \theta} {P \theta} \vp) \geq
  \min (\extd h (\Eval N \rho), \extd h (\Eval P \rho)) = \extd h
  (\Eval N \rho \wedge \Eval P \rho) = \extd h (\Eval {\pifz M N P}
  \rho)$, since $\extd h$ preserves binary infima
  (Proposition~\ref{prop:Qtop}, item~3).

  If $\Eval M \rho \neq \bot$, then
  $\Eval {\pifz M N P} \rho = \Eval {\ifz M N P} \rho$.  We have
  already seen that
  $\ifz {M \theta} {N \theta} {P \theta} \R_{\F{} \tau} \Eval {\ifz M N
    P} \rho$, so
  $\Prob (C \cdot \ifz {M \theta} {N \theta} {P \theta} \vp) \geq
  \extd h (\Eval {\ifz M N P} \rho) = \extd h (\Eval {\pifz M N P}
  \rho)$.  For every $a \in \rat \cap [0, 1)$, if we can derive
  $C \cdot \ifz {M \theta} {N \theta} {P \theta} \vp a$, then we can
  also derive $C \cdot \pifz {M \theta} {N \theta} {P \theta} \vp a$,
  so
  $\Prob (C \cdot \pifz {M \theta} {N \theta} {P \theta} \vp) \geq
  \Prob (C \cdot \ifz {M \theta} {N \theta} {P \theta} \vp)$, and that
  is larger than or equal to $\extd h (\Eval {\pifz M N P} \rho)$.  \qed
\end{proof}

Given a ground term (or context) $M$, $\Eval M \rho$ does not depend
on $\rho$, and we will simply write $\Eval M$ in this case.
\begin{proposition}[Adequacy]
  \label{prop:adeq}
  In any of the languages CBPV$(\Demon,\Nature)$,
  CBPV$(\Demon, \allowbreak \Nature)+\pifzkw$,
  CBPV$(\Demon,\Nature)+\bigcirc$, and
  CBPV$(\Demon,\Nature)+\pifzkw+\bigcirc$, for every ground term
  $M \colon \F\Vt\unitT$,
  \[
    \Prob (M \vp) = \extd h (\Eval M),
  \]
  where $h$ is the map
  $\nu \in \Val_{\leq 1} \Sierp \mapsto \nu (\{\top\})$.

  Explicitly: either $\Eval M = \bot$ and $\Prob (M \vp) = 0$, or
  $\Eval M = \emptyset$ and $\Prob (M \vp)=1$, or
  $\Eval M \neq \bot, \emptyset$ and
  $\Prob (M \vp) = \min_{\nu \in \Eval M} \nu (\{\top\})$.
\end{proposition}
\begin{proof}
  By Proposition~\ref{prop:R} applied to $\theta = []$,
  $M \R_{\F\Vt\unitT} \Eval M$.  By Lemma~\ref{lemma:R:FVunit}, item~2,
  $[\_] \R_{\Vt{} \unitT}^* h$, so
  $\Prob ([\_] \cdot M \vp) \geq \extd h (\Eval M)$.  The converse
  inequality is by soundness (Proposition~\ref{prop:sound}, item~2).
  \qed
\end{proof}

\section{Consequences of Adequacy}
\label{sec:cons-adeq}

\begin{definition}
  \label{defn:app:preorder}
  The \emph{applicative preorder} $\precsim^{app}_\tau$ between ground
  CBPV$(\Demon, \Nature)$ terms of value type $\tau$ is defined by
  $M \precsim^{app}_\tau N$ if and only if for every ground term
  $Q \colon \tau \to \F\Vt\unitT$, $\Prob (QM \vp) \leq \Prob (QN \vp)$.

\end{definition}
While the applicative preorder is only defined at \emph{value} types,
one can extend it fairly trivially to computation types by letting
$M \precsim^{app}_{\underline\tau} N$ if and only if $\thunk M
\precsim^{app}_{\U\underline\tau} \thunk N$.

As for $\precsim_\anytype$ (Definition~\ref{defn:context:preorder}),
we will freely reuse the notations $\precsim^{app}_\tau$ for all the
variants of CBPV$(\Demon,\Nature)$ considered in this paper, with or
without $\bigcirc$ and $\pifzkw$.  Any result that does not mention
the language considered holds for all four: this will notably be the
case in the current section.

\begin{lemma}
  \label{lemma:precsim:app}
  For all ground terms $M, N \colon \sigma \to \underline\tau$ such
  that $M \precsim_{\sigma \to \underline\tau} N$, for every ground
  term $P \colon \sigma$, $MP \precsim_{\underline\tau} NP$.
\end{lemma}
\begin{proof}
  We must show that for every ground evaluation context
  $C \colon \underline\tau \vdash \F\Vt\unitT$,
  $\Prob (C \cdot MP \vp) \leq \Prob (C \cdot NP \vp)$.  By
  Lemma~\ref{lemma:prob:disc},
  $\Prob (C \cdot MP \vp) = \Prob (C [\_ P] \cdot M \vp)$.  Similarly,
  $\Prob (C \cdot NP \vp) = \Prob (C [\_ P] \cdot N \vp)$.  Since
  $M \precsim_{\sigma \to \underline\tau} N$,
  $\Prob (C [\_ P] \cdot M \vp) \leq \Prob (C [\_ P] \cdot N \vp)$,
  and we conclude.  \qed
\end{proof}

We reuse the logical relation of Section~\ref{sec:adequacy}.

The following is sometimes called \emph{Milner's Context Lemma} in the
setting of PCF, and we will prove it by using a variant of an argument
due to A. Jung \cite[Theorem~8.1]{Streicher:pcf}.
\begin{theorem}[Contextual=applicative]
  \label{thm:context=app}
  For every value type $\tau$, the contextual preorder $\precsim_\tau$
  and the applicative preorder $\precsim^{app}_\tau$ on ground
  CBPV$(\Demon,\Nature)$ terms of type $\tau$ are the same relation.
\end{theorem}
\begin{proof}
  Let $M$, $N$ be two ground terms of type $\tau$.  If
  $M \precsim^{app}_\tau N$, then consider a ground evaluation context
  $C \colon \tau \vdash \F\Vt\unitT$.  By Lemma~\ref{lemma:prob:disc:_},
  $\Prob (C [M] \vp)$, which is equal to
  $\Prob ([\_] \cdot C [M] \vp)$ by definition, is equal to
  $\Prob (C \cdot M \vp)$.  By adequacy (Proposition~\ref{prop:adeq}),
  $\Prob (C [M] \vp) = \extd h (\Eval {C [M]})$ where $h$ is the map
  $\nu \in \Val_{\leq 1} \Sierp \mapsto \nu (\{\top\})$.  Let
  $Q = \lambda x_\tau . C [x_\tau]$, where $x_\tau$ is a fresh
  variable of type $\tau$.  Then $\Eval {C [M]} = \Eval {QM}$.  By
  adequacy again, $\Prob (QM \vp) = \extd h (\Eval {QM})$, so
  $\Prob (C [M] \vp) = \Prob (QM \vp)$.  Similarly,
  $\Prob (C [N] \vp) = \Prob (QN \vp)$.  Since
  $M \precsim^{app}_\tau N$, the former is less than or equal to the
  latter, so $M \precsim_\tau N$.

  Conversely, let us assume $M \precsim_\tau N$.  Consider a ground
  term $Q \colon \tau \to \F\Vt\unitT$.  By Proposition~\ref{prop:R} with
  $\theta = []$, $M \R_\tau \Eval M$.  By Lemma~\ref{lemma:R:reduc},
  since $M \precsim_\tau N$, we also have $N \R_\tau \Eval M$.  By
  Proposition~\ref{prop:R} again, $Q \R_{\tau \to \F\Vt\unitT} \Eval Q$.
  Hence $QN \R_{\F\Vt\unitT} \Eval {QM}$.  By Lemma~\ref{lemma:R:FVunit},
  $[\_] \R_{\Vt{} \unitT}^* h$, where $h$ is as above.  Using the
  definition of $\R_{\F\Vt\unitT}$,
  $\Prob (QN \vp) = \Prob ([\_] \cdot QN \vp) \geq \extd h (\Eval
  {QM})$.  The latter is equal to $\Prob (QM \vp)$ by adequacy
  (Proposition~\ref{prop:adeq}).  We have shown
  $\Prob (QM \vp) \leq \Prob (QN \vp)$, where $Q$ is arbitrary, hence
  $M \precsim_\tau^{app} N$.  \qed
\end{proof}

\begin{corollary}
  \label{corl:context=app:comp}
  For every computation type $\underline\tau$, the contextual preorder
  $\precsim_{\underline\tau}$ and the applicative preorder
  $\precsim^{app}_{\underline\tau}$ on ground CBPV$(\Demon,\Nature)$
  terms of type $\underline\tau$ are the same relation.
\end{corollary}
\begin{proof}
  We claim that $M \precsim_{\underline\tau} N$ if and only if
  $\thunk M \precsim_{\U\underline\tau} \thunk N$.  The result will
  then follow from Theorem~\ref{thm:context=app}, since $\thunk M
  \precsim_{\U\underline\tau} \thunk N$ is equivalent to $\thunk M
  \precsim^{app}_{\U\underline\tau} \thunk N$, hence to $M
  \precsim^{app}_{\underline\tau} N$, by definition.

  If $M \precsim_{\underline\tau} N$, let $C$ be any ground evaluation
  context of type $\U\underline\tau \vdash \F\Vt\unitT$.  Let us write
  $C$ as $E_0 E_1 E_2 \cdots E_n$, where
  $E_i \colon \anytype_{i+1} \vdash \anytype_i$, $\anytype_{n+1}=\U\underline\tau$ and
  $\anytype_0=\F\Vt\unitT$.  Since $\U\underline\tau$ is not $\unitT$,
  $\Vt\unitT$, or $\F\Vt\unitT$, $n$ must be at least $1$.  The only
  elementary context $E_n$ of type $\U\underline\tau \vdash \anytype_n$ is
  $[\force \_]$.  Let $C' = E_0 E_1 E_2 \cdots E_{n-1}$.  Then
  $\Prob (C \cdot \thunk M \vp) = \Prob (C' [\force \_] \cdot \thunk M
  \vp) = \Prob (C' [\force \thunk M] \vp)$ (by
  Lemma~\ref{lemma:prob:disc:_})
  $= \extd h (\Eval {C' [\force \thunk M]})$ (by adequacy, where $h$
  is given in Proposition~\ref{prop:adeq})
  $= \extd h (\Eval {C' [M]})$ (because $\force$ and $\thunk$ are both
  interpreted as identity maps)
  $= \Prob (C' [M] \vp) = \Prob (C' \cdot M \vp)$.  Similarly,
  $\Prob (C \cdot \thunk N \vp) = \Prob (C' \cdot N \vp)$.  Since
  $M \precsim_{\underline\tau} N$, the former is less than or equal to
  the latter.  This allows us to conclude that
  $\thunk M \precsim_{\U\underline\tau} \thunk N$.

  Conversely, we assume that
  $\thunk M \precsim_{\U\underline\tau} \thunk N$, and we consider an
  arbitrary ground evaluation context
  $C \colon \underline\tau \vdash \F\Vt\unitT$.  Then $C [\force \_]$ is
  a ground evaluation context of type
  $\U\underline\tau \vdash \F\Vt\unitT$, so
  $\Prob (C [\force \_] \cdot \thunk M \vp) \leq \Prob (C [\force \_]
  \cdot \thunk N \vp)$.  As above, we have
  $\Prob (C [\force \_] \cdot \thunk M \vp) = \extd h (\Eval {C
    [\force \thunk M]}) = \extd h (\Eval {C [M]}) = \Prob (C \cdot M
  \vp)$, and
  $\Prob (C [\force \_] \cdot \thunk N \vp) = \Prob (C \cdot N \vp)$,
  and the former is less than or equal to the latter.  \qed
\end{proof}

The following proposition is a form of extensionality: two
abstractions are related by $\precsim_{\sigma \to \underline\tau}$ if
and only if applying them to the same ground terms yield related
results.
\begin{proposition}
  \label{prop:precsim:lambda}
  Let $M, N \colon \underline\tau$ be two terms with $x_\sigma$ as
  sole free variable.  Then
  $\lambda x_\sigma . M \precsim_{\sigma \to \underline\tau} \lambda
  x_\sigma . N$ if and only if for every ground term
  $P \colon \sigma$,
  $M [x_\sigma:=P] \precsim_{\underline\tau} N [x_\sigma:=P]$.
\end{proposition}
\begin{proof}
  If
  $\lambda x_\sigma . M \precsim_{\sigma \to \underline\tau} \lambda
  x_\sigma . N$, then
  $(\lambda x_\sigma . M) P \precsim_{\underline\tau} (\lambda
  x_\sigma . N) P$ for every ground term $P \colon \sigma$, by
  Lemma~\ref{lemma:precsim:app}.  Hence for every ground evaluation
  context $C \colon \tau \vdash \F\Vt\unitT$,
  $\Prob (C \cdot (\lambda x_\sigma . M) P \vp) \leq \Prob (C \cdot
  (\lambda x_\sigma . N) P \vp)$.  Using
  Lemma~\ref{lemma:prob:disc:_}, we obtain
  $\Prob (C [(\lambda x_\sigma . M) P] \vp) \leq \Prob (C [(\lambda
  x_\sigma . N) P] \vp)$.  By adequacy (Proposition~\ref{prop:adeq}),
  $\Prob (C [(\lambda x_\sigma . M) P] \vp) = \extd h (\Eval {C
    [(\lambda x_\sigma . M) P]})$, where $h (\nu) = \nu (\{\top\})$.
  That is equal to $\extd h (\Eval {C [M [x_\sigma := P]]})$, hence to
  $\Prob (C [M [x_\sigma := P]] \vp) = \Prob (C \cdot M [x_\sigma :=
  P] \vp)$.  Similarly,
  $\Prob (C [(\lambda x_\sigma . N) P] \vp) = \Prob (C \cdot N
  [x_\sigma := P] \vp)$, so
  $\Prob (C \cdot M [x_\sigma := P] \vp) \leq \Prob (C \cdot N
  [x_\sigma := P] \vp)$.  Since $C$ is arbitrary,
  $M [x_\sigma := P] \precsim_{\underline\tau} N [x_\sigma := P]$.

  Conversely, assume that
  $M [x_\sigma:=P] \precsim_\tau N [x_\sigma:=P]$ for every ground
  term $P \colon \sigma$.  We wish to show that for every ground
  evaluation context
  $C \colon (\sigma \to \underline\tau) \vdash \F\Vt\unitT$,
  $\Prob (C \cdot \lambda x_\sigma . M \vp) \leq \Prob (C \cdot
  \lambda x_\sigma . N \vp)$.  Let us write $C$ as
  $E_0 E_1 E_2 \cdots E_n$, where each $E_i$ is of type
  $\anytype_{i+1} \vdash \anytype_i$, $\anytype_0 = \F\Vt\unitT$ and
  $\anytype_{n+1} = \sigma \to \underline\tau$.
  We cannot have $n=0$, since $\sigma \to \underline\tau$ is none of
  the types $\unitT$, $\Vt\unitT$, $\F\Vt\unitT$.  By inspection of the
  possible shape of the elementary context $E_n$, we see that it must
  be of the form $[\_ P]$ for some (ground) term $P \colon \sigma$.
  Let
  $C' = E_0 E_1 E_2 \cdots E_{n-1} \colon \underline\tau \to
  \F\Vt\unitT$.  Using Lemma~\ref{lemma:prob:disc:_} and adequacy as
  above,
  $\Prob (C \cdot \lambda x_\sigma . M \vp) = \Prob (C' \cdot M
  [x_\sigma:=P] \vp)$, and similarly with $N$ instead of $M$.  We have
  $\Prob (C' \cdot M [x_\sigma:=P] \vp) \leq \Prob (C' \cdot N
  [x_\sigma:=P] \vp)$ since
  $M [x_\sigma:=P] \precsim_{\underline\tau} N [x_\sigma:=P]$, so
  $\Prob (C \cdot \lambda x_\sigma . M \vp) \leq \Prob (C \cdot
  \lambda x_\sigma . N \vp)$.  \qed
\end{proof}

A final, expected, consequence of adequacy is the following.
\begin{proposition}
  \label{prop:fa:easy}
  For every value type $\tau$, for every two ground terms
  $M, N \colon \tau$, if $\Eval M \leq \Eval N$ then
  $M \precsim_\tau N$.
\end{proposition}
\begin{proof}
  For every ground term $Q \colon \tau \to \F\Vt\unitT$,
  $\Eval {QM} = \Eval Q (\Eval M) \leq \Eval Q (\Eval N)$, hence
  $\extd h (\Eval {QM}) \leq \extd h (\Eval {QN})$ for every
  continuous map $h \colon \Val_{\leq 1} \Sierp \to [0, 1]$.  By
  adequacy (Proposition~\ref{prop:adeq}),
  $\Prob (QM \vp) = \extd h (\Eval {QM})$, and
  $\Prob (QN \vp) = \extd h (\Eval {QN})$ where $h$ is the map
  $\nu \mapsto \nu (\{\top\})$.  Therefore
  $\Prob (QM \vp) \leq \Prob (QN \vp)$.  \qed
\end{proof}
The converse implication, if it holds, is full abstraction.

\section{The Failure of Full Abstraction}
\label{sec:fail-full-abstr}

We will show that CBPV$(\Demon, \Nature)$ is not fully abstract, for
two reasons.  One is the expected lack of a parallel if operator, just
as in PCF \cite{Plotkin:PCF}.  The other is the lack of a statistical
termination tester, as in \cite{jgl-jlap14}.

Our main tool is a variant on our previous logical relations
${(\RR_\anytype)}_{\anytype \text{ type}}$.  This time, $\RR_\anytype$ will be an $I$-ary
relation, for some non-empty set $I$, between semantical
values---namely, $\RR_\anytype \subseteq \Eval \anytype^I$.  The construction is
parameterized by a finite family $\mathcal J$ of subsets of $I$, and
two $I$-ary relations $\tr, \underline\tr \subseteq [0, 1]^I$.  Again
we will also define auxiliary relations $\RR_\sigma^\perp$, and
$\RR_\sigma^*$, which are certain sets of $I$-tuples of
Scott-continuous maps from $\Eval \sigma$ to $[0, 1]$.  We write
$\vec a$ for ${(a_i)}_{i \in I}$, and similarly with $\vec \nu$,
$\vec Q$, etc. For every $\vec a \in [0, 1]^I$ and every subset $J$ of
$I$, we write $\vec a_{|J}$ for the vector obtained from $\vec a$ by
replacing every element $a_i$, $i \in J$, by $0$; namely, $a_{|J\;i}$
is equal to $0$ if $i \in J$, to $a_i$ otherwise.  We require the
following:
\begin{itemize}
\item $I \in \mathcal J$, $\mathcal J$ is closed under binary unions,
  and is well-founded: every filtered family ${(J_k)}_{k \in K}$ in
  $\mathcal J$ has a least element $J_{k_1}$, $k_1 \in K$.
\item $\tr$ is non-empty, closed under directed suprema, convex
  (notably, if ${(a_i)}_{i \in I}$ and ${(b_i)}_{i \in I}$ are in
  $\tr$ then so is ${((a_i+b_i)/2)}_{i \in I}$), and is
  \emph{$\mathcal J$-lower}, meaning that for every $\vec a \in \tr$,
  for every $J \in \mathcal J$, $\vec a_{|J}$ is in $\tr$;
\item $\underline\tr$ is closed under directed suprema, under pairwise
  minima (if ${(a_i)}_{i \in I}$ and ${(b_i)}_{i \in I}$ are in
  $\underline\tr$ then so is ${(\min (a_i+b_i))}_{i \in I}$), contains
  the all one vector $\vec 1$, and is $\mathcal J$-lower.
\end{itemize}
We define the following.
\begin{itemize}
\item $\vec a \in \RR_{\U\underline\tau}$ iff $\vec a \in \RR_\tau$;
\item $\vec a \in \RR_{\unitT}$ (resp., $\RR_{\intT}$) iff: the set
  $J = \{i \in I \mid a_i = \bot\}$ is in $\mathcal J$ and the
  components $a_i$, $i \in I \diff J$, are all equal;
\item $\vec a \in \RR_{\sigma \times \tau}$, where $a_i = (b_i, c_i)$
  for every $i \in I$, iff $\vec b \in \RR_\sigma$ and
  $\vec c \in \RR_\tau$;
\item $\vec \nu \in \RR_{\Vt{} \sigma}$ iff for all
  $\vec h \in \RR_\sigma^\perp$,
  ${(\int_{x \in \Eval \sigma} h_i (x) d\nu_i)}_{i \in I} \in \tr$;
\item $\vec h \in \RR_\sigma^\perp$ iff 
  for all $\vec a \in \RR_\sigma$, ${(h_i (a_i))}_{i \in I} \in \tr$;
  %
\item $\vec Q \in \RR_{\F{} \sigma}$ iff for all $\vec h \in
  \RR_\sigma^*$, ${(\extd {h_i} (Q_i))}_{i \in I} \in \underline\tr$;
\item $\vec h \in \RR_\sigma^*$ iff 
  for all $\vec a \in \RR_\sigma$,
  ${(h_i (a_i))}_{i \in I} \in \underline\tr$;
\item $\vec f \in \RR_{\sigma \to \underline\tau}$ iff for all $\vec a
  \in \RR_\sigma$, ${(f_i (a_i))}_{i \in I} \in \RR_{\underline\tau}$.
\end{itemize}
For every $I$-indexed tuple $\vec \rho$ of environments, finally,
$\vec \rho \in \RR_*$ if and only if for every variable $x_\sigma$,
${(\rho_i (x_\sigma))}_{i \in I} \in \RR_\sigma$.

\begin{lemma}
  \label{lemma:tr:basic}
  \begin{enumerate}
  \item For every $\vec \rho \in \RR_*$, for every
    CBPV$(\Demon, \Nature)$ term $M \colon \tau$,
    ${(\Eval M \rho_i)}_{i \in I}$ is in $\RR_\tau$.
  \item The same remains true for all CBPV$(\Demon, \Nature)+\pifzkw$
    terms if $\mathcal J \subseteq \{\emptyset, I\}$.
  \end{enumerate}
\end{lemma}
\begin{proof}
  We first show: $(a)$ $\RR_\anytype$ is closed under directed suprema taken
  in $\Eval \anytype^I$, and contains $\vec \bot = {(\bot)}_{i \in I}$.  This
  is by induction on the type $\anytype$.  Most cases are trivial.  We deal
  with the remaining ones:
  \begin{itemize}
  \item When $\anytype = \unitT$ or $\anytype = \intT$, $\vec\bot$ is in $\RR_\anytype$,
    because $I \in \mathcal J$.  We must show that the supremum
    $\vec a$ of every directed family ${(\vec a_k)}_{k \in K}$ in
    $\RR_\anytype$ is in $\RR_\anytype$.  Let $J_k = \{i \in I \mid a_{ki} = \bot\}$
    for each $k \in K$, and $J = \{i \in I \mid a_i = \bot\}$.  Define
    $k \sqsubseteq k'$ if and only if $\vec a_k \leq \vec a_{k'}$.
    Then $k \sqsubseteq k'$ implies $J_k \supseteq J_{k'}$, so
    ${(J_k)}_{k \in K}$ is a filtered family.  Since $\mathcal J$ is
    well-founded, there is an index $k_1 \in K$ such that
    $J_k=J_{k_1}$ for every $k \sqsupseteq k_1$.  Then, for every
    $i \in I$, $a_i = \sup_{k \sqsupseteq k_1} a_{ki}$ is equal to
    $\bot$ if $i \in J$, and is different from $\bot$ otherwise.  In
    particular, $J = J_{k_1}$.  Letting $b_k$ be the common value of
    the terms $a_{ki}$, $i \in I \diff J_k = I \diff J$, for each
    $k \sqsupseteq k_1$, we have $a_i = \sup_{k \sqsupseteq k_1} b_k$
    for every $i \in I \diff J$, and all these values are equal.
    Therefore $\vec a$ is in $\RR_\anytype$.
  \item When $\anytype = \Vt{} \sigma$, $\vec \bot$ is the tuple consisting of
    zero valuations only, and for every $\vec h \in \RR_\sigma^\perp$,
    ${(\int_{x \in \Eval \sigma} h_i (x) d0)}_{i \in I} = \vec 0$ is
    in $\tr$ (since $\tr$ is $\mathcal J$-lower and
    $I \in \mathcal J$), so $\vec \bot \in \RR_\anytype$.
    In order to show closure under directed suprema, let
    ${(\vec \nu_j)}_{j \in J}$ be a directed family in $\RR_\anytype$, with
    $\vec \nu_j = {(\nu_{ji})}_{i \in I}$.  Its supremum is
    $\vec \nu = {(\nu_i)}_{i \in I}$ where
    $\nu_i = \sup_{j \in J} \nu_{ji}$.
    For every $\vec h \in \RR_\sigma^\perp$,
    ${(\int_{x \in \Eval \sigma} h_i (x) d\nu_i)}_{i \in I} = {(\sup_j
      \int_{x \in \Eval \sigma} h_i (x) d\nu_{ji})}_{i \in I}$, since
    integration is Scott-continuous in the
    valuation.  
    That is a directed supremum of values in $\tr$, hence is in $\tr$.
    It follows that $\vec \nu$ is in $\RR_{\Vt{} \sigma}$.
  \item When $\anytype = \F{} \sigma$, $\vec \bot$ is in $\RR_{\F{} \sigma}$
    because, for every $\vec h \in \RR_\sigma^*$,
    ${(\extd {h_i} (\bot))}_{i \in I}$ is equal to $\vec 0$ (since
    $\extd {h_i}$ is strict), and that is in $\underline\tr$:
    $\vec 0 = \vec 1_{|I}$, which is in $\underline\tr$ because
    $\underline\tr$ is $\mathcal J$-lower and $I \in \mathcal J$.  As
    far as closure under directed suprema is concerned, let
    ${(\vec Q_j)}_{j \in J}$ be a directed family in $\RR_\anytype$, where
    $\vec Q_j = {(Q_{ji})}_{i \in I}$.  Let us write its supremum as
    $\vec Q = {(Q_i)}_{i \in I}$. For every $h \in \RR_\sigma^*$,
    ${(\extd {h_i} (Q_i))}_{i \in I}$ is the supremum of the family of
    tuples ${(\extd {h_i} (Q_{ji}))}_{i \in I}$, $j \in J$, since
    $\extd {h_i}$ is Scott-continuous (Proposition~\ref{prop:Qtop},
    item~2), and all those tuples are in $\underline\tr$.  Since the
    latter is closed under directed suprema,
    ${(\extd {h_i} (Q_i))}_{i \in I}$ is in $\underline\tr$, so
    $\vec Q$ is in $\RR_{\F{} \sigma}$.
  \end{itemize}

  Next, we claim that: $(b)$ for every $\vec \nu \in \RR_{\Vt{} \sigma}$
  and for every $\vec f \in \RR_{\sigma \to \Vt{} \tau}$,
  ${(f_i^\dagger (\nu_i))}_{i \in I}$ is in $\RR_{\Vt{} \tau}$.  To this
  end, let $\vec h \in \RR_\tau^\perp$.  Our goal is to show that
  ${(\int_{y \in \Eval \tau} h_i (y) df_i^\dagger (\nu_i))}_{i \in I}$
  is in $\tr$.  Using (\ref{eq:dagger}), this boils down to showing
  that ${(\int_{x \in \Eval \sigma} h'_i (x) d\nu_i)}_{i \in I}$ is in
  $\tr$, where $h'_i$ is the map
  $x \mapsto \int_{y \in \Eval \tau} h_i (y) df_i (x)$.  We note that
  $\vec h' = {(h'_i)}_{i \in I}$ is in $\RR_\sigma^\perp$: for every
  $\vec a \in \RR_\sigma$, ${(f_i (a_i))}_{i \in I}$ is in
  $\RR_{\Vt{} \tau}$ (by definition of $\RR_{\sigma \to \Vt{} \tau}$); by the
  definition of $\RR_{\Vt{} \tau}$, and using the fact that
  $\vec h \in \RR_\tau^\perp$,
  ${(\int_{y \in \Eval \tau} h_i (y) df_i (a_i))}_{i \in I}$ is in
  $\tr$, in other words ${(h'_i (a_i))}_{i \in I}$ is in $\tr$.  Since
  ${(\vec h'_i)}_{i \in I} \in \RR_\sigma^\perp$ and
  $\vec \nu \in \RR_{\Vt{} \sigma}$, the claim follows.

  We also claim that: $(c)$ for every $\vec Q \in \RR_{\F{} \sigma}$ and
  for every $\vec f \in \RR_{\sigma \to \F{} \tau}$,
  ${(\extd {f_i} (Q_i))}_{i \in I}$ is in $\RR_{\F{} \tau}$.  Let
  $\vec h \in \RR_\tau^*$.  We wish to show that
  ${(\extd {h_i} (\extd {f_i} (Q_i)))}_{i \in I}$ is in
  $\underline\tr$.  Using Proposition~\ref{prop:Qtop}, item~4, this
  amounts to showing that
  ${(\extd {(\extd {h_i} \circ f_i)} (Q_i))}_{i \in I}$ is in
  $\underline\tr$.  We note that ${(\extd {h_i} \circ f_i)}_{i \in I}$
  is in $\RR_\sigma^\perp$: for every $\vec a \in \RR_\sigma$,
  ${(f_i (a_i))}_{i \in I}$ is in $\RR_{\F{} \tau}$, and $\vec h$ is in
  $\RR_\tau^*$, so ${(\extd {h_i} (f_i (a_i)))}_{i \in I}$ is in
  $\underline\tr$.  Since $\vec Q \in \RR_{\F{} \sigma}$, and using the
  definition of $\RR_{\F{} \sigma}$,
  ${(\extd {(\extd {h_i} \circ f_i)} (Q_i))}_{i \in I}$ is in
  $\underline\tr$.

  Finally, for every vector $\vec a$ in $\Eval \anytype^I$, and for every
  subset $J$ of $I$, we define $\vec a_{|J}$ as the vector obtained
  from $\vec a$ by replacing each component $a_i$ with $i \in J$ by
  $\bot$.  We claim that: $(d)$ for every $\vec a \in \RR_\anytype$, for
  every $J \in \mathcal J$, $\vec a_{|J}$ is in $\RR_\anytype$.  This is by
  induction on $\anytype$.  For types $\anytype$ of the form $\sigma \times \tau$,
  $\U\underline\tau$, and $\sigma \to \underline\tau$, we simply call
  the induction hypothesis.  When $\anytype$ is $\unitT$ or $\intT$, let
  $J' = \{i \in I \mid a_i=\bot\}$, and
  $J'' = \{i \in I \mid a_{|J\;i}=\bot\}$.  We have $J''=J' \cup J$,
  and $J' \in \mathcal J$ by induction hypothesis.  Since $\mathcal J$
  is closed under binary unions, $J''$ is in $\mathcal J$.  Moreover,
  all the components $a_{|J\;i}$ with $i \in I \diff J''$ are equal to
  $a_i$, and they are all equal.  Therefore $\vec a_{|J}$ is in
  $\RR_\anytype$.  When $\anytype = \Vt\sigma$, let $\vec \nu \in \RR_{\Vt{} \sigma}$.
  For every $\vec h \in \RR_\sigma^\perp$,
  $\vec b = {(\int_{x \in \Eval \sigma} h_i (x) d\nu_i)}_{i \in I}$ is
  in $\tr$.  The vector
  ${(\int_{x \in \Eval \sigma} h_i (x) d\nu_{|J\;i})}_{i \in I}$ is
  equal to $\vec b_{|J}$, hence is in $\tr$ as well, since $\tr$ is
  $\mathcal J$-lower.  When $\anytype = \F{} \sigma$, let
  $\vec Q \in \RR_{\F{} \sigma}$.  For every $\vec h \in \RR_\sigma^*$,
  ${(\extd {h_i} (Q_i))}_{i \in I}$ is in $\underline\tr$.  Since each
  function $\extd {h_i}$ is strict, for every $i \in J$,
  $\extd {h_i} (Q_{|J\;i}) = \bot$.  For every $i \in I \diff J$,
  $\extd {h_i} (Q_{|J\;i}) = \extd {h_i} (Q_i)$.  Therefore
  ${(\extd {h_i} (Q_{|J\;i}))}_{i \in I}$ is equal to $\vec a_{|J}$
  where $\vec a = {(\extd {h_i} (Q_i))}_{i \in I}$, and that is in
  $\underline\tr$ by assumption and the fact that $\underline\tr$ is
  $\mathcal J$-lower.  Hence $\vec Q_{|J}$ is in $\RR_{\F{} \sigma}$.

  1. We now prove the lemma by induction on $M$.  For variables, this
  follows from the assumption that $\vec \rho \in \RR_*$.  The case of
  constants $\underline *$ and $\underline n$ is clear.  The case of
  $\lambda$-abstractions and of applications is immediate from the
  definition of $\RR_{\sigma \to \underline\tau}$.  Similarly, the
  case of terms $\pi_1 M$, $\pi_2 M$ and $\langle M, N \rangle$ are
  immediate from the definition of $\RR_{\sigma \times \tau}$.  For
  terms of the form $\thunk M$, with $M \colon \underline\tau$, or
  terms of the form $\force M$, with $M \colon \U\underline\tau$, the
  claim is trivial.

  For terms of the form $\produce M$, with $M \colon \sigma$, by
  induction hypothesis ${(\Eval M \rho_i)}_{i \in I}$ is in
  $\RR_\sigma$.  In order to show that ${(\Eval {\produce M}
    \rho_i)}_{i \in I} = {(\eta^\Smyth (\Eval M \rho_i))}_{i \in I}$
  is in $\RR_{\F{} \sigma}$, we fix $\vec h \in \RR_\sigma^*$, and we
  check that ${(\extd {h_i} (\eta^\Smyth (\Eval M \rho_i)))}_{i \in
    I}$ is in $\underline\tr$.  Since $\extd {h_i} (\eta^\Smyth (\Eval
  M \rho_i)) = h_i (\Eval M \rho_i)$ (Proposition~\ref{prop:Qtop},
  item~2), this follows from the definition of $\RR_\sigma^*$.

  For terms of the form $\pto M {x_\sigma} N$, with
  $M \colon \F{} \sigma$ and $N \colon \F{} \tau$, we must show that
  ${(\Eval {\pto M {x_\sigma} N} \rho_i)}_{i \in I}$ is in
  $\Eval {\F{} \tau}$.  Since
  $\Eval {\pto M {x_\sigma} N} \rho_i = \extd {f_i} (\Eval M \rho_i)$,
  where $f_i (V) = \Eval N \rho_i [x_\sigma \mapsto V]$, we will
  obtain this as a consequence of $(c)$ if we can show that
  ${(f_i)}_{i \in I}$ is in $\RR_{\sigma \to \F{} \tau}$.  For every
  $\vec a \in \RR_\sigma$,
  ${(\rho_i [x_\sigma \mapsto a_i])}_{i \in I}$ is in $\RR_*$, so
  ${(f_i (a_i))}_{i \in I} = {(\Eval N \rho_i [x_\sigma \mapsto
    a_i])}_{i \in I}$ is indeed in $\RR_{\sigma \to \F{} \tau}$.
  
  For terms of the form $\retkw M$, where $M \colon \tau$, we must
  show that for every $\vec h \in \RR_\tau^\perp$,
  ${(\int_{x \in \Eval \tau} h_i (x) d \Eval {\retkw M} \rho_i)}_{i
    \in I}$ is in $\tr$.  Since
  $\int_{x \in \Eval \tau} h_i (x) d \Eval {\retkw M} \rho_i = \int_{x
    \in \Eval \tau} h_i (x) d\delta_{\Eval M \rho_i} = h_i (\Eval M
  \rho_i)$, this follows from the fact that
  $\vec h \in \RR_\tau^\perp$ and the definition of $\RR_\tau^\perp$.

  For terms of the form $\dokw {x_\sigma \leftarrow M}; N$, with $M
  \colon \Vt{} \sigma$ and $N \colon \Vt{} \tau$, we wish to show that
  ${(\Eval {\dokw {x_\sigma \leftarrow M}; N} \rho_i)}_{i \in I}$ is
  in $\RR_{\Vt{} \tau}$, namely that ${(f_i^\dagger (\Eval N \rho_i))}_{i
    \in I}$ is in $\RR_{\Vt{} \tau}$, where $f_i (V) = \Eval N \rho_i
  [x_\sigma \mapsto V]$.  As in the case of $\kwfont{to}$ terms,
  ${(f_i)}_{i \in I}$ is in $\RR_{\sigma \to \Vt{} \tau}$, and ${(\Eval N
    \rho_i)}_{i \in I}$ is in $\RR_{\Vt{} \sigma}$, so the claim is proved
  by applying $(b)$.

  For terms of the form $\suc M$, with $M \colon \intT$, by induction
  hypothesis ${(\Eval M \rho_i)}_{i \in I}$ is in $\RR_{\intT}$.  Let
  $J = \{i \in I \mid \Eval M \rho_i=\bot\}$, and $n \in \Z$ be the
  common value of $\Eval M \rho_i$, $i \in I \diff J$ (or an arbitrary
  element of $\Z$ if $I=J$).  Then $J$ is also equal to
  $\{i \in I \mid \Eval {\suc M} \rho_i = \bot\}$, which is therefore
  in $\mathcal J$.  Moreover, $n+1$ is the common value of
  $\Eval {\suc M} \rho_i$, $i \in I \diff J$.  We reason similarly for
  terms of the form $\pred M$.

  For terms of the form $\ifz M N P$, where $M \colon \intT$ and
  $N, P \colon \anytype$, by hypothesis, in particular,
  ${(\Eval M \rho_i)}_{i \in I}$ is in $\RR_{\intT}$.  Let
  $J = \{i \in I \mid \Eval M \rho_i=\bot\}$, and $n$ be the common
  value of $\Eval M \rho_i$, $i \in I \diff J$ (or any element of $\Z$
  if $J=I$).  ${(\Eval {\ifz M N P} \rho_i)}_{i \in I}$ is equal to
  $\vec a_{|J}$, where $\vec a = {(\Eval N \rho_i)}_{i \in I}$ if
  $n=0$ and $\vec a = {(\Eval P \rho_i)}_{i \in I}$ if $n\neq 0$.  The
  latter is in $\RR_\anytype$ by induction hypothesis, so the former is in
  $\RR_\anytype$, too, by $(d)$.

  For terms of the form $M; N$, with $M \colon \unitT$ and
  $N \colon \anytype$, ${(\Eval M \rho_i)}_{i \in I}$ is in $\RR_{\unitT}$ by
  induction hypothesis. Let
  $J = \{i \in I \mid \Eval M \rho_i=\bot\}$.
  ${(\Eval {M; N})}_{i \in I}$ is equal to $\vec a_{|J}$ where
  $\vec a = {(\Eval N \rho_i)}_{i \in I}$, which is in $\RR_\anytype$ by
  induction hypothesis and $(d)$.

  For terms of the form $M \oplus N$, with $M, N \colon \Vt{} \tau$,
  we wish to show that the tuple
  ${(\Eval {M \oplus N} \rho_i)}_{i \in I}$ is in $\RR_{\Vt{} \tau}$.
  Let $\vec h \in \RR_\tau^\perp$.  By induction hypothesis, the
  tuples
  ${(\int_{x \in \Eval \tau} h_i (x) d \Eval M \rho_i)}_{i \in I}$ and
  ${(\int_{x \in \Eval \tau} h_i (x) d \Eval N \rho_i)}_{i \in I}$ are
  in $\tr$.  Since $\tr$ is convex,
  ${(\frac 1 2 (\int_{x \in \Eval \tau} h_i (x) d \Eval M \rho_i +
    \int_{x \in \Eval \tau} h_i (x) d \Eval N\rho_i ))}_{i \in I}$ is
  also in $\tr$, and that is just
  ${(\int_{x \in \Eval \tau} h_i (x) d \Eval {M \oplus N} \rho_i)}_{i
    \in I}$.

  For terms of the form $M \owedge N$, with $M, N \colon \F{} \tau$,
  we wish to show that the tuple
  ${(\Eval {M \owedge N} \rho_i)}_{i \in I}$ is in $\RR_{\F{} \tau}$.
  Let $\vec h \in \RR_\tau^*$.  By induction hypothesis,
  ${(\extd {h_i} (\Eval M \rho_i))}_{i \in I}$ and
  ${(\extd {h_i} (\Eval N \rho_i))}_{i \in I}$ are in $\underline\tr$.
  Since $\underline\tr$ is closed under pairwise minima, and
  $\extd {h_i}$ commutes with pairwise infima
  (Proposition~\ref{prop:Qtop}, item~3),
  ${(\extd {h_i} (\Eval {M \owedge N} \rho_i))}_{i \in I}$ is also in
  $\underline\tr$, showing the claim.

  For $\abort_{\F{} \tau}$ , we consider an arbitrary vector
  $\vec h \in \RR_\tau^*$, and we must show that
  ${(\extd {h_i} (\emptyset))}_{i \in I}$ is in $\underline\tr$.  By
  Proposition~\ref{prop:Qtop}, item~3, $\extd {h_i} (\emptyset)$ is
  the top element, $\vec 1$, of $[0, 1]$, and the claim follows from
  the fact that $\vec 1 \in \underline\tr$.

  For terms of the form $\rec x_\sigma . M$, let $f_i$ be the map
  defined by $f_i (V) = \Eval M \rho_i [x_\sigma \mapsto V]$.  For
  every $\vec a \in \RR_\sigma$,
  ${(\rho_i [x_\sigma \mapsto a_i])}_{i \in I}$ is in $\RR_*$, hence
  by induction hypothesis ${(f_i (a_i))}_{i \in I}$ is in
  $\RR_\sigma$.  Let us write ${(f_i (a_i))}_{i \in I}$ as
  $\vec f (\vec a)$.  By $(a)$, $\vec \bot = {(\bot)}_{i \in I}$ is in
  $\RR_\sigma$, and therefore $\vec f (\vec \bot)$,
  $\vec f (\vec f (\vec \bot))$, \ldots, are all in $\RR_\sigma$.  By
  the other part of $(a)$, $\sup_{n \in \nat} {(\vec f)}^n (\vec a)$
  in in $\RR_\sigma$ as well.  That tuple is just
  ${(\lfp (f_i))}_{i \in I}$, namely
  ${(\Eval {\rec x_\sigma . M} \rho_i)}_{i \in I}$.

  2. In the case of terms of the form $\pifz M N P$, of type $\anytype$, and
  assuming $\mathcal J \subseteq \{\emptyset, I\}$, the induction
  hypothesis ${(\Eval M \rho_i)}_{i \in I} \in \RR_{\intT}$ implies
  that all the values $\Eval M \rho_i$ are the same: letting
  $J = \{i \in I \mid \Eval M \rho_i=\bot\}$, either $J=I$ and they
  are all equal to $\bot$, or $J = \emptyset$ and they are all equal
  by definition of $\RR_{\intT}$.  Hence
  ${(\Eval {\pifz M N P} \rho_i)}_{i \in I}$ is equal to
  ${(\Eval N \rho_i)}_{i \in I}$, to ${(\Eval P \rho_i)}_{i \in I}$,
  or to ${(\Eval {M \owedge N} \rho_i)}_{i \in I}$, and they are all
  in $\RR_\anytype$.  \qed






  


\end{proof}

\subsection{The Need for Parallel If}
\label{sec:need-parallel-if}

In this section, we let $I = \{1, 2, 3\}$,
$\mathcal J = \{\emptyset, \{1,3\}, \{2,3\}, \{1,2,3\}\}$, $\tr$ be
arbitrary (e.g., the whole of $[0, 1]^3$), and
$\underline\tr = \{(0,0,0), (0,1,0), (1,0,0), (1,1,1)\}$.  The latter
is the smallest possible set that satisfies the constraints required
of $\underline\tr$, and is the graph of the infimum function
$\wedge \colon \{0, 1\}^2 \to \{0, 1\}$.

A triple $(n_1, n_2, n_3)$ in $\Z_\bot^3$ is in $\RR_{\intT}$ if and
only if $\{i \mid n_i=\bot\}$ is empty, equal to $\{1, 3\}$ or
$\{2, 3\}$, or to $\{1, 2, 3\}$, and all the non-bottom components are
equal.  Those are the triples $(n,n,n)$, $(\bot,n,\bot)$,
$(n,\bot,\bot)$ and $(\bot,\bot,\bot)$ (with $n \neq \bot$).

\begin{lemma}
  \label{lemma:RR:h:por}
  The triples $(h_1, h_2, h_3)$ in $\RR_{\intT}^*$ are the triples of
  characteristic maps $(\chi_{U_1}, \chi_{U_2}, \chi_{U_3})$ of open
  subsets $U_1$, $U_2$, $U_3$ of $\Z_\bot$ of one of the following
  forms:
  \begin{enumerate}
  \item $U_1=U_2=U_3= \Z_\bot$;
  \item $U_1=U_2=U_3=\{n\}$ for some $n \in \Z$;
  \item $U_1=U_3=\emptyset$, $U_2$ arbitrary;
  \item $U_2=U_3=\emptyset$, $U_1$ arbitrary.
  \end{enumerate}
\end{lemma}
\begin{proof}
  A triple of Scott-continuous maps $(h_1, h_2, h_3)$ is in
  $\RR_{\intT}^*$ if and only if for all
  $(n_1, n_2, n_3) \in \RR_{\intT}$,
  $(h_1 (n_1), h_2 (n_2), h_3 (n_3)) \in \underline\tr$.  We claim
  that this is equivalent to: $(*)$ $h_1$, $h_2$, $h_3$ take their
  values in $\{0, 1\}$ and for all $n_1, n_2, n_3 \in \Z_\bot$ such
  that $n_3 = n_1 \wedge n_2$,
  $h_3 (n_3) = h_2 (n_1) \wedge h_2 (n_2)$.  In one direction, if
  $(h_1, h_2, h_3) \in \RR_{\intT}^*$, then since
  $(n, n, n) \in \RR_{\intT}$ for every $n \in \Z$,
  $(h_1 (n), h_2 (n), h_3 (n))$ is in $\underline\tr$ for every
  $n \in \Z$, in particular $h_1$, $h_2$, $h_3$ take their values in
  $\{0, 1\}$.  Also, for all $n_1, n_2, n_3 \in \Z_\bot$ such that
  $n_3 = n_1 \wedge n_2$, $h_3 (n_3) = h_2 (n_1) \wedge h_2 (n_2)$.
  Indeed, the triples $(n_1, n_2, n_3)$ such that
  $n_3 = n_1 \wedge n_2$ are of the form $(n, n, n)$, or
  $(\bot, n, \bot)$, or $(n, \bot, \bot)$, or $(\bot, \bot, \bot)$,
  with $n \in \Z$, hence are exactly the triples in $\RR_{\intT}$.
  Then $(h_1 (n_1), h_2 (n_2), h_3 (n_3))$ is in $\underline\tr$,
  hence $h_3 (n_3) = h_2 (n_1) \wedge h_2 (n_2)$ since $\underline\tr$
  is the graph of $\wedge$ on $\{0, 1\}$.  In the other direction, let
  us assume $(*)$.  For all $(n_1, n_2, n_3) \in \RR_{\intT}$, we have
  just seen that $n_3 = n_1 \wedge n_2$, so
  $(h_1 (n_1), h_2 (n_2), h_3 (n_3))$ is in the graph of $\wedge$ on
  $\{0, 1\}$.  It follows that $(h_1, h_2, h_3)$ is in
  $\RR_{\intT}^*$.

  Equivalently, $(*)$ means that $h_1$, $h_2$, $h_3$ are the
  characteristic maps $\chi_{U_1}$, $\chi_{U_2}$, $\chi_{U_3}$ of open
  subsets $U_1$, $U_2$, $U_3$ of $\Z_\bot$ such that: $(**)$ for all
  $n_1, n_2, n_3 \in \Z_\bot$ such that $n_3 = n_1 \wedge n_2$,
  $n_3 \in U_3$ if and only if $n_1 \in U_1$ and $n_2 \in U_2$.
  Clearly, any of the cases~1--4 implies $(**)$.

  Let us assume that $(**)$ holds.  By taking $n_1=n_2=n_3$, we obtain
  that $U_3 = U_1 \cap U_2$.  If $U_1$ is empty, then $U_3$ is empty
  and we are in case~3.  If $U_2$ is empty, then $U_3$ is empty and we
  are in case~4.  Henceforth, let us assume that $U_1$ and $U_2$ are
  non-empty.  If $\bot \in U_1$, then pick any $n_2 \in U_2$: we can
  then take $n_3=\bot$, so $\bot$ is in $U_3$ by $(**)$; this implies
  that $U_3=\Z_\bot$, hence also $U_1=U_2=\Z_\bot$, since
  $U_3 = U_1 \cap U_2$; hence we are in case~1.  We reason similarly
  if $\bot$ is in $U_2$.  It remains to examine the cases where $U_1$
  and $U_2$ are non-empty subsets of $\Z$.  If there are two distinct
  elements $n_1 \in U_1$ and $n_2 \in U_2$, then
  $n_3 = n_1 \wedge n_2$ is equal to $\bot$ and must be in $U_3$ by
  $(**)$, so $U_3=\Z_\bot$, and again $U_1=U_2=\Z_\bot$, meaning that
  we are in case~1.  Otherwise, $U_1=U_2=\{n\}$ for some $n \in \Z$,
  then $U_3=\{n\}$ as well, and we are in case~2.  \qed
\end{proof}

\begin{lemma}
  \label{lemma:P:seq}
  For every ground CBPV$(\Demon,\Nature)$ term
  $P \colon \intT \to \intT \to \F{}\intT$ such that
  $\Eval P (\bot) (0)=\Eval P (0) (\bot)=\{0\}$, the equality $\Eval P
  (\bot) (\bot) = \{0\}$ holds.
\end{lemma}
\begin{proof}
  Let $Q \in \Eval {\F{} \intT}$ be such that $(\{0\}, \{0\}, Q)$ is
  in $\RR_{\F{}\intT}$.  Consider any triple
  $(\chi_{U_1}, \chi_{U_2}, \chi_{U_3}) \in \RR_{\intT}^*$, as given
  in Lemma~\ref{lemma:RR:h:por}.  By definition, and recalling that
  $\underline\tr$ is the graph of the infimum map,
  $\extd {\chi_{U_3}} (Q) = \extd {\chi_{U_1}} (\{0\}) \wedge \extd
  {\chi_{U_2}} (\{0\}) = \chi_{U_1} (0) \wedge \chi_{U_2} (0)$.  If
  $Q$ were empty, then $\extd {\chi_{U_3}} (Q)$ would be equal to $1$,
  so $\chi_{U_1} (0)$ and $\chi_{U_2} (0)$ would be equal to $1$, and
  that is contradicted by the case~1 triple $U_1=U_2=U_3=\Z_\bot$ for
  example.  By considering the case~2 triple $U_1=U_2=U_3=\{0\}$, we
  obtain that $\extd {\chi_{\{0\}}} (Q)=1$, namely that
  $Q \subseteq \{0\}$.  Therefore the only $Q \in \Eval {\F{} \intT}$
  such that $(\{0\}, \{0\}, Q) \in \RR_{\F{}\intT}$ is $\{0\}$.  (One
  can also check that $(\{0\}, \{0\}, \{0\})$ is indeed in
  $\RR_{\F{}\intT}$, but that will not be needed.)

  By Lemma~\ref{lemma:tr:basic}, item~1, for all
  $(m_1, m_2, m_3) \in \RR_{\intT}$ and
  $(n_1, n_2, n_3) \in \RR_{\intT}$, the triple
  $(\Eval P (m_1) (n_1), \Eval P (m_2) (n_2), \Eval P (m_3) (n_3))$ is
  in $\RR_{\F{}\intT}$.  The triples $(0, \bot, \bot)$ and
  $(\bot, 0, \bot)$ are in $\RR_{\intT}$.  Hence
  $(\Eval P (0) (\bot), \Eval P (\bot) (0), \Eval P (\bot) (\bot))$ is
  in $\RR_{\F{}\intT}$.  Explicitly,
  $(\{0\}, \{0\}, \Eval P (\bot) (\bot))$ is in $\RR_{\F{}\intT}$.  We
  have just seen that this implies $\Eval P (\bot) (\bot) = \{0\}$.
  \qed
\end{proof}

We introduce the following abbreviations.
\begin{itemize}
\item $\Omega_\sigma$ denotes $\rec x_\sigma . x_\sigma$ for
  every value type $\sigma$.  We have $\Eval {\Omega_\sigma} = \bot$.
\item $\Omega_{\underline\tau}$ denotes
  $\force \Omega_{\U\underline\tau}$,
  for every computation type $\underline\tau$.  We have
  $\Eval {\Omega_{\underline\tau}} = \bot$.
\item For all $M \colon \F{}\intT$ and $N \colon \F{} \unitT$,
  $M\eq \underline 0\mathbin{\&} N$ abbreviates
  $\pto M {x_{\intT}} {\ifz {x_{\intT}} N {\Omega_{\F{}\unitT}}}$, where
  $x_{\intT}$ is not free in $N$.
  $\Eval {M \eq \underline 0 \mathbin{\&} N} \rho$ is equal to
  $\Eval N \rho$ if $\Eval M \rho = \{0\}$, to $\emptyset$ if
  $\Eval M \rho = \emptyset$, and to $\bot$ in all other cases.
\item Similarly, $M \eq \underline 1\mathbin{\&} N$ abbreviates
  $\pto M {x_{\intT}} {\ifz {(\pred x_{\intT})} N
    {\Omega_{\F{}\unitT}}}$, so that
  $\Eval {M \eq \underline 1 \mathbin{\&} N} \rho$ is equal to
  $\Eval N \rho$ if $\Eval M \rho = \{1\}$, to $\emptyset$ if
  $\Eval M \rho = \emptyset$, and to $\bot$ in all other cases.
\item Finally, for all $M, N \colon \F{}\unitT$, let $M \mathbin{\&} N$
  abbreviate $\pto M {x_{\intT}} N$, where $x_{\intT}$ is not free in
  $N$, so that $\Eval {M \mathbin{\&} N} \rho$ is equal to
  $\Eval N \rho$ if $\Eval M \rho=\{\top\}$ or if
  $\Eval M \rho = \Sierp$, to $\emptyset$ if $\Eval M \rho=\emptyset$
  and to $\bot$ if $\Eval M \rho=\bot$.
\end{itemize}
We let $\mathbin{\&}$ associate to the right, so $A \mathbin{\&} B
\mathbin{\&} C$ means $A \mathbin{\&} (B \mathbin{\&} C)$.

\begin{proposition}
  \label{prop:RR:por}
  For every term $P \colon \U{} (\intT \to \intT \to \F{}\intT)$, let:
  \begin{align*}
    M(P) & = \force P (\Omega_{\intT}) (\underline 0) \eq \underline 0 \mathbin{\&}
           \force P (\underline 0) (\Omega_{\intT}) \eq \underline 0
           \mathbin{\&} \produce \underline* \\
    N(P) & = M(P) \mathbin{\&} \force P (\Omega_{\intT})
           (\Omega_{\intT}) \eq \underline 0
           \mathbin{\&} \produce \underline*, \\
  \end{align*}
  We also define 
  $M$ as $\lambda g . M (g)$, and $N$ as $\lambda g . N (g)$, where
  $g$ has type $\U{} (\intT \to \intT \to \F{}\intT)$.

  In CBPV$(\Demon, \Nature)$,
  $M \precsim_{\U{} (\intT \to \intT \to \F{}\intT) \to \F{}\unitT} N$, but
  $\Eval M \not\leq \Eval N$.
\end{proposition}
\begin{proof}
  $\Eval M$ applied to any Scott-continuous map
  $G \colon \Z_\bot \to \Z_\bot \to \Smyth^\top_\bot (\Z_\bot)$
  returns:
  \begin{itemize}
  \item $\{\top\}$ if $G (\bot) (0)=G (0) (\bot)=\{0\}$;
  \item $\emptyset$ if $G (\bot) (0)=\emptyset$ or if
    $G (\bot) (0)=\{0\}$ and $G (0) (\bot)=\emptyset$;
  \item and $\bot$ in all other cases.
  \end{itemize}
  Then
  $\Eval N$ applied to $G$ returns:
  \begin{itemize}
  \item $\{\top\}$ if $\Eval M (G)=\{\top\}$ and
    $G (\bot) (\bot) = \{0\}$;
  \item $\emptyset$ if $\Eval M (G)=\{\top\}$ and
    $G (\bot) (\bot)=\emptyset$;
  \item $\emptyset$ if $\Eval M (G)=\emptyset$;
  \item $\bot$ in all other cases.
  \end{itemize}
  In particular, $\Eval M \not\leq \Eval N$: defining $G$ to be the
  parallel or map ($G (0)(n)=G(n)(0)=\{0\}$ for every $n \in \Z_\bot$,
  $G (1)(1)=\{1\}$, $G (m)(n)=\bot$ for all
  $m, n \in \Z_\bot \diff \{0\}$ such that $(m,n)\neq (1,1)$),
  $\Eval M (G) = \{\top\}$, but $\Eval N (G) = \bot$.  Note, by the
  way, that the argument would also work with other choices of map
  $G$, for example $G (0)(n)=G(n)(0)=\{0\}$ for every $n \in \Z_\bot$,
  and $G(m)(n)=\bot$ in all other cases.

  For every ground CBPV$(\Demon,\Nature)$ term
  $P \colon \intT \to \intT \to \F{}\intT$, $\Eval {M(P)}=\{\top\}$ if
  and only if $\Eval P (\bot) (0) = \Eval P (0) (\bot) = \{0\}$, and
  if so, $\Eval P (\bot) (\bot) = \{0\}$ by Lemma~\ref{lemma:P:seq}.
  Taking $G = \Eval P$, it follows that the second case of the
  definition of $\Eval N (G)$ does not occur, so
  $\Eval {N (P)} = \Eval N (G)$ is equal to $\{\top\}$ if
  $\Eval {M(P)}=\{\top\}$ (and then $G (\bot) (\bot) = \{0\}$ is
  automatic), to $\emptyset$ if $\Eval {M(P)} = \emptyset$, and to
  $\bot$ in all other cases.  Hence $\Eval {N(P)} = \Eval {M(P)}$.
  Since in particular $\Eval {M(P)} \leq \Eval {N(P)}$, by
  Proposition~\ref{prop:fa:easy}, $M(P) \precsim_{\F{}\unitT} N(P)$.
  Since $P$ is arbitrary,
  $M \precsim_{\U{} (\intT \to \intT \to \F{}\intT) \to \F{}\unitT} N$ by
  Proposition~\ref{prop:precsim:lambda}.  \qed
\end{proof}

\subsection{The Need for Statistical Termination Testers}
\label{sec:need-stat-term}

We turn to justify the need for a $\bigcirc$ operator, following
similar ideas as in \cite[Proposition~8.5]{jgl-jlap14}.

Here we let $I = \{1, 2\}$, $\mathcal J = \{\emptyset, I\}$,
$\tr = \{(a_1, a_2) \in [0, 1]^2 \mid a_1 + 1 \geq 2a_2\}$, and
$\underline\tr = \{(a_1, a_2) \in [0, 1]^2 \mid a_1 \geq a_2\}$.

In that case, $(h_1, h_2) \in \RR_{\unitT}^\perp$ if and only if for every
$b \in \Sierp$, $h_1 (b) +1 \geq 2h_2 (b)$.  Letting $\alpha_i = h_i
(\bot)$ and $\beta_i = h_i (\top)$, $i \in \{1, 2\}$, $(h_1, h_2) \in
\RR_{\unitT}^\perp$ if and only if:
\begin{align}
  \label{eq:a}
  \alpha_1 + 1 &\geq 2\alpha_2 \\
  \label{eq:b}
  \beta_1 + 1 &\geq 2 \beta_2 \\
  \label{eq:1}
  1 \geq \beta_1 &\geq \alpha_1 \geq 0 \\
  \label{eq:2}
  1 \geq \beta_2 &\geq \alpha_2 \geq 0.
\end{align}




This defines a (bounded) polytope of $\real^4$, and we claim that its
vertices are:
\begin{equation}
  \label{eq:RR:vertices}
  \begin{array}{cccc}
    \alpha_1 & \beta_1 & \alpha_2 & \beta_2 \\
    \hline
    0 & 0 & 0 & 0 \\
    0 & 0 & 0 & 1/2 \\
    0 & 0 & 1/2 & 1/2 \\
    0 & 1 & 0 & 0 \\
    0 & 1 & 0 & 1 \\
    0 & 1 & 1/2 & 1/2 \\
    0 & 1 & 1/2 & 1 \\
    1 & 1 & 0 & 0 \\
    1 & 1 & 0 & 1 \\
    1 & 1 & 1 & 1 \\
  \end{array}
\end{equation}
We check that those points satisfy all the given inequalities.
Conversely, let $(\alpha_1, \beta_1, \alpha_2, \beta_2)$ satisfy
(\ref{eq:a})--(\ref{eq:2}).  Because it satisfies (\ref{eq:1}),
$(\alpha_1, \beta_1)$ is a linear convex combination
$a (0, 0) + b (0, 1) + c (1, 1)$, where $a, b, c \geq 0$ and
$a+b+c=1$, and the remaining inequalities become
$c+1 \geq 2 \alpha_2$, $b+c+1 \geq 2 \beta_2$,
$1 \geq \beta_2 \geq \alpha_2 \geq 0$, whose solutions in
$(\alpha_2, \beta_2)$ are the convex combinations of $(0, 0)$,
$(0, (b+c+1)/2)$, $((c+1)/2, (c+1)/2)$, and $((c+1)/2, (b+c+1)/2)$ (as
a two-dimensional picture will show), say
$a' (0, 0) + b' (0, (b+c+1)/2) + c' ((c+1)/2, (c+1)/2) + d' ((c+1)/2,
(b+c+1)/2)$, with $a', b', c', d' \geq 0$ and $a'+b'+c'+d'=1$.  Since
the latter is affine in $a$, $b$ and $c$, it follows that the
solutions of (\ref{eq:a})--(\ref{eq:2}) are all of the form $a$ times
$(0, 0) . (a' (0, 0) + b' (0, 1/2) + c' (1/2, 1/2) +d' (1/2, 1/2))$,
plus $b$ times
$(0, 1) . (a' (0, 0) + b' (0, 1) + c' (1/2, 1/2) + d' (1/2, 1))$, plus
$c$ times $(1, 1) . (a' (0, 0) + b' (0, 1) + c' (1, 1) + d' (1, 1))$
(where $.$ denotes concatenation of tuples, i.e.,
$(x,y) . (z,t)=(x,y,z,t)$), hence a a convex combination of the 10
tuples of the above table.

  

\begin{lemma}
  \label{lemma:RR:delta}
  For all $a_1, a_2 \in [0, 1]$,
  $(a_1 \delta_\top, a_2 \delta_\top) \in \RR_{\Vt{} \unitT}$ if and only
  if $a_1+1 \geq 2a_2$.
\end{lemma}
\begin{proof}
  We have $(a_1 \delta_\top, a_2 \delta_\top) \RR_{\Vt{} \unitT}$ if and only
  if for all $(h_1, h_2) \in \RR_{\unitT}^\perp$, $a_1 h_1 (\top) + 1
  \geq 2 a_2 h_2 (\top)$.  Writing $h_1$ and $h_2$ as above, the
  domain of variation of $(\alpha_1, \beta_1, \alpha_2, \beta_2)$ is
  the convex hull of the 10 points in (\ref{eq:RR:vertices}), and we
  must check that $a_1 \beta_1 + 1 \geq 2 a_2 \beta_2$ for all those
  $4$-tuples.  The domain of variation of the pairs $(\beta_1,
  \beta_2)$ alone is the convex hull of $(0, 0)$, $(0, 1/2)$, $(1,
  0)$, $(1, 1)$ (and $(1, 1/2)$, which is already a convex combination
  of the others).  By linearity, it is equivalent to check $a_1
  \beta_1 + 1 \geq 2 a_2 \beta_2$ for just those four values of
  $(\beta_1, \beta_2)$.  Therefore $(a_1 \delta_\top, a_2 \delta_\top)
  \RR_{\Vt{} \unitT}$ if and only if $1 \geq 0$, $1 \geq a_2$, $a_1 + 1
  \geq 0$, and $a_1 + 1 \geq 2 a_2$.  Since the first three are always
  true, only the last one remains.  \qed
%
%
\end{proof}

\begin{lemma}
  \label{lemma:RR:nu}
  Let $\nu_1, \nu_2 \in \Eval {\Vt{} \unitT}$.  If $(\nu_1, \nu_2) \in
  \RR_{\Vt\unitT}$ then $\nu_1 + \delta_\top \geq 2 \nu_2$.
\end{lemma}
\begin{proof}
  For every $(h_1, h_2) \in \RR_{\unitT}^\perp$,
  $\int_{x \in \Sierp} h_1 (x) d\nu_1 + 1 \geq \int_{x \in \Sierp} h_2
  (x) d\nu_2$.  Considering the case where
  $h_1 \colon \bot \mapsto \alpha_1, \top \mapsto \beta_1$ and
  $h_2 \colon \bot \mapsto \alpha_2, \top \mapsto \beta_2$ are given
  by the data of the 5th row of (\ref{eq:RR:vertices}), we obtain
  $\nu_1 (\{\top\}) + 1 \geq 2 \nu_2 (\{\top\})$.  Considering the
  last row instead, we obtain
  $\nu_1 (\{\bot, \top\}) + 1 \geq 2 \nu_2 (\{\bot, \top\})$.  The
  inequality
  $\nu_1 (\emptyset) + \delta_\top (\emptyset) \geq 2 \nu_2
  (\emptyset)$ is obvious.  \qed
\end{proof}

\begin{lemma}
  \label{lemma:RR:h}
  Let $k$ be any Scott-continuous map from $\Val_{\leq 1} \Sierp$ to
  $[0, 1]$.  Then $(k (\frac 1 2 \_ + \frac 1 2 \delta_\top), k) \in \RR_{\Vt\unitT}^*$.
\end{lemma}
\begin{proof}
  It suffices to verify that for all
  $(\nu_1, \nu_2) \in \RR_{\Vt\unitT}$,
  $(k (\frac 1 2 \nu_1 + \frac 1 2 \delta_\top), k (\nu_2))$ is in
  $\underline\tr$, namely that
  $k (\frac 1 2 \nu_1 + \frac 1 2 \delta_\top) \geq k (\nu_2)$.  By
  Lemma~\ref{lemma:RR:nu},
  $\frac 1 2 \nu_1 + \frac 1 2 \delta_\top \geq \nu_2$, and we
  conclude since $k$, being Scott-continuous, is monotonic.  \qed
\end{proof}

\begin{proposition}
  \label{prop:RR:test}
  Let:
  \begin{align*}
    M & = \lambda g . \force g (\Omega_{\Vt\unitT} \oplus
        \retkw \underline*), \\
    N & = \lambda g .
        \pto {(\force g (\Omega_{\Vt\unitT}))} {y_{\Vt\unitT}} {\produce (y_{\Vt\unitT} \oplus
        \retkw \underline*)}.
  \end{align*}
  where $g$ has type $\U{} (\Vt{} \unitT \to \F\Vt\unitT)$.

  In CBPV$(\Demon, \Nature)$ and also in
  CBPV$(\Demon, \Nature)+\pifzkw$,
  $M \precsim_{\U{} (\Vt{} \unitT \to \F\Vt\unitT) \to \F\Vt\unitT} N$, but
  $\Eval M \not\leq \Eval N$.
\end{proposition}
\begin{proof}
  Let $P$ be any ground CBPV$(\Demon, \Nature)$ or
  CBPV$(\Demon, \Nature)+\pifzkw$ term of type
  $\U{} (\Vt\unitT \to \F\Vt\unitT)$, and:
  \begin{align*}
    M(P) & = \force P (\Omega_{\Vt\unitT} \oplus \retkw \underline*) \\
    N(P) & = \pto {(\force P (\Omega_{\Vt\unitT}))} {y_{\Vt\unitT}} {\produce (y_{\Vt\unitT} \oplus
           \retkw \underline*)}.
  \end{align*}
  We have:
  \begin{align*}
    \Eval {M(P)}  & = \Eval P  (\frac 1 2 \delta_\top) \\
    \Eval {N(P)}  & = \extd g (\Eval P (0)),
  \end{align*}
  where $g (\nu) = \upc (\frac 1 2 \nu + \frac 1 2 \delta_\top)$.  By
  Lemma~\ref{lemma:RR:delta}, $(0, \frac 1 2 \delta_\top)$ is in
  $\RR_{\Vt\unitT}$.  By Lemma~\ref{lemma:tr:basic} (item~2),
  $(\Eval P, \Eval P)$ is in $\RR_{\Vt{} \unitT \to \F\Vt\unitT}$, so
  $(\Eval P (0), \Eval P (\frac 1 2 \delta_\top))$ is in
  $\RR_{\F\Vt\unitT}$.  Using Lemma~\ref{lemma:RR:h}, for every
  Scott-continuous map $k \colon \Eval {\Vt\unitT} \to [0, 1]$,
  $(k (\frac 1 2 \_ + \frac 1 2 \delta_\top), k) \in \RR_{\Vt\unitT}^*$,
  so
  $(\extd {(k (\frac 1 2 \_ + \frac 1 2 \delta_\top))} (\Eval P (0)),
  \extd k (\Eval P (\frac 1 2 \delta_\top)))$ is in $\underline\tr$.
  In other words,
  $\extd {(k (\frac 1 2 \_ + \frac 1 2 \delta_\top))} (\Eval P (0))
  \geq \extd k (\Eval P (\frac 1 2 \delta_\top))$.
    
  Since
  $g = \eta^\Smyth \circ (\nu \mapsto \frac 1 2 \nu + \frac 1 2
  \delta_\top)$, and using Proposition~\ref{prop:Qtop}, item~2,
  $\extd k \circ g = (k (\frac 1 2 \_ + \frac 1 2 \delta_\top))$.  It
  follows that $\extd k \circ \extd g = \extd {(\extd k \circ g)}$
  (using Proposition~\ref{prop:Qtop}, item~4)
  $=\extd {(k (\frac 1 2 \_ + \frac 1 2 \delta_\top))}$.  Hence our
  previous equality can be read, alternatively, as
  $\extd k (\extd g (\Eval P (0))) \geq \extd k (\Eval P (\frac 1 2
  \delta_\top))$.  In other words,
  \begin{align}
    \label{eq:k}
    \extd k (\Eval {N(P)}) & \geq \extd k (\Eval {M(P)})
  \end{align}
  for every Scott-continuous map
  $k \colon \Eval {\Vt\unitT} \to [0, 1]$.
  
  We claim that $M (P) \precsim_{\F\Vt\unitT} N(P)$.  To this end, we let
  $C$ be any ground evaluation context of type
  $\F\Vt\unitT \vdash \F\Vt\unitT$, and we aim to show that
  $\Prob (C \cdot M(P) \vp) \leq \Prob (C \cdot N(P) \vp)$.  By
  adequacy (Proposition~\ref{prop:adeq}) and
  Lemma~\ref{lemma:prob:disc:_}, this means showing that
  $\extd h (\Eval C (\Eval {M(P)})) \leq \extd h (\Eval C (\Eval
  {N(P)}))$, where $h (\nu) = \nu (\{\top\})$.

  Let us write $C$ as $E_0 E_1 E_2 \cdots E_n$, where
  $E_i \colon \anytype_{i+1} \vdash \anytype_i$,
  $\anytype_{n+1}=\anytype_0 = \F\Vt\unitT$.  All the types $\anytype_i$
  have rank $1$, namely, are computation types.  It follows that
  $E_0 = [\_]$, and that the elementary contexts $E_i$,
  $1\leq i\leq n$ are of the form $[\pto {\_} {x_\tau} N]$ or
  $[\_ N]$.  However, $\anytype_{n+1}$ is an $\F{}$-type (i.e., of the
  form $\F{} \tau$ for some value type $\tau$), and that implies that
  $E_n$ must be of the form $[\pto {\_} {x_\tau} N]$ and $\anytype_n$
  must be an $\F{}$-type again.  Then $E_{n-1}$ must again be of the form
  $[\pto {\_} {x_\tau} N]$ and $\anytype_{n-1}$ must be an $\F{}$-type,
  and so on: the elementary contexts $E_i$, $1\leq i\leq n$, are all
  of the form $[\pto {\_} {x_\tau} N]$, and $\anytype_i = \F{} \tau_i$
  for some value type $\tau_i$, with $\tau_{n+1} = \tau_1 = \Vt\unitT$.
  In particular, $\Eval {E_i} = \extd {f_i}$ for some Scott-continuous
  map $f_i \colon \Eval {\Vt{} \tau_{i+1}} \to \Eval {\Vt{} \tau_i}$.  It
  follows that
  $\Eval C = \extd {f_1} \circ \extd {f_2} \circ \cdots \circ \extd
  {f_n}$.  If $n \neq 0$, then applying Proposition~\ref{prop:Qtop},
  item~4, repeatedly, we obtain that $\Eval C = \extd f$ for some
  Scott-continuous map
  $f \colon \Eval {\Vt\unitT} \to \Eval {\F\Vt\unitT}$; applying it one
  more time, $\extd h \circ \Eval C = \extd {(\extd h \circ f)}$.  If
  $n=0$, then $\extd h \circ \Eval C = \extd h$.  In both cases,
  $\extd h \circ \Eval C$ is equal to $\extd k$ for some
  Scott-continuous map $k \colon \Eval {\Vt{} \unitT} \to [0, 1]$.

  By (\ref{eq:k}), $\extd h (\Eval C (\Eval {N(P)})) \geq \extd h
  (\Eval C (\Eval {M(P)}))$, and this is what we needed to show
  to establish $M (P) \precsim_{\F\Vt\unitT} N(P)$.

  Since $P$ is arbitrary, by Proposition~\ref{prop:precsim:lambda},
  $M \precsim_{\U{} (\Vt{} \unitT \to \F\Vt\unitT) \to \F\Vt\unitT} N$.

  For every $b \in (0, 1)$, let $[> b]$ be the map that sends every
  $\nu \in \Eval {\Vt\unitT}$ to $\upc \{\delta_\top\}$ if
  $\nu (\{\top\}) > b$, and to $\bot$ otherwise.  This is is easily
  seen to be Scott-continuous.  For $b < 1/2$ (e.g., $b=1/4$),
  \begin{align*}
    \Eval M ([> b]) & = [> b] (\frac 1 2 \delta_\top) = \upc
                      \{\delta_\top\} \\
    \Eval N ([> b]) & = \extd g ([> b] (0)) = \extd g (\bot) = \bot.
  \end{align*}
  In particular, $\Eval M ([> b]) \not\leq \Eval N ([> b])$, so $\Eval
  M \not\leq \Eval N$.  \qed
\end{proof}
The function $[> b]$ is, of course, the semantics of $\bigcirc_{> b}$.
As a consequence, it is not definable in PCBV$(\Demon, \Nature)$ and
even CBPV$(\Demon, \Nature)+\pifzkw$, at least for $b < 1/2$.  A
similar argument would show that it is not definable for any
$b \in (0, 1)$, replacing the definition of $\tr$ by
$\tr = \{(a_1, a_2) \in [0, 1]^2 \mid a a_1 + 1-a \geq a_2\}$, for any
dyadic number $a \in (0, 1)$.

\section{Full Abstraction}
\label{sec:full-abstraction}

Full abstraction for CBPV$(\Demon,\Nature)+\pifzkw+\bigcirc$ will
follow from a series of auxiliary results that show that the Scott
topology on various dcpos coincides with some other, simpler
topologies.  Before we make that precise, let us say that our goal is
that every type should be \emph{describable}, in the following sense.
For a Scott-open subset $U$ of $\Eval \anytype$, where $\anytype$ is a
type, recall that $\chi_U \in [\Eval \anytype \to \Sierp]$ is its
characteristic map.  We write $\tilde\chi_U$ for the map
$\Eval {[\produce \retkw \_]} \circ \chi_U$, which maps every
$x \in \Eval \anytype$ to $\{\delta_\top\}$ if $x \in U$, to $\bot$
otherwise.

\begin{definition}
  \label{defn:descr}
  An element of $\Eval \anytype$, for a type $\anytype$, is \emph{definable} if and
  only if it is equal to $\Eval M$ for some ground
  CBPV$(\Demon,\Nature)+\pifzkw+\bigcirc$ term $M \colon \anytype$.

  A Scott-open subset $U$ of $\Eval \tau$, for a value type $\tau$, is
  \emph{definable} if and only if $\tilde\chi_U = \Eval M$ for some
  ground CBPV$(\Demon,\Nature)+\pifzkw+\bigcirc$ term
  $M \colon \tau \to \F\Vt\unitT$.

  For a computation type $\underline\tau$, a Scott-open subset $U$ of
  $\Eval {\underline\tau}$ is \emph{definable} if and only if
  $\tilde\chi_U = \Eval M$ for some ground
  CBPV$(\Demon,\Nature)+\pifzkw+\bigcirc$ term
  $M \colon U \underline\tau \to \F\Vt\unitT$.

  A type $\anytype$ is \emph{describable} if and only if $\Eval \anytype$ has a
  basis of definable elements and the Scott topology on $\Eval \anytype$ has
  a subbase of definable open subsets.
\end{definition}



As a first, easy example of a describable type, we have:
\begin{lemma}
  \label{lemma:descr:unit}
  $\unitT$ is describable.
\end{lemma}
\begin{proof}
  All the elements of $\Eval {\unitT} = \Sierp$ are definable, since
  $\Eval {\Omega_{\unitT}} = \bot$ and $\Eval {\underline *} = \top$.
  The function
  $P = \lambda x_{\unitT} . (x_{\unitT}; \produce \retkw \underline *)$
  defines the open subset $\{\top\}$, which is by itself a subbase of
  the Scott topology.  \qed
\end{proof}


\begin{lemma}
  \label{lemma:desc:int}
  $\intT$ is describable.
\end{lemma}
\begin{proof}
  All the elements of $\Eval {\intT} = \Z_\bot$ are definable:
  $\Eval {\Omega_{\intT}} = \bot$, $\Eval {\underline n} = n$.  A
  subbase of the Scott topology consists of the sets $\{n\}$,
  $n \in \Z$, and they are definable by
  $\lambda x_{\intT} . \ifz {\underbrace{\pred (\pred \cdots
      (\pred}_{n\text{ times}} M))} {(\produce \retkw \underline *)}
  {\Omega_{\F\Vt\unitT}}$ if $n \geq 0$, and by
  $\lambda x_{\intT} . \ifz {\underbrace{\suc (\suc \cdots
      (\suc}_{n\text{ times}} M))} {(\produce \retkw \underline *)}
  {\Omega_{\F\Vt\unitT}}$ otherwise.  \qed
\end{proof}

We now consider more complex types.  It will be useful to realize that
every describable type has a base, not just a subbase, of definable
open subsets.  Moreover this base, which is obtained as the collection
of finite intersections of subbasic open sets, is closed under finite
intersections.  We call \emph{strong base} any base that is closed
under finite intersections.

\begin{lemma}
  \label{lemma:descr:base}
  For every describable type $\anytype$, the Scott topology on $\Eval \anytype$ has
  a strong base of definable open subsets.
\end{lemma}
\begin{proof}
  For any two terms $M, N \colon \F\Vt\unitT$, let $M \wedge N$ be the
  term $\pto M {x_{\Vt\unitT}} N$, where $x_{\Vt\unitT}$ is a fresh
  variable.  For every environment $\rho$,
  $\Eval {M \wedge N} \rho = \extd h (\Eval M \rho)$ where
  $h (\nu) = \Eval N \rho [x_{\Vt\unitT} \mapsto \nu] = \Eval N \rho$
  (since $x_{\Vt\sigma}$ is not free in $N$).  Since $\extd h$ is
  strict, if $\Eval M \rho=\bot$, then
  $\Eval {M \wedge N} \rho = \bot$.  Otherwise, by
  Proposition~\ref{prop:Qtop}, item~2,
  $\extd h (\Eval M \rho) = \bigwedge_{\nu \in \Eval M \rho} h (\nu)$.
  If $\Eval M \rho \neq \emptyset$, in particular if
  $\Eval M = \{\delta_\top\}$, this is equal to $\Eval N \rho$.

  We write $M_1 \wedge \cdots \wedge M_n$ for
  $M_1 \wedge (M_2 \wedge \cdots (M_n \wedge \produce \retkw
  \underline*) \cdots)$.  This implements logical and, in the sense
  that if $\Eval {M_i} \rho$ is either equal to $\{\delta_\top\}$ or
  to $\bot$ for every $i$, its denotation in any environment $\rho$ is
  $\{\delta_\top\}$ if $\Eval {M_i} \rho = \{\delta_\top\}$ for every
  $i$, and is $\bot$ if $\Eval {M_i} \rho = \bot$ for some $i$.
  
  Given finitely many open subsets $U_1$, \ldots, $U_n$ defined by
  terms $M_1, \cdots, \allowbreak M_n \colon \sigma \to \F\Vt\unitT$
  respectively (where $\sigma=\anytype$ if $\anytype$ is a value type,
  $\sigma = \U\anytype$ if $\anytype$ is a computation type), the term
  $\lambda x_\sigma . (M_1 x_\sigma) \wedge \cdots \wedge (M_n
  x_\sigma)$ then defines the intersection $U_1 \cap \cdots \cap U_n$.
  \qed
\end{proof}

\subsection{Product types}
\label{sec:product-types}

\begin{lemma}
  \label{lemma:subbase:prod}
  Let $X$, $Y$ be two continuous dcpos.  Let $B_X$ be a basis of $X$,
  $B_Y$ be a basis of $Y$, $\mathcal S_X$ be a subbase of the Scott
  topology on $X$, $\mathcal S_Y$ be a subbase of the Scott topology
  on $Y$.  Then:
  \begin{itemize}
  \item The set $B_{X \times Y} = B_X \times B_Y$ is a basis of $X
    \times Y$.
  \item The set
    $\mathcal S_{X \times Y} = \{U \times V \mid U \in \mathcal S_X, V
    \in \mathcal S_Y\}$ is a subbase of the Scott topology on
    $X \times Y$.
  \end{itemize}
\end{lemma}
\begin{proof}
  The second part follows from the fact that the Scott topology on a
  product of continuous dcpos is the product topology, because it is
  generated by sets of the form
  $\uuarrow (x, y) = \uuarrow x \times \uuarrow y$.  (This is not true
  of non-continuous dcpos.)  \qed
\end{proof}

\begin{proposition}
  \label{prop:descr:prod}
  For any two describable value types $\sigma$ and $\tau$, $\sigma
  \times \tau$ is describable.
\end{proposition}
\begin{proof}
  We use Lemma~\ref{lemma:subbase:prod} with $X = \Eval \sigma$, $Y =
  \Eval \tau$.   $B_X$ (resp., $B_Y$) is the basis of
  definable elements of $\Eval \sigma$ (resp.,
  $\Eval {\tau}$).  $\mathcal B_X$ is the base of
  definable open subsets at type $\sigma$, obtained by
  Lemma~\ref{lemma:descr:base}, and similarly for $\mathcal B_Y$.
  The elements of $B_{X \times Y}$ are definable as $\langle M, N
  \rangle$, where $\Eval M \in B_X$, $\Eval N \in B_Y$, and the
  elements $U \times V$ of $\mathcal B_{X \times Y}$ are definable as
  $\lambda z_{\sigma \times \tau} . M (\pi_1 z) \land N (\pi_2 z)$,
  where $U$ is defined by $M$ and $V$ is defined by $N$, and where
  $\wedge$ was defined in the course of the proof of
  Lemma~\ref{lemma:descr:base}.  \qed
\end{proof}

\subsection{Function types}
\label{sec:arrow-types}

Semantically, at function types, the key result will be the following
Proposition~\ref{prop:subbase:fun}, which in particular says that the
Scott topology on $[X \to Y]$ coincides with the topology of pointwise
convergence, under certain assumptions.

A standard basis of $[X \to Y]$ is given by the \emph{step functions}
$\sup_{i=1}^m U_i \searrow b_i$, where each $U_i$ is open in $X$, each
$b_i$ is in $Y$, and $U \searrow b$ denotes the map that maps every
element of $U$ to $b$, and all others to $\bot$.  We show that this
can be refined by requiring $U_i$ to be taken from some given strong
base $\mathcal B_X$ of the topology on $X$, and $b_i$ to be taken from
some basis $B_Y$ of $Y$.  We note that $\sup_{i=1}^m U_i \searrow b_i$
maps each point $x \in X$ to $\sup_{i \in I} b_i$, where
$I = \{i \mid 1\leq i\leq m, x \in U_i\}$.  In general,
$\sup_{i \in I} b_i$ will not be in $B_Y$.  To avoid this problem, we
require our step functions to be of a special form.
\begin{definition}
  \label{defn:special:step}
  Let $X$ be a topological space, $\mathcal B_X$ be a strong base of
  the topology of $X$, $Y$ be a continuous dcpo, and $B_Y$ be a basis
  of $Y$.  A \emph{$(\mathcal B_X, B_Y)$-step function} is any step
  function of the form
  $\sup_{I \subseteq \{1, \cdots, m\}} U_I \searrow y_I$ where:
  \begin{enumerate}
  \item each $U_I$ is in $\mathcal B_X$;
  \item each $y_I$ is in $B_Y$;
  \item $U_\emptyset = X$ and $U_I \cap U_J = U_{I \cup J}$ for all
    $I, J \subseteq \{1, \cdots, m\}$;
  \item for all $I \subseteq J$, $y_I \leq y_J$.
  \end{enumerate}
\end{definition}

We make a preliminary remark.
\begin{lemma}
  \label{lemma:rem:basis}
  Given a continuous dcpo $Z$, and a
  family $B \subseteq Z$, in order to show that $B$ is a basis of $Z$
  it is enough to show that for every $z \in Z$, every Scott-open
  neighborhood $W$ of $z$ contains a $d \in B$ such that $d \ll z$.
\end{lemma}
\begin{proof}
  If so, then the family $B_z = \{d \in B \mid d \ll z\}$ is non-empty
  (take $W=Z$) and directed (for any two $d_1, d_2 \in B_z$, take
  $W = \uuarrow d_1 \cap \uuarrow d_2$), and $\sup B_z = z$ (for every
  open neighborhood $W$ of $z$, some element of $B_z$ is in $W$ so
  $\sup B_z \geq z$, and the converse inequality is obvious).  \qed
\end{proof}

A \emph{core-compact} topological space $X$ is one whose lattice of
open subsets is a continuous dcpo.  We write $\Subset$ for the
way-below relation on that lattice.  Every locally compact space is
core-compact, with $U \Subset V$ if and only if
$U \subseteq Q \subseteq V$ for some compact saturated set $Q$.
\begin{lemma}
  \label{lemma:step:fun}
  Let $X$ be a core-compact space, $\mathcal B_X$ be a strong base of
  the topology of $X$, $Y$ be a a continuous complete lattice, and
  $B_Y$ be a basis of $Y$.  Then $[X \to Y]$ is a continuous complete
  lattice, with a basis of $(\mathcal B_X, B_Y)$-step functions.
\end{lemma}
Note: one could replace ``continuous complete lattice'' by
``bc-domain'' here, and the proof would only be slightly more
complicated.

\begin{proof}
  We apply Lemma~\ref{lemma:rem:basis} to $Z = [X \to Y]$.  By
  Proposition~2 of \cite{EEK:waybelow}, $Z$ is a bounded complete
  continuous dcpo with a basis $B_0$ of step functions, and since it
  has a top, it is a continuous complete lattice.  Let $B_1$ be the
  family of step functions of the form $\sup_{i=1}^m V_i \searrow y_i$
  where each $V_i$ is in $\mathcal B_X$.  For every $f \in [X \to Y]$,
  and every Scott-open neighborhood $W$ of $f$, there is an element
  $\sup_{i=1}^m U_i \searrow y_i$ of $B_0$, way-below $f$, and in $W$.
  Let us write $U_i$ as a union $\bigcup_{j \in J_i} V_{ij}$ of
  elements of $\mathcal B_X$.  The family of maps
  $\sup_{i=1}^m (\bigcup_{j \in F_i} V_{ij}) \searrow y_i$, where
  $F_i$ ranges over the finite subsets of $J_i$ for each $i$, is
  directed (since an upper bound of
  $\sup_{i=1}^m (\bigcup_{j \in F_i} V_{ij}) \searrow y_i$ and of
  $\sup_{i=1}^m (\bigcup_{j \in F'_i} V_{ij}) \searrow y_i$ is
  $\sup_{i=1}^m (\bigcup_{j \in F_i \cup F'_i} V_{ij}) \searrow y_i$),
  and has $\sup_{i=1}^m U_i \searrow y_i$ as supremum (because for
  every $x \in X$, letting $I = \{i \mid x \in U_i\}$, there are
  indices $j_i \in J_i$ for each $i \in I$ such that $x \in V_{ij}$,
  hence $x \in F_i$ where $F_i = \{j_i\}$).  Hence
  $\sup_{i=1}^m (\bigcup_{j \in F_i} V_{ij}) \searrow y_i$ is in $W$
  for some finite subsets $F_i$ of $J_i$, $1\leq i\leq m$.  Now
  $\sup_{i=1}^m (\bigcup_{j \in F_i} V_{ij}) \searrow y_i$ is equal to
  $\sup_{i=1}^m \sup_{j \in F_i} V_{ij} \searrow y_i$, showing that it
  is in $B_1$.  Moreover,
  $\sup_{i=1}^m (\bigcup_{j \in F_i} V_{ij}) \searrow y_i \leq
  \sup_{i=1}^m U_i \searrow y_i \ll f$.  We can therefore apply our
  preliminary remark and conclude that $B_1$ is a basis of
  $[X \to Y]$.

  Given any element $\sup_{i=1}^m U_i \searrow y_i$ of $B_1$, we can
  write it as $\sup_{I \subseteq \{1, \cdots, m\}} U_I \searrow y_I$,
  where for each $I \subseteq \{1, \cdots, m\}$,
  $U_I = \bigcap_{i \in I} U_i$ and $y_I = \sup_{i \in I} y_i$.  (In
  case $Y$ were a bc-domain, the same argument would apply provided we
  only considered the subsets $I$ such that $U_I$ is non-empty.)  Note
  that: $(a)$ for all $I, J \subseteq \{1, \cdots, m\}$, $I \subseteq J$
  implies $y_I \leq y_J$.  Also: $(b)$ $U_\emptyset = X$, for all
  $I, J \subseteq \{1, \cdots, m\}$, $U_I \cap U_J = U_{I \cup J}$,
  and each $U_I$ is in $\mathcal B_X$ (because $\mathcal B_X$ is a
  strong base).

  Let $B_2$ be the family of maps
  $\sup_{I \subseteq \{1, \cdots, m\}} U_I \searrow y_I$, where $U_I$
  and $y_I$ satisfy conditions $(a)$ and $(b)$ and, additionally, each
  $y_I$ is in $B_Y$.  For every $f \in [X \to Y]$, and every
  Scott-open neighborhood $W$ of $f$, $W$ contains an element
  $g = \sup_{I \subseteq \{1, \cdots, m\}} U_I \searrow y_I$ of $B_1$
  satisfying conditions $(a)$ and $(b)$ and way-below $f$.

  Here is the idea of the rest of the proof.  Enumerating the subsets
  $I$ of $\{1, \cdots, m\}$ so that the cardinality of $I$ never goes
  down, starting from the empty set, we replace $y_I$ by an element
  $z_I$ such that $z_I \ll y_I$ and $z_I \in B_Y$; at each step, we
  also replace $y_J$ by $\sup (z_I, y_J)$ for all strict supersets $J$
  of $I$, so that $(a)$ still holds.  Since $B_Y$ is a basis, for
  $z_I$ large enough, the resulting function will be in $W$.  We can
  also require that $z_J \leq z_I$ for all $J \subsetneq I$, since all
  those elements $z_J$ have been chosen in previous steps so that
  $z_J \ll y_J$.  At the end of the enumeration, we obtain a function
  $h = \sup_{I \subseteq \{1, \cdots, m\}} U_I \searrow z_I$ of $B_1$
  satisfying conditions $(a)$ and $(b)$, in $W$, below $g$ hence
  way-below $f$, and such that $z_I \in B_Y$ and $z_I \ll y_I$ for
  every $I \in \{1, \cdots, m\}$.  In particular, that element $h$ is
  in $B_2$.  Let us now prove that formally.

  We claim that: $(*)$ for every downwards-closed family $\mathcal I$
  of $\pow (\{1, \cdots, m\})$ (downwards-closed with respect to
  inclusion), there is an element $h$ of the form
  $\sup_{I \subseteq \{1, \cdots, m\}} U_I \searrow z_I$ in $B_1$
  satisfying conditions $(a)$ and $(b)$, lying in $W$, such that
  $z_I \leq y_I$ for every $I \subseteq \{1, \cdots, m\}$ (in
  particular, $h \leq g$), and such that $z_I \in B_Y$ and
  $z_I \ll y_I$ for every $I \in \mathcal I$.  This is proved by
  induction on the cardinality of $\mathcal I$.  This is vacuous if
  $\mathcal I$ is empty.  Hence consider a non-empty downwards-closed
  family $\mathcal I$ of $\pow (\{1, \cdots, m\})$, let $I_0$ be a
  maximal element of $\mathcal I$, and let
  $\mathcal I' = \mathcal I \diff \{I_0\}$.  Notice that $\mathcal I'$
  is again downwards-closed.  Hence, by induction hypothesis, there an
  element $h = \sup_{I \subseteq \{1, \cdots, m\}} U_I \searrow z_I$
  of $B_1$ satisfying conditions $(a)$ and $(b)$, in $W$, such that
  $z_I \leq y_I$ for every $I \subseteq \{1, \cdots, m\}$, and such
  that $z_I \in B_Y$ and $z_I \ll y_I$ for every $I \in \mathcal I'$.
  Since $B_Y$ is a basis, we can write $y_{I_0}$ as the supremum of a
  directed family ${(y'_j)}_{j \in J}$ of elements of $B_I$.  For each
  $j \in J$, let
  $h[y'_j] = \sup (\sup_{I \subseteq \{1, \cdots, m\}, I \neq I_0} U_I
  \searrow z_I, (U_{I_0} \searrow y'_j))$.  One checks easily that
  ${(h [y'_j])}_{j \in J}$ is a directed family whose supremum is
  above $h$.  Hence $h [y'_j]$ is in $W$ for some $j \in J$.  For
  every $I \subsetneq I_0$, $z_I \ll y_I \leq y_{I_0}$ (using $(a)$),
  so there is a $j_I \in J$ such that $z_I \leq y'_{j_I}$.  By
  directedness, we can assume without loss of generality that $j$ and
  all the indices $j_I$, $I \subsetneq I_0$, are equal.  (Otherwise
  replace them by some $k \in J$ such that $y'_j \leq y'_k$ and
  $y'_{j_I} \leq y'_k$ for every $I \subsetneq I_0$.)  For every
  $I \subseteq \{1, \cdots, m\}$, we define $z'_I$ as $z_I$ if $I$
  does not contain $I_0$, as $y'_j$ if $I=I_0$, and as
  $\sup (z_I, y'_j)$ if $I$ contains $I_0$ strictly.  We have:
  \begin{itemize}
  \item For all $I \subseteq J \subseteq \{1, \cdots, m\}$,
    $z'_I \leq z'_J$.  The key case is when $J=I_0$, which follows
    from the fact that $y'_j = z'_{I_0}$ was chosen larger than or
    equal to every $z_I$, $I \subsetneq I_0$, and that $z'_I = z_I$ in
    this case.  When $I=I_0$ instead, $z'_I = y'_j \leq y_{I_0} \leq
    y_J$ (by $(a)$).  The cases where $I$, $J$ are both different from
    $I_0$ are easy verifications.
  \item Hence the function
    $h' = \sup_{I \subseteq \{1, \cdots, m\}} U_I \searrow z'_I$
    satisfies $(a)$, and trivially $(b)$ as well.
  \item For every $I \subseteq \{1, \cdots, m\}$, $z'_I \leq y_I$.
    When $I=I_0$, this is because $z'_I = y'_j \leq y_{I_0}$.  When $I
    \supsetneq I_0$, $z'_I = \sup (z_I, y'_j) \leq \sup (y_I, y_{I_0})
    = y_I$, using $(a)$.  
  \item We have built $h'$ so that it is in $W$.
  \item For every $I \in \mathcal I$, $z'_I$ is in $B_Y$: when $I=I_0$,
    this is because $z'_{I_0} = y'_j$ is in $B_Y$; otherwise, since
    $I_0$ is maximal in $\mathcal I$, $I$ cannot contain $I_0$, so
    $z'_I=z_I$, which is in $B_Y$ because $I \in \mathcal I'$, using
    the induction hypothesis.
  \item For every $I \in \mathcal I$, $z'_I \ll y_I$: when $I=I_0$,
    this is because $z'_I = y'_j \ll y_{I_0}$; otherwise, $z'_i=z_I
    \ll y_I$ because $I \in \mathcal I'$, using the induction
    hypothesis.
  \end{itemize}
  This finishes to prove claim $(*)$.  Applying this claim to the case
  where $\mathcal I$ is the whole of $\pow (\{1, \cdots, m\})$, we
  obtain an element
  $h = \sup_{I \subseteq \{1, \cdots, m\}} U_I \searrow z_I$ of $B_1$
  satisfying conditions $(a)$ and $(b)$, in $W$, below $g$ hence way-below
  $f$, and such that $z_I \in B_Y$ and $z_I \ll y_I$ for every
  $I \in \{1, \cdots, m\}$.  In particular, that element $h$ is in
  $B_2$.  By our preliminary remark, $B_2$ is a basis of $[X \to Y]$.
  \qed
\end{proof}

\begin{proposition}
  \label{prop:subbase:fun}
  Let $X$ be a continuous dcpo and $Y$ be a bc-domain.  Let $B_X$ be a
  basis of $X$, $\mathcal B_X$ be a base of the Scott topology on $X$.
  Let $B_Y$ be a basis of $Y$, and $\mathcal S_Y$ be a subbase of the
  Scott topology on $Y$.  Then:
  \begin{itemize}
  \item The set $B_{[X \to Y]}$ of all $(\mathcal B_X, B_Y)$-step
    functions is a basis of $[X \to Y]$.
  \item The set $\mathcal S_{[X \to Y]}$ of all opens $[x \mapsto V]$,
    $x \in B_X$, $V \in \mathcal S_Y$, is a subbase of the Scott
    topology on $[X \to Y]$.  We write $[x \mapsto V]$ for the open
    subset $\{f \in [X \to Y] \mid f (x) \in V\}$.
  \end{itemize}
\end{proposition}
\begin{proof}
  The first part is Lemma~\ref{lemma:step:fun}.  The second part is
  based on Lemma~5.16 of \cite{JGL-mscs09}, which states that the
  subsets $[x \mapsto V]$, $x \in X$, $V$ open in $Y$, form a subbase
  of the topology of $[X \to Y]$, as soon as $X$ is a continuous poset
  and $Y$ is a bc-domain.  \qed
\end{proof}

We introduce the following abbreviations.
\begin{itemize}
\item For all $M \colon \intT$, $N, P \colon \underline\tau$, and for
  every $n \in \nat$, $\pif {(M \eq \underline n)} N P$ denotes
  $\pifz {\underbrace{\pred (\pred \cdots (\pred}_{n\text{ times}}
    M))} N P$.
\item Given terms $M \colon \intT$ and $N_1$, \ldots, $N_n$ of type
  $\underline\tau$,
  $\pswitchkw\; M \colon \underline 1 \mapsto N_1 \mid \cdots \mid
  \underline n \mapsto N_n$ abbreviates:
  \[
    \begin{array}{l}
      \pif {(M \eq \underline n)} {N_n} {(} \\
      \qquad \pif {(M \eq \underline {n-1})} {N_{n-1}} {(} \\
      \qquad\qquad \cdots \\
      \qquad \qquad \pif {(M \eq \underline 2)} {N_2} {(} \\
      \qquad \qquad \qquad \pif {(M \eq \underline 1)} {N_1} \\
      \qquad \qquad \qquad \qquad
      {\abort_{\underline\tau}}
      \quad))).
    \end{array}
  \]
  In particular, if $n=0$, this is equal to $\abort_{\underline\tau}$.
\item Given terms $M \colon \unitT$ and $N \colon \unitT$,
  $M \vee N \colon \unitT$ is the term defined as
  $\bigcirc_{> 1/2} (\pifz {(M; \underline 0)} {(\produce \retkw
    \underline *)} {(\produce \retkw N)})$.
\item Given a term $M \colon \unitT$, $n \in \nat$ and $i \in \nat$
  such that $1 \leq i \leq n$, $\tagc M i n$ is the term of type
  $\F{} \intT$ defined as
  $\pifz {(M; \underline 0)} {\abort_{\F{} \intT}} {(\produce i)}$.
\item Given terms $M_1$, \ldots, $M_n$ of type $\unitT$ and
  $N_1$, \ldots, $N_n$ of type $\underline\tau$,
  $\pcasekw\; M_1 \mapsto N_1 \mid \cdots \mid M_n \mapsto N_n$
  abbreviates:
  \[
    \begin{array}{l}
      (\tagc {M_1} 1 n \owedge \cdots \owedge \tagc {M_n} n n)) \\
      \qquad \pto \relax {y_{\intT}} {\pswitchkw\; y_{\intT} \colon \underline 1 \mapsto N_1 \mid \cdots \mid
  \underline n \mapsto N_n}.
    \end{array}
  \]
\end{itemize}

\begin{lemma}
  \label{lemma:pswitch}
  \begin{enumerate}
  \item $\Eval {\pif {(M \eq \underline n)} N P} \rho$ is equal to $\Eval
    N \rho$ if $\Eval M \rho = n$, to $\Eval P \rho$ if $\Eval M \rho
    \neq n, \bot$ and to $\Eval N \rho \wedge \Eval P \rho$ if $\Eval
    M \rho = \bot$;
  \item
    $\Eval {\pswitchkw\; M \colon \underline 1 \mapsto N_1 \mid \cdots
      \mid \underline n \mapsto N_n} \rho$ is equal to
    $\Eval {N_m} \rho$ if $\Eval M \rho$ is an element $m$ of
    $\{1, \cdots, n\}$, to
    $\bigwedge_{i \in \{1, \cdots, n\}} \Eval {N_i} \rho$ if
    $\Eval M \rho = \bot$, and to $\top$ otherwise.
  \item $\Eval {M \vee N} \rho = \sup (\Eval M \rho, \Eval N \rho)$.
  \item $\Eval {\tagc M i n} \rho$ is equal to $\emptyset$ if
    $\Eval M \rho=\top$, to $\{i\}$ if $\Eval M \rho = \bot$.
  \end{enumerate}
\end{lemma}
\begin{proof}
  1, 3 and 4 are clear.  We prove item~2 by induction on $n$.  If
  $n=0$, then
  $\Eval {\pswitchkw\; M \colon \underline 1 \mapsto N_1 \mid \cdots
    \mid \underline n \mapsto N_n} \rho = \Eval
  {\abort_{\underline\tau}} \rho$, which is the top element $\top$ of
  $\Eval {\underline\tau}$, as an easy induction on $\underline\tau$
  shows.  Note that, in case $\Eval M \rho=\bot$, this is also equal
  to $\bigwedge_{i \in \{1, \cdots, n\}} \Eval {N_i} \rho$ since
  $n=0$.  If $n \geq 1$, then by item~1,
  $\Eval {\pswitchkw\; M \colon \underline 1 \mapsto N_1 \mid \cdots
    \mid \underline n \mapsto N_n} \rho$ is equal to
  $\Eval {N_n} \rho$ if $\Eval M \rho=n$, to
  $\Eval {\pswitchkw\; M \colon \underline 1 \mapsto N_1 \mid \cdots
    \mid \underline {n-1} \mapsto N_{n-1}} \rho$ if
  $\Eval M \rho \in \Z \diff \{n\}$, and to their infimum if
  $\Eval M \rho=\bot$.  We then use the induction hypothesis to
  conclude.  \qed
\end{proof}

\begin{lemma}
  \label{lemma:pto:fininf}
  Let $M \colon \F{} \sigma$ and $N \colon \underline\tau$.  Assume that
  $\Eval M \rho$ is of the form
  $\upc \{V_1, \cdots, \allowbreak V_k\}$.  Then
  $\Eval {\pto M {x_\sigma} N} \rho = \bigwedge_{i=1}^k \Eval N \rho
  [x_\sigma \mapsto V_i]$.
\end{lemma}
\begin{proof}
  By structural induction on $\underline\tau$.  Let $f$ be the map
  $V \in \Eval \sigma \mapsto \Eval N \rho [x_\sigma \mapsto V]$.  If
  $\underline\tau$ is of the form $\F{} \tau$, then:
  \begin{align*}
    \Eval {\pto M {x_\sigma} N} \rho
    & = \extd f (\upc \{V_1, \cdots, V_k\}) \\
    & = \bigwedge_{i=1}^k \extd f (\eta^\Smyth (V_i)) & \text{by
                                                        Proposition~\ref{prop:Qtop},
                                                        item~3} \\
    & = \bigwedge_{i=1}^k f (V_i),
  \end{align*}
  by Proposition~\ref{prop:Qtop}, item~2.

  If $\underline\tau$ is of the form $\lambda \to \underline\tau'$,
  then:
  \begin{align*}
    \Eval {\pto M {x_\sigma} N} \rho
    & = \Eval {\lambda y_\lambda . \pto M {x_\sigma} {N
      y_\lambda}} \rho \\
    & = (V \in \Eval \lambda \mapsto \bigwedge_{i=1}^k f (V_i) (V))
    & \text{by induction hypothesis} \\
    & = \bigwedge_{i=1}^k f (V_i).
  \end{align*}
  The last equality follows from the fact that finite infima of
  continuous functions are computed pointwise, by
  Lemma~\ref{lemma:sound:aux}, item~2.  \qed
\end{proof}

\begin{lemma}
  \label{lemma:pcase}
  $\Eval {\pcasekw\; M_1 \mapsto N_1 \mid \cdots \mid M_n \mapsto N_n}
  \rho$ is equal to $\bigwedge_{i \in I} \Eval {N_i} \rho$, where
  $I = \{i \in \{1, \cdots, n\} \mid \Eval {M_i} \rho \neq \top\}$.
\end{lemma}
\begin{proof}
  By Lemma~\ref{lemma:pswitch}, item~4,
  $\Eval {\tagc {M_1} 1 n \owedge \cdots \owedge \tagc {M_n} n n} \rho
  = \bigwedge_{i=1}^n \Eval {\tagc {M_i} i n} \rho = \bigcup_{i=1}^n
  \Eval {\tagc {M_i} i n} \rho = \bigcup_{i \in I} \{i\} = I = \upc
  I$.  By Lemma~\ref{lemma:pto:fininf}, it then follows that
  $\Eval {\pcasekw\; M_1 \mapsto N_1 \mid \cdots \mid M_n \mapsto N_n}
  \rho$ is equal to
  $\bigwedge_{i \in I} \Eval P \rho [y_{\intT} \mapsto i]$ where
  $P = \pswitchkw\; y_{\intT} \colon \underline 1 \mapsto N_1 \mid
  \cdots \mid \underline n \mapsto N_n$, and by
  Lemma~\ref{lemma:pswitch}, item~2, this is equal to
  $\bigwedge_{i \in I} \Eval {N_i} \rho$.  \qed
\end{proof}

It follows:
\begin{proposition}
  \label{prop:descr:to}
  For every describable value type $\sigma$, for every describable
  computation type $\underline\tau$, the type $\sigma \to
  \underline\tau$ is describable.
\end{proposition}
\begin{proof}
  We use Proposition~\ref{prop:subbase:fun}, with $X = \Eval \sigma$,
  $Y = \Eval {\underline\tau}$.  $B_X$ (resp., $B_Y$) is the basis of
  definable elements of $\Eval \sigma$ (resp.,
  $\Eval {\underline\tau}$).  $\mathcal B_X$ is the base of definable
  open subsets at type $\sigma$, obtained thanks to
  Lemma~\ref{lemma:descr:base}, and $\mathcal S_Y$ is the subbase of
  definable open subsets at type $\underline\tau$.

  We first show that all the elements of $B_{[X \to Y]}$ are
  definable.  This will imply that the definable elements at type
  $\sigma \to \underline\tau$ form a basis of
  $\Eval {\sigma \to \underline\tau}$.  We recall that such an element
  is a $(\mathcal B_X, B_Y)$-step function
  $f = \sup_{I \subseteq \{1, \cdots, m\}} U_I \searrow y_I$.

  Let $U_I$ be defined by ground terms
  $M_I \colon \sigma \to \F\Vt\unitT$, i.e.,
  $\tilde\chi_{U_I} = \Eval {M_I}$, and let $y_I$ be defined by ground
  terms $N_I \colon \underline\tau$.  Let us pick a variable
  $x_\sigma$.  For every subset $I$ of $\{1, \cdots, m\}$, let
  $M_I^\perp (x_\sigma) = \bigvee_{J \subseteq \{1, \cdots, m\}, J
    \not\subseteq I} \bigcirc_{> 1/2} (M_J x_\sigma)$.  (If
  $I = \{1, \cdots, m\}$, the empty disjunction is $\Omega_{\unitT}$.)
  For every environment $\rho$, and letting $a = \rho (x_\sigma)$, by
  Lemma~\ref{lemma:pcase},
  $\Eval {\pcasekw\; \{M_I^\perp (x_\sigma) \mapsto N_I \mid I
    \subseteq \{1, \cdots, m\}\}} \rho$ is equal to the infimum of the
  values $\Eval {N_I} = y_I$ over the subsets $I$ of
  $\{1, \cdots, m\}$ such that
  $\Eval {M_I^\perp (x_\sigma)} \rho \neq \top$, i.e., such that
  $a \not\in \bigcup_{J \not\subseteq I} U_J$.

  Let $I_0$ be the set of indices $i$ between $1$ and $m$ such that
  $a \in U_{\{i\}}$.  For every $J \subseteq \{1, \cdots, m\}$,
  $a \in U_J$ if and only if for every $i \in J$, $a$ is in
  $U_{\{i\}}$, if and only if $J \subseteq I_0$.  For every
  $I \subseteq \{1, \cdots, m\}$,
  $a \not\in \bigcup_{J \not\subseteq I} U_J$ if and only if for every
  $J \not\subseteq I$, $a \not\in U_J$, if and only if for every
  $J \subseteq \{1, \cdots, m\}$, $a \in U_J$ implies $J \subseteq I$
  (by contraposition), if and only if for every
  $J \subseteq \{1, \cdots, m\}$, $J \subseteq I_0$ implies
  $J \subseteq I$, if and only if $I_0 \subseteq I$.  Therefore
  $\Eval {\pcasekw\; \{M_I^\perp (x_\sigma) \mapsto N_I \mid I
    \subseteq \{1, \cdots, m\}\}} \rho$ is equal to
  $\bigwedge_{I \supseteq I_0} y_I = y_{I_0} = f (a)$.  It follows
  that $f$ is definable as
  $\lambda x_\sigma . \pcasekw\; \{M_I^\perp (x_\sigma) \mapsto N_I
  \mid I \subseteq \{1, \cdots, m\}\}$.
  


  

  Second, we show that all the elements of $\mathcal S_{[X \to Y]}$
  are definable as ground terms of type
  $\U{} (\sigma \to \underline\tau) \to \F\Vt\unitT$.  Such an element is of
  the form $[x \mapsto V]$, where $x = \Eval M$ for some ground term
  $M \colon \sigma$, and $V = \Eval P$ for some ground term
  $P \colon \U\underline\tau \to \F\Vt\unitT$.  Then $[x \mapsto V]$ is
  definable as the ground term
  $\lambda f_{\U{} (\sigma \to \underline\tau)} . P (\thunk (\force f_{\U{}
    (\sigma \to \underline\tau)} M))$.  \qed
\end{proof}

\subsection{Valuation Types}
\label{sec:probability-types}

We have already mentioned in Section~\ref{sec:language} that, for
every continuous dcpo, $\Val_{\leq 1} X$ is a pointed continuous dcpo,
and that its Scott topology coincides with the weak upwards topology
\cite{AMJK:scs:prob}.  The latter has a subbase of open sets of the
form $[U > r]$, for every open subset $U$ of $X$ and
$r \in \Rp \diff \{0\}$, where
$[U > r] = \{\nu \in \Val_{\leq 1} X \mid \nu (U) > r\}$.  We can
restrict $r$ further so that $r < 1$, since otherwise $[U > r]$ is
empty.  Call a number \emph{dyadic} if and only if it is of the form
$a/2^k$, with $a, k \in \nat$.
\begin{proposition}
  \label{prop:V1:basis}
  Let $X$ be a pointed continuous dcpo.  Let $B_X$ be a basis of $X$,
  $\mathcal B_X$ be a base of the Scott topology on $X$.  Then:
  \begin{itemize}
  \item The set $B_{\Val_{\leq 1} X}$ of all simple probability
    valuations $\sum_{i=1}^n a_i \delta_{x_i}$, where each $a_i$ is a
    dyadic number in $[0, 1]$, $\sum_{i=1}^n a_i \leq 1$, and each
    $x_i$ is a point in $B_X$, is a basis of $\Val_{\leq 1} X$.
  \item The set $\mathcal S_{\Val_{\leq 1} X}$ of all opens $[U > r]$,
    where $U$ is an element of $\mathcal B_X$, and $r$ is a dyadic
    number in $(0, 1)$, is a subbase of the Scott topology on
    $\Val_{\leq 1} X$.
  \end{itemize}
\end{proposition}
\begin{proof}
  By a theorem of 
  Jones \cite[Theorem~5.2]{Jones:proba}, the simple subprobability
  valuations form a basis of $\Val_{\leq 1} X$.  For every simple
  subprobability valuation
  $\nu = \sum_{i=1}^n a_i \delta_{x_i}$,
  one easily checks that the collection
  $D_\nu$ of simple subprobability valuations $\sum_{i=1}^n b_i
  \delta_{y_i}$
  with $b_i$ dyadic and way-below $a_i$ in $[0, 1]$, and $y_i \in B_X$
  way-below $x_i$, is directed, and $\sup D_\nu = \nu$.

  We check that every element of $D_\nu$, as written above, is
  way-below $\nu$.
  For convenience, we let
  $\mu = \sum_{i=1}^n b_i \delta_{y_i}$.  Let ${(\nu_k)}_{k \in K}$ be
  a directed family in $\Val_{\leq 1} X$ with a supremum above $\nu$.
  We wish to show that there is a $k \in K$ such that for every open
  subset $U$ of $X$, $\mu (U) \leq \nu_k (U)$.  In order to do so, we
  show that for every subset $J$ of $\{1, \cdots, n\}$, there is an
  index $k = k_J \in K$ such that for every open subset $U$ of $X$
  such that $J = \{i \in \{1, \cdots, n\} \mid y_i \in U\}$,
  $\mu (U) \leq \nu_k (U)$.  By directedness, there is a $k \in K$
  such that $\nu_{k_J} \leq \nu_k$ for every such $J$, and this will
  show the claim.

  Henceforth, let us fix $J \subseteq \{1, \cdots, n\}$.
  We have $\sum_{i \in J} b_i \ll \sum_{i \in J} a_i$ (because
  $b_i \ll a_i$ for each $i$, and recalling that $b_i \ll a_i$ iff
  $b_i=0$ or $b_i < a_i$) $\leq \nu (\bigcup_{i \in J} \uuarrow y_i)$
  (because $y_i \ll x_i$ for each $i$), so there is a $k \in K$ such
  that
  $\sum_{i \in J} b_i \leq \nu_k (\bigcup_{i \in J} \uuarrow y_i)$.
  For every open subset $U$ with
  $J = \{i \in \{1, \cdots, n\} \mid y_i \in U\}$,
  $\mu (U) = \sum_{i \in J} b_i \leq \nu_k (\bigcup_{i \in J} \uuarrow
  y_i) \leq \nu_k (U)$, which finishes the proof.

  It is standard domain theory that given a dcpo $Z$, a point $z \in
  Z$ that is the supremum of a directed family ${(z_i)}_{i \in I}$,
  where $z_i$ is itself the supremum of a directed family $D_i$ of
  points way-below $z_i$, then $\bigcup_{i \in I} D_i$ is directed and
  has $z$ as supremum.  In our case $D_\nu$ is included in $B_{\Val_{\leq 1}
    X}$, showing that every continuous probability valuation is the
  supremum of a directed family of elements of $B_{\Val_{\leq 1} X}$.

  In order to show the second part of the proposition, we consider an
  arbitrary subbasic open set $[U > r]$ of the weak upwards (=Scott)
  topology, $U$ open in $X$, $r \in (0, 1)$.  We write $U$ as
  $\bigcup_{i \in I} U_i$, where each $U_i$ is in $\mathcal B_X$, and
  $r$ as the infimum of the numbers $r_n = \lceil 2^n r \rceil / 2^n$.
  Since $0 < r < 1$, $r_n$ is in $(0, 1)$ for $n$ large enough.  For
  every $\nu \in \Val_{\leq 1} X$, $\nu (U) > r$ if and only if for some $n$
  large enough $\nu (U) > r_n$, if and only if for some $n$ large
  enough and some finite subset $A$ of $I$,
  $\nu (\bigcup_{i \in I} U_i) > r_n$.  Hence
  $[U > r] = \bigcup_{A \text{ finite }\subseteq I, n / r_n < 1}
  [\bigcup_{i \in A} U_i> r_n]$, showing that $\mathcal S_{\Val_{\leq 1} X}$
  is a subbase of the weak upwards (=Scott) topology.  \qed
\end{proof}

\begin{corollary}
  \label{corl:descr:V}
  For every describable value type $\sigma$, the type $\Vt\sigma$ is
  describable.
\end{corollary}
\begin{proof}
  Let $X = \Eval \sigma$, $B_X$ be the basis of definable elements of
  $\Eval \sigma$, $\mathcal B_X$ be a strong base of definable
  open subsets at type $\sigma$ guaranteed by
  Lemma~\ref{lemma:descr:base}, and let us use
  Proposition~\ref{prop:V1:basis}.

  Although $\oplus$ is not associative, we can make sense of sums
  $M_1 \oplus M_2 \oplus \cdots \oplus M_{2^n}$ of $2^n$ terms of type
  $\Vt\sigma$: when $n=0$, this is just $M_1$, otherwise this is
  $(M_1 \oplus \cdots \oplus M_{2^{n-1}}) \oplus (M_{2^{n-1}+1} \oplus
  \cdots \oplus M_{2^n})$.  This way,
  $\Eval {M_1 \oplus M_2 \oplus \cdots \oplus M_{2^n}} \rho$ is simply
  equal to $\frac 1 {2^n} \sum_{i=1}^{2^n} \Eval {M_i} \rho$.
  
  For every element
  $\nu = \sum_{i=1}^n a_i \delta_{x_i}$
  in $B_{\Val_{\leq 1} X}$,
  we can write each $a_i$ ($1\leq i \leq n$) as $k_i /
  2^m$, where $k_i \in \nat$ and with the same
  $m$ for all values of $i$.
  Hence, and letting $k_0 = 2^m - \sum_{i=1}^n k_i$,
  $\nu$ can be written as a sum $\frac 1 {2^m} \sum_{i=1}^{2^m - k_0}
  \delta_{v_i} + \frac 1 {2^m} \sum_{i=2^m-k_0+1}^{2^m}
  0$, where each $v_i$ is in $B_X$.  Since
  $\sigma$ is describable, for each $i$ ($1\leq i\leq 2^m-k_0$)
  $v_i$ is equal to $\Eval {M_i}$ for some ground term $M_i \colon
  \sigma$ ($\bot$ is equal to $\Eval {\Omega_\sigma}$).  Also,
  $0$ is equal to $\Eval {\Omega_{\Vt\sigma}}$, so
  $\nu$ is definable as the sum of the $2^m - k_0$ terms $\retkw
  M_i$, plus $k_0$ terms $\Omega_{\Vt\sigma}$.

  Let $[U > r]$ be an element of $\mathcal S_{\Val_{\leq 1}
    X}$, where $U = \bigcup_{i=1}^m U_i$, $U_i \in \mathcal
  B_X$, and $r$ is a dyadic number in $(0, 1)$.  Each
  $U_i$ is definable, that is, $\tilde\chi_{U_i} = \Eval
  {M_i}$ for some ground term $M_i \colon \sigma \to
  \F\Vt\unitT$.  Let us fix a variable $x_\sigma$.  For each $i$, let $M'
  (x_\sigma) = \bigcirc_{> 1/2} (M_1 x_\sigma) \vee \cdots \vee
  \bigcirc_{> 1/2} (M_m x_\sigma)$: $\Eval {M' (x_\sigma)} \rho =
  \top$ if $\rho (x_\sigma) \in U$, $\bot$ otherwise.  Then $[U >
  r]$ is definable by the term $\lambda y_{\Vt\sigma} . \bigcirc_{> r}
  (\produce (\dokw {x_\sigma \leftarrow y_{\Vt\sigma}}; {\retkw M'
    (x_\sigma)}))$.  Indeed, letting $\nu = \rho (y_{\Vt\sigma})$,
  \begin{align*}
    \Eval {\dokw {x_\sigma \leftarrow y_{\Vt\sigma}}; {\retkw M' (x_\sigma)}} \rho (\{\top\})
    & = {(a \in \Eval \sigma \mapsto \Eval {\retkw M' (x_\sigma)} \rho [x_\sigma \mapsto
      a])}^\dagger (\nu) (\{\top\}) \\
    & = \int_{a \in \Eval \sigma} \Eval {\retkw M' (x_\sigma)} \rho [x_\sigma \mapsto
      a] (\{\top\}) d\nu \\
    & = \int_{a \in \Eval \sigma} \delta_{\Eval {M' (x_\sigma)} \rho [x_\sigma \mapsto
      a]} (\{\top\}) d\nu \\
    & = \int_{a \in \Eval \sigma} \chi_U (a) d\nu =
      \nu (U),
  \end{align*}
  so
  $\Eval {\bigcirc_{> r} (\produce (\dokw {x_\sigma \leftarrow
      y_{\Vt\sigma}}; {\retkw M'}))} \rho$ is equal to $\top$ if
  $r \ll \nu (U)$, $\bot$ otherwise.  \qed
\end{proof}

\subsection{$\F{}$ Types}
\label{sec:f-types}

The \emph{upper Vietoris topology} on $\Smyth^\top (X)$ (resp.,
$\Smyth^\top_\bot (X)$) has basic open sets
$\Box U = \{Q \in \Smyth^\top (X) \mid Q \subseteq U\}$, where $U$
ranges over the open subsets of $X$.  The operator $\Box$ commutes
with finite intersections and with directed suprema.  Moreover,
$\Box U$ is Scott-open if $X$ is well-filtered.
\begin{proposition}
  \label{prop:Q:basis}
  Let $X$ be a pointed, coherent, continuous dcpo.  Let $B_X$ be a
  basis of $X$, $\mathcal S_X$ be a subbase of the Scott topology on
  $X$.  Then:
  \begin{itemize}
  \item The set $B_{\Smyth^\top_\bot X}$ consisting of $\bot$ plus the
    compact saturated sets of the form $\upc \{x_1, \cdots, x_n\}$, $n
    \in \nat$, where each $x_i$ is in $B_X$, is a basis of
    $\Smyth^\top_\bot (X)$.
  \item The set $\mathcal S_{\Smyth^\top_\bot X}$ of all opens
    $\Box U$, where $U$ ranges over non-empty finite unions of
    elements of $\mathcal S_X$, plus the whole space
    $\Smyth^\top_\bot X$ itself, is a base of the Scott topology on
    $\Smyth^\top_\bot (X)$.
  \end{itemize}
\end{proposition}
\begin{proof}
  By Proposition~\ref{prop:Qtop}, item~1, $Q \ll Q'$ if and only if
  $Q = \bot$ or $Q' \subseteq \interior Q$.  Now $\interior Q$ can be
  written as $\bigcup_{x \in Q \cap B_X} \uuarrow x$, and since $Q'$
  is compact, if $Q' \subseteq \interior Q$ then there are finitely
  many elements $x_1$, \ldots, $x_n$ of $Q \cap B_X$ such that
  $Q' \subseteq \bigcup_{i=1}^n \uuarrow x_i = \interior {\upc \{x_1,
    \cdots, x_n\}}$.  By Lemma~\ref{lemma:rem:basis}, this shows the
  first part.

  Let $\mathcal U$ be a Scott-open subset of $\Smyth^\top_\bot (X)$.
  If $\bot \in \mathcal U$, then $\mathcal U$ is the whole space,
  which is in $\mathcal S_{\Smyth^\top_\bot X}$.  Otherwise,
  $\mathcal U$ is a Scott-open subset of $\Smyth^\top (X)$.  By
  Proposition~\ref{prop:Qtop:1}, item~1, $\Smyth^\top (X)$ is a
  continuous complete lattice, so $\mathcal U$ is a union of sets of
  the form $\uuarrow Q$, where $Q$ ranges over the elements of
  $\mathcal U$ belonging to any given basis, and $\uuarrow$ is
  understood in $\Smyth^\top (X)$.  Using the first part, we can take
  those elements $Q$ of the form $\upc \{x_1, \cdots, x_n\}$, and then
  $\uuarrow Q = \{Q \in \Smyth^\top (X) \mid Q' \subseteq \interior
  {\upc \{x_1, \cdots, x_n\}}\} = \Box \interior {\upc \{x_1, \cdots,
    x_n\}}$.

  We can therefore write $\mathcal U$ as a union of sets $\Box U$, $U$
  open in $X$, and then we can write $U$ as a union of finite
  intersections (taken in $\Smyth^\top (X)$) of elements of
  $\mathcal S_X$, hence as a directed union of finite unions of finite
  intersections of elements of $\mathcal S_X$, hence (by
  distributivity) as a directed union of finite intersections of
  finite unions of elements of $\mathcal S_X$.  In $\Smyth^\top (X)$
  (not $\Smyth^\top_\bot (X)$), $\Box$ commutes with directed unions
  and finite intersections (this would not hold for the empty
  intersection in $\Smyth^\top_\bot (X)$).  The result follows.  \qed
\end{proof}

\begin{corollary}
  \label{corl:descr:F}
  For every describable value type $\sigma$, $\F{}\sigma$ is a
  describable computation type.
\end{corollary}
\begin{proof}
  Let $X = \Eval \sigma$, $B_X$ be the basis of definable elements of
  $\Eval \sigma$, and $\mathcal S_X$ be the subbase of definable open
  subsets at type $\sigma$, and let us use
  Proposition~\ref{prop:Q:basis}.

  For every element $Q = \upc \{x_1, \cdots, x_n\}$ of
  $B_{\Smyth^\top_\bot X}$, where each $x_i$ is in $B_X$, hence
  $x_i = \Eval {M_i}$ for some ground term $M_i \colon \sigma$, the
  term $M = \produce M_1 \owedge \cdots \owedge \produce M_n$
  ($\abort_{\F{}\sigma}$ if $n=0$) defines $Q$, in the sense that
  $Q = \Eval M$.  The term $\Omega_{\F{}\sigma}$ defines $\bot$.

  We deal with the second part.  The whole space $\Eval {\F{}\sigma}$ is
  definable as an open set by the term
  $\lambda x_{\F{} \sigma} . \produce \retkw \underline*$.  We consider
  the other elements $\Box U$ of $\mathcal S_{\Smyth^\top_\bot X}$.
  Let us write $U$ as a finite union $U = \bigcup_{i=1}^m U_i$ of
  elements of $\mathcal S_X$, where $U_i$ is defined by
  $M_i \colon \sigma \to \F\Vt\unitT$ in the sense that
  $\tilde\chi_{U_i} = \Eval {M_i}$.  Then $\Box U$ is defined by
  $\lambda x_{\F{}\sigma} . \pto {x_{\F{}\sigma}} {y_\sigma} {M
    (y_\sigma)}$, where
  $M (y_\sigma) = (\bigcirc_{> 1/2} (M_1 y_\sigma) \vee \cdots \vee
  \bigcirc_{> 1/2} (M_m y_\sigma)); \produce \retkw \underline*$.
  Indeed, for every environment $\rho$, letting
  $Q = \rho (x_{\F{}\sigma})$, if $Q = \bot$ then
  $\Eval {\lambda x_{\F{}\sigma} . \pto {x_{\F{}\sigma}} {y_\sigma} {M
      (y_\sigma)}} (\bot)=\bot$, matching the fact that $Q$ is not in
  $\Box U$.  Otherwise, using the fact that
  $\Eval {\bigcirc_{> 1/2} (M_1 y_\sigma) \vee \cdots \vee \bigcirc_{>
      1/2} (M_m y_\sigma)} \rho'$ is equal to $\top$ if
  $\rho' (y_\sigma) \in U$ and to $\bot$ otherwise, for every
  environment $\rho'$, we obtain:
  \begin{align*}
    \Eval {\pto {x_{\F{}\sigma}} {y_\sigma} {M (y_\sigma)}} \rho
    & = \extd {(\tilde\chi_U)} (Q)
    = \bigwedge_{a \in Q} \tilde\chi_U (a),
  \end{align*}
  by Proposition~\ref{prop:Qtop}, item~2, and that is equal to
  $\{\delta_\top\}$ if $Q \subseteq U$, and to $\bot$ otherwise.  \qed
\end{proof}

\subsection{Full Abstraction}
\label{sec:full-abstraction-1}

By induction on types, using Lemma~\ref{lemma:descr:unit} ($\unitT$),
Lemma~\ref{lemma:desc:int} ($\intT$),
Proposition~\ref{prop:descr:prod} (product types),
Proposition~\ref{prop:descr:to} (function types),
Corollary~\ref{corl:descr:V} ($\Vt{}$ types), and
Corollary~\ref{corl:descr:F} ($\F{}$ types), every type is describable
(the case of $\U{}$ types is trivial).

\begin{theorem}[Full abstraction]
  \label{thm:fa}
  CBPV$(\Demon,\Nature)+\pifzkw+\bigcirc$ is inequationally fully
  abstract.  For every value type $\tau$, for every two ground
  CBPV$(\Demon,\Nature)+\pifzkw+\bigcirc$ terms $M, N \colon \tau$,
  the following are equivalent:
  \begin{enumerate}
  \item $M \precsim^{app}_\tau N$;
  \item $M \precsim_\tau N$;
  \item $\Eval M \leq \Eval N$.
  \end{enumerate}
\end{theorem}
\begin{proof}
  The equivalence between 1 and 2 is Theorem~\ref{thm:context=app}.
  Item~3 implies item~1 by Proposition~\ref{prop:fa:easy}.  In the
  converse direction, we assume that $\Eval M \not\leq \Eval N$ and we
  claim that there is a ground term $Q \colon \tau \to \F\Vt\unitT$ such
  that $\Prob (QM \vp) \not\leq \Prob (QN \vp)$.  Since $\leq$ is the
  specialization ordering of the Scott topology on $\Eval \tau$, and
  the latter has a subbase of definable elements, there is a ground
  term $Q \colon \tau \to \F\Vt\unitT$ such that $\Eval M \in U$ and
  $\Eval N \not\in U$, where $\tilde\chi_U = \Eval Q$.  Hence
  $\Eval {QM} = \tilde\chi_U (\Eval M) = \{\delta_\top\}$, while
  $\Eval {QN} = \bot$.  By adequacy (Proposition~\ref{prop:adeq}), and
  letting $h (\nu) = \nu (\{\top\})$,
  $\Prob (QM \vp) = \extd h (\Eval {QM}) = 1$, while
  $\Prob (QN \vp) = 0$.  \qed
\end{proof}

\section{Conclusion and Open Problems}
\label{sec:concl-open-probl}

We started from the question of using call-by-push-value as a way of
getting around our ignorance of the existence of a Cartesian-closed
category of continuous dcpos that would be closed under the
probabilistic powerdomain functor.  This led us to define a pretty
expressive call-by-push-value language with probabilistic choice and
demonic non-determinism.  We have gone so far as to show that it is
inequationally fully abstract, once extended with parallel if
$\pifzkw$ and statistical termination testers $\bigcirc$---and those
are required for that.

One should note that both are implementable: $\pifzkw$ by standard
dovetailing techniques, or more concretely by using threads, and
$\bigcirc_{> b} M$ by guessing and checking a derivation of
$[\_] \cdot M \vp b$, or more concretely by simulating all the
execution traces of $M$ and counting their probabilities.  The latter
can be done, concretely, by running $M$ under a hypervisor that forks
the process it emulates at each random binary choice $\oplus$: each
subprocess that terminates after having gone through $n$ random binary
choices contributes $1/2^n$ to a global counter, and the hypervisor
itself terminates when that counter exceeds $b$.

A few questions remain:

\begin{enumerate}
\item Is $\pifzkw$ definable in CBPV$(\Demon, \Nature)+\bigcirc$?  Is
  CBPV$(\Demon, \Nature)+\bigcirc$ fully abstract?   The results of
  Section~\ref{sec:need-parallel-if} fail to answer those questions.
\item We have defined languages with an $\abort_{\F{}\sigma}$ operator,
  and where computation types are interpreted as continuous lattices.
  Would bc-domains be enough, namely, can we do without an
  $\abort_{\F{}\sigma}$ operator and still obtain a full abstraction
  result?  Note that $\bigcirc_{> b}$ does not just estimate
  probabilities of termination, but also catches the exception raised
  by $\abort_{\F{}\sigma}$, hence serves more than one purpose.
\item Since the type $\U\F{} \tau \to \U\F{} \tau \to \F{} \tau$ is
  describable in CBPV$(\Demon,\Nature)+\bigcirc+\pifzkw$ for every
  value type $\tau$, the binary supremum map on $\Eval {\F{} \tau}$ is
  obtainable as a directed supremum of definable values.  Whereas
  $\owedge$ implements demonic non-determinism, binary suprema
  implement \emph{angelic} non-determinism.  Is binary supremum itself
  definable?
\end{enumerate}




\bibliographystyle{plain}

\newif\ifinlinebib
\inlinebibtrue
\ifinlinebib

\else
\bibliography{fa}
\fi







\end{document}
